\documentclass[11pt]{article} %
\usepackage[margin=1in]{geometry}

\usepackage{bm}
\usepackage{multirow}
\usepackage{comment}
\usepackage{csquotes}
\usepackage{booktabs}
\usepackage{tablefootnote}
\usepackage[colorlinks=true, linkcolor=blue, urlcolor=blue, citecolor=ForestGreen]{hyperref}
\usepackage{float}
\usepackage{framed}
\usepackage[linesnumbered,boxruled,vlined,noend]{algorithm2e}
\usepackage{verbatim}
\usepackage{amssymb,amsfonts,amsmath,amsthm}
\usepackage{amsmath}
\usepackage{amssymb}
\usepackage{amsthm} %
\usepackage{thmtools,thm-restate} %

\usepackage{enumerate}
\usepackage{footnote}
\usepackage[font=small,labelfont=bf]{caption}
\usepackage{color}
\usepackage{graphicx}
\usepackage{wrapfig}
\usepackage{tabulary}
\usepackage{diagbox}
\usepackage[dvipsnames]{xcolor}
\usepackage{pgfplots}
\usepackage{comment}
\usepackage{setspace}
\usepackage[capitalise]{cleveref} %
\usepackage{crossreftools}
\usepackage{xparse}
\usepackage[inline]{enumitem}
\newlist{inenum}{enumerate*}{1}
\setlist[inenum]{label=(\roman*)}
\newlist{compitem}{itemize}{3}
\setlist[compitem]{label=\textbullet,nosep}

\usepackage[draft,inline,marginclue]{fixme}
\FXRegisterAuthor{hf}{ahf}{\color{Green}Hendrik}

\SetKwProg{Fn}{Function}{}{}
\SetKwRepeat{Do}{do}{while}

\def\defeq{:=}

\NewDocumentCommand{\dmax}{}{%
	\Delta%
}

\newcommand{\costbound}{\ensuremath{\OPT'}}
\NewDocumentCommand{\cost}{o m O{k} o}{%
	\mathrm{cost}\IfNoValueF{#4}{^{#4}}_{#3}(#2\IfNoValueF{#1}{, #1})%
}
\NewDocumentCommand{\opt}{m O{k}}{%
	\mathrm{opt}_{#2}(#1)%
}

\def\adv/{\textsc{Adv}}
\def\alg/{\textsc{Alg}}
\NewDocumentCommand{\auxgraph}{m o o}{%
	#1%
	\IfNoValueF{#2}{_{#2%
			\IfNoValueF{#3}{,#3}%
	}}%
}
\NewDocumentCommand{\auxnum}{m o o}{%
	\lvert {#1}%
	\IfNoValueF{#2}{_{#2%
			\IfNoValueF{#3}{,#3}%
	}} \rvert%
}
\NewDocumentCommand{\kg}{o o}{%
	\auxgraph{G}[#1][#2]%
}
\NewDocumentCommand{\actno}{o o}{%
	\auxgraph{A}[#1][#2]%
}
\NewDocumentCommand{\actnum}{o o}{%
	\auxnum{A}[#1][#2]%
}
\NewDocumentCommand{\pasno}{o o}{%
	\auxgraph{P}[#1][#2]%
}
\NewDocumentCommand{\pasnum}{o o}{%
	\auxnum{P}[#1][#2]%
}
\NewDocumentCommand{\disno}{o o}{%
	\auxgraph{D}[#1][#2]%
}
\NewDocumentCommand{\disnum}{o o}{%
	\auxnum{D}[#1][#2]%
}
\NewDocumentCommand{\akg}{o o}{%
	\auxgraph{H}[#1][#2]%
}
\NewDocumentCommand{\s}{o}{%
	\sigma%
	\IfNoValueF{#1}{_{#1}}%
}
\NewDocumentCommand{\allpoints}{m}{%
	Q%
	\IfNoValueF{#1}{_{#1}}%
}
\NewDocumentCommand{\points}{m}{%
	P%
	\IfNoValueF{#1}{_{#1}}%
}

\NewDocumentCommand{\q}{o o}{%
	q%
	\IfNoValueF{#1}{_{#1%
			\IfNoValueF{#2}{,#2}%
	}}
}
\NewDocumentCommand{\V}{m}{%
	V(#1)%
}
\NewDocumentCommand{\E}{m}{%
	E(#1)%
}
\NewDocumentCommand{\ngh}{o m}{%
	\Gamma%
	\IfNoValueF{#1}{_{#1}}%
	(#2)%
}
\NewDocumentCommand{\dg}{o m}{%
	\operatorname{deg}%
	\IfNoValueF{#1}{_{#1}}%
	(#2)%
}
\NewDocumentCommand{\dsp}{o m m}{%
	d%
	\IfNoValueF{#1}{_{#1}}%
	(#2, #3)%
}
\NewDocumentCommand{\ans}{m}{%
	\textrm{ans}(#1)%
}

\declaretheoremstyle[
bodyfont=\normalfont\itshape,
]{normal_style}
\declaretheoremstyle[
headfont=\normalfont\itshape,
notefont=\normalfont\itshape,
qed=\qedsymbol
]{proof_style}

\declaretheorem[
name=Theorem,
style=normal_style]{theorem}
\declaretheorem[
name=Corollary,
numberlike=theorem,
style=normal_style]{corollary}
\declaretheorem[
name=Lemma,
numberlike=theorem,
style=normal_style]{lemma}
\declaretheorem[
name=Claim,
numberlike=theorem,
style=normal_style]{claim}
\declaretheorem[
name=Observation,
numberlike=theorem,
style=normal_style]{observation}
\declaretheorem[
name=Proposition,
numberlike=theorem,
style=normal_style]{proposition}
\declaretheorem[
name=Definition,
numberlike=theorem,
style=normal_style]{definition}
\makeatletter
\let\proof\@undefined
\let\endproof\@undefined
\makeatother
\declaretheorem[
name=Proof,
numbered=no,
style=proof_style]{proof}

\def\xth/{%
	\textsuperscript{th}%
}

\newcommand{\delete}{\texttt{Delete}}

\newcommand{\ins}{\ttx{Insert}}

\newcommand{\poly}{\operatorname{poly}}
\newcommand{\polylog}{\operatorname{polylog}}

\newcommand{\LFMIS}{\operatorname{LFMIS}} 
\newcommand{\elim}{\operatorname{elim}} 
\newcommand{\OPT}{\operatorname{OPT}} 
\newcommand{\ALG}{\mathcal{L}_{k+1}} 
\newcommand{\Run}{\operatorname{Time}}

\newcommand{\R}{\mathbb{R}}

\newcommand{\N}{\mathbb{N}}

\newcommand{\ex}[1]{\mathop{{\bf E}\left[ #1 \right]}}
\newcommand{\exx}[2]{\mathop{{\bf E}}_{#1}\left[ #2 \right]}

\newcommand{\pr}[1]{\operatorname{{\bf Pr}}\left[ #1 \right]}
\newcommand{\prb}[2]{\mathop{{\bf Pr}}_{#1}\left[ #2 \right]}

\newcommand{\aw}[1]{{\color{orange}#1}}

\def\rem#1{{\marginpar{\raggedright\scriptsize #1}}}
\newcommand{\awr}[1]{\rem{\textcolor{orange}{#1}}}

\renewcommand{\awr}[1]{}
\renewcommand{\aw}[1]{#1}

\newcommand{\bD}{\mathbf{D}}

\newcommand{\bX}{\mathbf{X}}

\newcommand{\bZ}{\mathbf{Z}}

\newcommand{\cA}{\mathcal{A}}
\newcommand{\cB}{\mathcal{B}}
\newcommand{\cC}{\mathcal{C}}
\newcommand{\cD}{\mathcal{D}}
\newcommand{\cE}{\mathcal{E}}
\newcommand{\cF}{\mathcal{F}}

\newcommand{\cH}{\mathcal{H}}
\newcommand{\cI}{\mathcal{I}}

\newcommand{\cL}{\mathcal{L}}

\newcommand{\cP}{\mathcal{P}}
\newcommand{\cQ}{\mathcal{Q}}

\newcommand{\cX}{\mathcal{X}}

\newcommand{\eps}{\epsilon} 
\newcommand{\ttx}[1]{\texttt{#1}}

\usepackage[framemethod=TikZ]{mdframed}
\newcounter{Frame}
\newenvironment{Frame}[1][h]{%
	\refstepcounter{Frame}
	\begin{mdframed}[%
		frametitle={#1},
		skipabove=\baselineskip plus 2pt minus 1pt,
		skipbelow=\baselineskip plus 2pt minus 1pt,
		linewidth=1.0pt,
		frametitlerule=true,
		]%
	}{%
	\end{mdframed}
}

\title{Optimal Fully Dynamic $k$-Center Clustering for Adaptive and Oblivious Adversaries} %
\author{
	MohammadHossein Bateni\\
	Google Research \\
	\texttt{bateni@google.com}
	\and
Hossein Esfandiari \\
Google Research \\
	\texttt{esfandiari@google.com}
	\and 
	Hendrik Fichtenberger\\
	Google Research	\\
	\texttt{fichtenberger@google.com}
	\and
	Monika Henzinger \\
	University of Vienna \\
	\texttt{monika.henzinger@univie.ac.at}
	\thanks{This project has received funding from the European Research Council (ERC) under the European Union's
		\begin{wrapfigure}{r}{2cm}%
			\includegraphics[width=2cm]{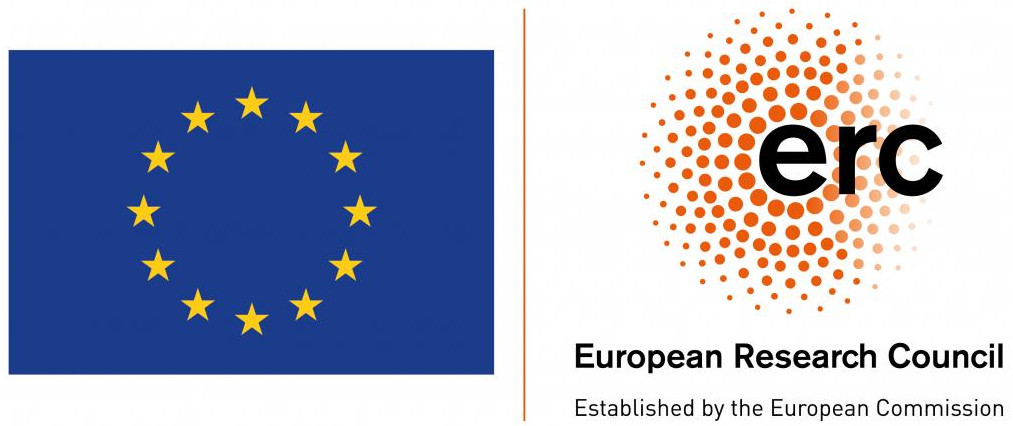}%
		\end{wrapfigure}%
		Horizon 2020 research and innovation programme (Grant agreement No. 101019564 ``The Design of Modern Fully Dynamic Data Structures (MoDynStruct)'' and from the Austrian Science Fund (FWF) project ``Fast Algorithms for a Reactive Network Layer (ReactNet)'', P~33775-N, with additional funding from the \textit{netidee SCIENCE Stiftung}, 2020--2024.}
	\and
	Rajesh Jayaram \\
	Google Research\\
	\texttt{rkjayaram@google.com} 
	\and
	Vahab Mirrokni\\
	Google Research\\
	\texttt{mirrokni@google.com}
	\and 
	Andreas Wiese \\
	Technical University of Munich\\
	\texttt{andreas.wiese@tum.de}
}

\date{}

\begin{document}
\maketitle

	\begin{abstract}
	In fully dynamic clustering problems, a clustering of a given data set in a metric space must be maintained while it is modified through insertions and deletions of individual points.
		In this paper, we resolve the complexity of fully dynamic $k$-center clustering against both adaptive and oblivious adversaries. Against oblivious adversaries, we present the first algorithm for fully dynamic $k$-center in an arbitrary metric space that maintains an optimal $(2+\epsilon)$-approximation in $O(k \cdot  \mathrm{polylog}(n,\Delta))$ amortized update time. Here, $n$ is an upper bound on the number of active points at any time, and $\Delta$ is the aspect ratio of the metric space. Previously, the best known amortized update time was $O(k^2\cdot \mathrm{polylog}(n,\Delta))$, and is due to Chan, Gourqin, and Sozio (2018). Moreover, we demonstrate that our runtime is optimal up to $\mathrm{polylog}(n,\Delta)$ factors. In fact, we prove that even offline algorithms for $k$-clustering tasks in arbitrary metric spaces, including $k$-medians, $k$-means,  and $k$-center, must make at least $\Omega(n k)$ distance queries to achieve any non-trivial approximation factor. This implies a lower bound of $\Omega(k)$ which holds even for the insertions-only setting.
		
		For adaptive adversaries, we give the first deterministic algorithm for fully dynamic $k$-center which achieves a $O\left( \min\left\{\frac{\log(n/k)}{\log \log n},k\right\}\right)$ approximation in $O(k \cdot  \mathrm{polylog}(n,\Delta))$ amortized update time. Further, we demonstrate that any algorithm which achieves a $O\left( \min\left\{\frac{\log n}{ k \log f(k,2n)},k\right\}\right)$-approximation against adaptive adversaries requires $f(k,n)$ update time, for any arbitrary function $f$. Thus, in the regime where $k = O(\sqrt{\frac{\log n }{\log \log n}})$, we close the complexity of the problem up to $\mathrm{polylog}(n, \Delta)$ factors in the update time. Our lower bound extends to other $k$-clustering tasks in arbitrary metric spaces, including $k$-medians and $k$-means.

		Finally, despite the aforementioned lower bounds, we demonstrate that an update time sublinear in $k$ is possible against oblivious adversaries for metric spaces which admit locally sensitive hash functions (LSH), resulting in improved algorithms for a large class of metrics including Euclidean space, $\ell_p$-spaces, the Hamming Metric, and the Jaccard Metric. We also give the first fully dynamic $O(1)$-approximation algorithms for the closely related $k$-sum-of-radii and $k$-sum-of-diameter problems, with $O(\mathrm{poly}(k,\log \Delta))$ update time. 
		
	\end{abstract}
	
\thispagestyle{empty}
\newpage
\tableofcontents
\thispagestyle{empty}
\newpage

\pagenumbering{arabic}
\setcounter{page}{1}

\section{Introduction}

Clustering is a fundamental and well-studied problem in computer science,
which arises in approximation algorithms, unsupervised learning, computational
geometry, classification, community detection, image segmentation,
databases, and other areas \cite{hansen1997cluster,schaeffer2007graph,fortunato2010community,shi2000normalized,arthur2006k,tan2013data,coates2012learning}.
The goal of clustering is to find a structure in data by grouping
together similar data points. Clustering algorithms optimize a given
objective function which characterizes the quality of a clustering.
One of the classical and best studied clustering objectives is the
$k$-center objective.

Specifically, given a metric space $(\cX,d)$ and a set of points
$P\subseteq\cX$, the goal of $k$-center clustering is to output
a set $C\subset\cX$ of at most $k$ ``centers'', so that the
maximum distance of any point $p\in P$ to the nearest center $c\in C$
is minimized. In other words, the goal is to minimize the objective
function $\max_{p\in P}d(p,C)$, where $d(p,C)=\min_{c\in C}d(p,c)$. The $k$-center clustering problem
admits several well-known greedy $2$-approximation algorithms \cite{gonzalez1985clustering,hochbaum1986unified}.
However, it is known to be NP-hard to approximate the objective to
within a factor of $(2-\epsilon)$ for any constant $\eps>0$~\cite{hsu1979easy}.
Moreover, even restricted to Euclidean space, it is still NP-hard
to approximate beyond a factor of $1.822$ \cite{feder1988optimal,bern1997approximation}.

While the approximability of many clustering tasks, including $k$-center
clustering, is fairly well understood in the static setting, the same
is not true for \textit{dynamic datasets}. Recently, due to the proliferation
of data and the rise of modern computational paradigms where data
is constantly changing, there has been significant interest in developing
dynamic clustering algorithms~\cite{cohen2016diameter,lattanzi2017consistent,ChaFul18,DBLP:conf/esa/GoranciHL18,schmidt2019fully,goranci2019fully,HenzingerK20,henzinger2020dynamic,fichtenberger2021consistent}.
In the incremental dynamic setting, the dataset $P$ is observed via
a sequence of insertions of data points, and the goal is to maintain
a good $k$-center clustering of the current set of active points.
In the \textit{fully dynamic} setting, points can be both inserted
and deleted from $P$.

The study of dynamic algorithms for $k$-center was initated by Charikar,
Chekuri, Feder, and Motwani~\cite{charikar2004incremental}, whose
``doubling algorithm'' maintains an $8$-approximation in amortized
$O(k)$ update time. However, the doubling algorithm is unable to
handle deletions of data points. It was not until recently that the
first fully dynamic algorithm for $k$-center, with update time better
than naively recomputing a solution from scratch, was developed. In
particular, the work of Chan, Guerqin, and Sozio \cite{ChaFul18}
proposed a randomized algorithm that maintains an optimal $(2+\eps)$-approximation
in $O(\frac{\log\Delta}{\eps}k^{2})$ amortized time per update, where
$\Delta$ is the aspect ratio of the dataset. The algorithm is randomized
against an \emph{oblivious adversary}, i.e., the adversary has to
fix the input sequence in advance and cannot adapt it based on the
decisions of the algorithm. 

Since then, algorithms with improved running time have been demonstrated
for the special cases of Euclidean space \cite{schmidt2019fully}
(albeit, with a larger approximation factor), and for spaces with
bounded doubling dimension \cite{goranci2019fully}. However, despite
this progress, for general metrics the best known dynamic algorithm
for $k$-center is still the algorithm by Chan et al.~\cite{ChaFul18}
and it is open to find a dynamic algorithm with an update time that
is sub-quadratic in $k$ with any non-trivial approximation guarantee. Furthermore, it is open to find a dynamic algorithm for $k$-center
with any non-trivial approximation guarantee and update time against
an \emph{adaptive adversary}, i.e., an adversary that can choose the
next update depending on the previous decisions of the algorithm.
Note that deterministic algorithms are necessarily robust against adaptive adversaries, however, most of the aforementioned algorithms are randomized. This raises the following question.

\begin{center}
\emph{What are the best possible update times and approximation
ratios of\\ dynamic algorithms for $k$-center against oblivious and
adaptive adversaries?}
\end{center}

In addition to $k$-center, there are other popular clustering objectives whose complexity in the fully dynamic model is not fully understood. For example, in
the well-studied $k$-median and $k$-means problems, we minimize
$\sum_{p\in P}d(p,C)$ and $\sum_{p\in P}(d(p,C))^{2}$, respectively.  The latter problems are special cases of the more general \emph{$(k,z)$-clustering problem}, where one minimizes $\sum_{p\in S}d(p,C)^{z}$.  There are $O(1)$-approximation
algorithms known for these problems in the offline setting~\cite{gonzalez1985clustering,hochbaum1985best,hochbaum1986unified,ahmadian2016approximation}
as well as in the dynamic setting~\cite{HenzingerK20}.
However, as is the case for $k$-center, these dynamic algorithms are
randomized, and thus far no deterministic dynamic $O(1)$-approximation algorithms for these objectives
are known.

In addition to $k$-means and median, a natural variant of the $k$-center is the \emph{$k$-sum-of-radii} problem
(also known as the \emph{$k$-cover problem}) where we choose $k$ centers $C = \{c_1,\dots,c_k\}$, and seek to minimize
their sum $\sum_{i}r_{i}$ where $r_i$ is the maximum distance $d(p,c_i)$ over all points point $p$ assigned to $c_i$ (i.e., the radius of the clustered centered at $c_i$). Related to this is the \emph{$k$-sum-of-diameters}
problem, where one minimizes the sum of the diameters
of the clusters.  For each of these closely related variants to $k$-center, no algorithms with non-trivial guarantees were known to exist. 

\subsection{Our Contributions}
Our main contribution is the resolution of the complexity of fully dynamic $k$-center
clustering, up to polylogarithmic factors in the update time, against oblivious adversaries, thereby answering the prior open question, as well as the resolution of the problem against adaptive adversaries for the regime when $k = O(\sqrt{\log n/ \log \log (n+\Delta)})$. A summary of our upper and lower bounds for fully dynamic $k$-center is given in Table~\ref{tab:k-center-main}.

{
\def\arraystretch{1.1}
\begin{table}
\centering
\begin{tabular}{ccccc}
\toprule
$k$-center & \multicolumn{2}{c}{Upper bound} & \multicolumn{2}{c}{Lower bound}\tabularnewline \cmidrule(r){2-3} \cmidrule(l){4-5}
Adversary & update time & approx. ratio & update time & approx. ratio\tabularnewline \midrule
Oblivious & \textbf{\textcolor{black}{$\bm{\tilde{O}(k/\eps)}$}} & \textbf{\textcolor{black}{$\bm{2+\eps}$}} & \textbf{\textcolor{black}{$\bm{\Omega(k)}$}} & \bf{any}\tabularnewline
 & $\tilde{O}(k^{2}/\eps)$ & $2+\eps$~\cite{ChaFul18} &  & \tabularnewline
Adaptive & $\bm{\tilde{O}(k)}$ & $\bm{O\left(\min\{\frac{\log(n/k)}{\log \log n},k\}\right)}$ & $\bm{f(k,n)}$ & $\bm{\Omega\left(\min\{\frac{\log(n)}{k \log f(k,2n)},k\}\right)}$\tabularnewline
\bottomrule
\end{tabular}%
\def\arraystretch{1}

\caption{\label{tab:k-center-main}Our main new upper and lower bounds for $k$-center
where $f(k,n)$ is an arbitrary function. The notation $\tilde{O}$ hides factors that are polylogarithmic in $n,\Delta$.
}
\end{table}
}

\paragraph{Oblivous adversaries.}
We first design an algorithm with update time that is \textit{linear} in $k$, 
and only logarithmic in $n$ and $\Delta$, where $n$ is an upper
bound on the number of active points at a given point in time,\footnote{We remark that our algorithm does not need to know $n$ or the total number of updates in advance.} and
$\Delta$ denotes the aspect ratio of the given metric space. Our algorithm achieves an
approximation ratio of $2+\epsilon$, which is essentially tight since
it is NP-hard to obtain an approximation ratio of $2-\epsilon$. This improves on the prior best known quadratic-in-$k$ update time of \cite{ChaFul18}, for the same approximation ratio. 
\begin{theorem}
\label{thm:main} There is a randomized fully dynamic algorithm that,
on a sequence of insertions and deletions of points from an arbitrary
metric space, maintains a $(2+\eps)$-approximation to the optimal
$k$-center clustering. The amortized update time of the algorithm
is $O(\frac{\log\Delta\log n}{\eps}(k+\log n))$ in expectation, and
$O(\frac{\log\Delta\log n}{\eps}(k+\log n)\log\delta^{-1})$ with
probability $1-\delta$ for any $\delta\in(0,\frac{1}{2})$. 
\end{theorem}
 
We demonstrate that our update time is tight up to logarithmic factors,
even in the insertion only setting. In fact, our lower bound extends
immediately to all the other clustering problems defined above, see
Table~\ref{tab:other-LB} for a further overview of our lower bounds beyond those given in Table~\ref{tab:k-center-main}.
\begin{theorem}
\label{thm:LB}If a dynamic algorithm for $k$-center, $k$-median,
$k$-means, $(k,z)$-clustering for any $z>0$, $k$-sum-of-radii,
or $k$-sum-of-diameter has an approximation ratio that is bounded
by any function in $k$ and $n$, then it has an update time of $\Omega(k)$.
This holds already if points can only be inserted but not deleted.
\end{theorem}
We remark that our lower bound is unconditional, i.e., it does not depend
on any complexity theoretic assumptions. In fact, we prove an even
stronger statement: any \textit{offline} algorithm must 
query the distances of $\Omega(nk)$ pairs of points in order to guarantee
any bounded approximation ratio with constant probability (see \cref{sec:LB}). 
The amortized per-update time
of $\Omega(k)$ for dynamic insertion-only algorithms then follows immediately. 

{
\def\arraystretch{1.2}
\begin{table}
\centering
\begin{tabular}{ccccc}
\toprule
\multicolumn{1}{c}{} & \multicolumn{2}{c}{Oblivious adversary} & \multicolumn{2}{c}{Adaptive adversary}\tabularnewline \cmidrule(r){2-3} \cmidrule(l){4-5}
\multicolumn{1}{c}{} & update time & approx. ratio & update time & approx. ratio\tabularnewline
\midrule
$k$-median & \textbf{\textcolor{black}{$\bm{\Omega(k)}$}} & \bf{any} & $\bm{f(k,n)}$ & $\bm{\Omega\left(\frac{\log(n)}{\log f(1,2n)}\right)}$\tabularnewline
 & $\tilde{O}(k^{2}/\eps^{O(1)})$ & $5.3+\eps$~\cite{HenzingerK20} &  & \tabularnewline
$k$-means & \textbf{\textcolor{black}{$\bm{\Omega(k)}$}} & \bf{any} & $\bm{f(k,n)}$ & $\bm{\Omega\left(\left(\frac{\log(n)}{\log f(1,2n)}\right)^{2}\right)}$\tabularnewline
 & $\tilde{O}(k^{2}/\eps^{O(1)})$ & $36+\eps$~\cite{HenzingerK20} &  & \tabularnewline
$(k,z)$-clustering & \textbf{\textcolor{black}{$\bm{\Omega(k)}$}} & \bf{any} & $\bm{f(k,n)}$ & \textbf{\textcolor{black}{$\bm{\Omega\left(\left(\frac{\log(n)}{z+\log f(1,2n)}\right)^{z}\right)}$}}\tabularnewline
$k$-sum-of-radii & \textbf{\textcolor{black}{$\bm{\Omega(k)}$}} & \bf{any} & $\bm{f(k,n)}$ & $\bm{\Omega\left(\frac{\log(n)}{\log f(1,2n)}\right)}^\star$\tabularnewline
 & \textcolor{black}{$\bm{k^{O(1)}\log\Delta}\Delta$} & $\bm{O(1)}$ &  & \tabularnewline
$k$-sum-of-diam. & \textbf{\textcolor{black}{$\bm{\Omega(k)}$}} & \bf{any} & $\bm{f(k,n)}$ & $\bm{\Omega\left(\frac{\log(n)}{\log f(1,2n)}\right)}^\star$\tabularnewline
 & \textcolor{black}{$\bm{k^{O(1)}\log \Delta}\Delta$} & $\bm{O(1)}$ &  & \tabularnewline
\bottomrule
\end{tabular}

\caption{\label{tab:other-LB}Our lower bounds for $k$-median, $k$-means,
$k$-sum-of-radii, and $k$-sum-of-diameter and our upper bounds for
the latter two problems, where $f(k,n)$ is an arbitrary function. For the ratios marked with $\star$, we assume that the algorithm also outputs an estimate of the cost of an optimal solution.}
\end{table}
\def\arraystretch{1}
}

\paragraph{Deterministic algorithms and adaptive adversaries.}
For arbitrary metrics, our algorithms access the metric via queries to a distance oracle returning $d(x,y)$ between any two points which were inserted (and possibly deleted) up to that point. In the oblivious model, both the underlying metric and point insertions and deletions are fixed in advance. However, in the adaptive adversary model, the adversary can adaptively choose the values of $d(x,y)$ based on past outputs of the algorithm. Thus, the values of $d(x,y)$ are fixed only as they are queried (subject to obeying the triangle inequality). We refer to this as the \textit{metric-adaptive} model (see discussion below).
 
Given our bounds against oblivious adversaries, it is natural to ask whether constant factor approximations for $k$-center are also achievable in $\tilde{O}(k)$, or even  $\tilde{O}(\mathrm{poly}(k))$, time against metric-adaptive adversaries.
Perhaps surprisingly, we show that any $k$-center algorithm with update time poly-logarithmic in $n + \Delta$ (and arbitrary dependency on $k$) must incur an essentially logarithmic approximation factor if it needs to report an estimation of the cost. Moreover, for a large range of $k$, such a runtime is not possible even if the algorithm is tasked only with returning an approximately optimal set of $k$-centers. 
\begin{theorem}[{restate=[name=]thmDetLB}]
\label{thm:det-1lb} For any $k  \geq 1$, any dynamic algorithm which returns a set of $k$-centers against a metric-adaptive adversary
with an amortized update time of $f(k,n)$, for
an arbitrary function $f$,  must have an approximation ratio of $\Omega\left(\min\{ \frac{\log(n)}{k \log f(k,2n)}, k\}\right)$ for the $k$-center problem. In addition, even for the case of $k=1$, we show that \emph{any} algorithm
with an update time of $f(k,n)$

\begin{itemize}
\item for $1$-median
has an approximation ratio of $\Omega\left(\frac{\log(n)}{\log f(1,2n)}\right)$, 
\item for $1$-means has an approximation ratio of $\Omega\left(\left(\frac{\log(n)}{\log f(1,2n)}\right)^{2}\right)$, 
\item for $(1,z)$-clustering has an approximation ratio of $\Omega\left(\left(\frac{\log(n)}{z+\log f(1,2n)}\right)^{z}\right)$,
\item for $1$-center, $1$-sum-of-radii or $1$-sum-of-diameters has an approximation ratio of
$\Omega\left(\frac{\log(n)}{\log f(1,2n)}\right)$ if the algorithm is also able to estimate the cost of the optimal clustering.
\end{itemize}
\end{theorem}
In particular, Theorem \ref{thm:det-1lb} demonstrates that for any $k$ satisfying $\omega(1) \leq k \leq o(\frac{\log n}{\log \log (n+\Delta)})$, there is no constant-factor approximation algorithm for $k$-center running in $\polylog(n,\Delta)$ time which is correct against an adaptive adversary. 
In fact, we prove even more fine-grained trade-offs between the update
time of an algorithm and the possible approximation ratio (see \Cref{sec:LBadap} for details).

\paragraph{Separation between Adversarial Models.} An important consequence of Theorem \ref{thm:det-1lb} is a separation in the complexity of $k$-clustering tasks between two distinct types of adaptive adversaries. Specifically, Theorem \ref{thm:det-1lb} is a lower bound for adversaries which decide on the answers to distance queries on the fly, without having fixed the metric beforehand, subject to the constraint that the answers are always consistent with some metric. In other words, the metric itself, in addition to the points which are inserted and deleted, are adaptively chosen (i.e., they are metric-adaptive adversaries). This is as opposed to the setting where the metric is fixed in advance, and only the insertions and deletions of points from that space are adaptively chosen (call these \textit{point-adaptive} adversaries).

 To illustrate the difference, suppose the current active point set is $P$. When a point-adaptive adversary inserts a new point $q$, since the metric is fixed, it must irrevocably decide on the value of $d(p,q)$ for all $p \in P$. On the other hand, a metric-adaptive adversary can defer fixing these value until it is queried for them. For example, suppose $P=\{a,b\}$ and then a point $c$ is inserted, after which the algorithm queries the adversary for the value of $d(a,c)$, and then makes a (possibly randomized) decision $\cD$ based on $d(a,c)$ (e.g. $\cD$ could be whether to make $c$ a center). Then a metric-adaptive adversary can adaptively decide on the value of $d(b,c)$ based on the algorithm's decision $\cD$ (subject to not violating the triangle inequality), whereas a point-adaptive algorithm must have fixed $d(b,c)$ independent of $\cD$.
 
This distinction is nuanced, but has non-trivial consequences. 
 Namely, while Theorem \ref{thm:det-1lb} rules out fully dynamic $O(1)$-approximation algorithms with $\polylog(n,\Delta)$ update time for $k$-means and $k$-medians for metric-adaptive adversaries, in \cite{HenzingerK20} the authors design such algorithms against point-adaptive adversaries. Thus, Theorem \ref{thm:det-1lb} demonstrates that the algorithms of \cite{HenzingerK20} would not have been possible against metric-adaptive adversaries.
To the best of our knowledge, this is the first separation between such adversarial models for dynamic clustering.

Now while the \textit{randomized} algorithms of \cite{HenzingerK20} apply only to the weaker point-adaptive adversaries, note that deterministic algorithms are necessarily correct even against the stronger metric-adaptive adversaries. We show that such a deterministic algorithm in fact exists, with complexity matching the lower bound of \Cref{thm:det-1lb} for $k = O(\sqrt{\log n / \log \log n})$.
\begin{theorem}
\label{thm:mpt-det-upper} Fix any $B \geq  2$ and $\eps  \in (0,1)$. Then there is a deterministic dynamic algorithm
for $k$-center with an amortized update time of $O\left(\frac{kB \log n\log\Delta }{\eps}\right)$
and an approximation factor of $(4+\eps)\min\left\{\frac{\log(n/k)}{\log B},k\right\}$. 
 Furthermore, the worst-cast insertion time is $O\left(\frac{k B \log n\log\Delta}{\eps}\right)$, and the worst-case deletion time is $O\left(\frac{k^2 B \log n\log\Delta}{\eps}\right)$.
\end{theorem}
In particular, by setting $B = \log n$, we obtain a fully dynamic and deterministic algorithm for $k$-centers with an approximation 
ratio of $O\left(\min\left\{\frac{\log(n/k)}{\log \log n},k\right\}\right)$, which, by \cref{thm:det-1lb}, is optimal for the regime when $k = O(\sqrt{\frac{\log n }{\log \log n}})$
among algorithms that are robust against metric-adaptive adversaries and whose
update time is polylogarithmic in $n$
and~$\Delta$. Moreover, that by setting $B = n^{\eps}$ for any $\eps  \in (0,1)$ we obtain a $O(1/\eps)$-approximation in time $\tilde{O}(k n^\eps )$. Thus, \cref{thm:mpt-det-upper} gives a deterministic $O(1)$-approximate algorithm for fully dynamic $k$-centers running in time $\tilde{O}(k n^\eps )$ for any constant $\eps > 0$. 

\paragraph{Improved Fully Dynamic $k$-center via Locally Sensitive Hashing.}

The lower bound of Theorem \ref{thm:LB} demonstrates that, even against an oblivious adversary, one cannot beat $\Omega(k)$ amortized update time for fully dynamic
$k$-center. However,
in \cite{schmidt2019fully,goranci2019fully} it was shown that for
the case of Euclidean space or metrics with bounded doubling dimension,
update times \textit{sublinear }in $k$ are possible. These results rely on
nearest neighbor data structures with running times that are sublinear
in $k$, which are designed specifically for the respective metric
spaces, along with specialized clustering algorithms to employ them.
Note that for the case of Euclidean space, the resulting approximation
factors were still logarithmic. 

We significantly generalize and strengthen the above results, by demonstrating that \textit{any} metric space admitting sublinear
time nearest neighbor search data structures also admits fully dynamic
$k$-center algorithms whose update time is sublinear in~$k$.
We do this via a black-box reduction, showing that we can obtain a fully dynamic algorithm for $k$-center for any metric space, given a \textit{locally sensitive hash function}
(LSH) for that space.

Informally, an LSH for a space $(\cX,d)$ is a hash function $h$ mapping $\cX$ to some number of hash buckets, such that closer points in the metric are \textit{more} likely to collide than far points. Thus, after hashing a subset $S \subset \cX$ into the hash table, given a query point $q \in \cX$, to find the closest points to $q$ in $S$, it (roughly) suffices to search only through the hash bucket $h(q)$. The ``quality'' of an LSH family $\cH$ is parameterized by four values $(r,cr,p_1,p_2)$, meaning that for any $x,y \in \cX$:
 \begin{itemize}
 	\item If $d(x,y) \leq r$, then $\prb{h \sim \cH}{h(x) = h(y)} \geq p_1$. 
 	\item If $d(x,y) > cr$, then $\prb{h \sim \cH}{h(x) = h(y)} \leq p_2$. 
 \end{itemize}
Given the above definition, we can now informally state our main result for LSH-spaces (see Section \ref{sec:LSH} for formal statements). 

\renewcommand{\arraystretch}{1.3}
\begin{table}[t]
	\centering
	\begin{tabular}{ccccc}
		\toprule		
		Metric space & \textbf{Our approx.} &\textbf{Our runtime} & Prior approx. &  Prior runtime \\ \midrule

		Arbitrary metric space& $\bm{2+\eps} $& $\bm{\tilde{O}(  k)} $& $2+\eps$ &  $ \tilde{O}( k^2) $ \cite{ChaFul18}\\ %
		$(\R^d,\ell_p)$, $p \in [1,2]$ & $\bm{c (4+\eps) }$ &$\bm{\tilde{O}(  n^{1/c})} $& $O(c \cdot \log n)$ &  $\tilde{O}( n^{1/c}) $ \cite{schmidt2019fully}\\ %
		Eucledian space $(\R^d,\ell_2)$  & $\bm{c (\sqrt{8}+\eps)  }$ & $\bm{\tilde{O}(  n^{1/c^2 + o(1)}) }$& $O(c \cdot \log n)$ & $ \tilde{O}( n^{1/c}) $ \cite{schmidt2019fully}\\ %
		Hamming metric &  $\bm{c (4+\eps)} $ &\bm{$\tilde{O}(  n^{1/c})} $& -- &  -- \\ %
		Jaccard metric &  $\bm{c (4+\eps)} $ &\bm{$\tilde{O}(  n^{1/c})} $&  -- &  -- \\ %
		\begin{tabular}{c}
			EMD over $[D]^d$   \\
			$d= O(1)$
		\end{tabular} &  $\bm{O(c \cdot \log D)}  $ &$\bm{\tilde{O}(  n^{1/c})} $& -- &  -- \\ %
		\begin{tabular}{c}
			EMD over $[D]^d$   \\
			sparsity $s$
		\end{tabular} &  $\bm{O(c \cdot \log s n \log d )} $ &$\bm{\tilde{O}(  n^{1/c})} $& -- &  -- \\ %
		\bottomrule
	\end{tabular}
\caption{Our upper bounds for Fully Dynamic $k$-Centers against oblivious adversaries in different metric spaces. Results for specific metric spaces are obtained by applying known LSH families with Theorem \ref{thm:lshinformal}. } \label{table:LSHResults}
\end{table}

\begin{theorem}[informal]\label{thm:lshinformal}
 Let $(\cX,d)$ be a metric space that admits an LSH $\cH$ with
parameters $(r,cr,p_{1},p_{2})$ for every $r \geq 0$, and a running time of $\Run(\cH)$.
Then there is a fully dynamic algorithm for $k$-center on $(\cX,d)$
with an approximation ratio of $c(2+\eps)$ and an update time of
$O\left(\frac{\log\Delta}{\eps p_{1}}n^{2\rho}\cdot\Run(\cH)\right)$, where $\rho = \ln p_1^{-1}/ \ln p_2^{-1}$.
\end{theorem}
Using known LSH hash functions from the literature, Theorem \ref{thm:lshinformal} immediately yields improved state of the art algorithms for Euclidean spaces,
the Hamming metric, the Jaccard Metric, and the Earth Mover Distance
(EMD). In particular, our bounds significantly improve the prior best known results of \cite{schmidt2019fully} for Euclidean space, by at least a logarithmic factor in the approximation \aw{ratio}. 
See Table \ref{table:LSHResults} for an summary of these results, and see Section \ref{sec:LSHcorollaries} for the formal theorem statements for each metric space.

\paragraph{$k$-sum-of-radii and $k$-sum-of-diameter. }

Finally, we study the $k$-sum-of-radii and $k$-sum-of-diameter problems, for which there were previously no fully dynamic $O(1)$-approximation algorithms known with non-trivial update time. We design the first such algorithms, which hold against oblivious adversaries. Note that, as a consequence of Theorem \ref{thm:det-1lb}, such a constant-factor approximation is only possible against an oblivious adversary.

\begin{theorem}[{restate=[name=]thmPdMain}]
\label{thm:pd-main}
There are randomized dynamic algorithms for the
$k$-sum-of-radii and the $k$-sum-of-diameters problems with update
time $k^{O(1/\epsilon)}\log\Delta$ and with approximation ratios
of $13.008+\epsilon$ and $26.016+\epsilon$, respectively, against
an oblivious adversary. 
\end{theorem}
Thus, we complete the picture that all clustering problems defined above admit dynamic $O(1)$-approximation
algorithms against an oblivious adversary, but only (poly-)logarithmic
approximation ratios against an adaptive adversary.

\subsection{Technical Overview}\label{sec:tech}
We now describe the main technical steps employed in the primary results of the paper. 
For the remainder of the section, we fix a metric space $(\cX,d)$. Let $r_{\min}$ and $r_{\max}$ be values such that
$r_{\min} \leq d(x,y) \leq r_{\max}$ for  any $x,y \in \cX$. We set $\Delta = r_{\max}/r_{\min}$ as an upper bound on the aspect ratio of $(\cX,d)$. For any fixed point in time, we denote the current input points of a dynamic clustering algorithm by $P$, and write $n$ to denote the maximum size of $P$ during the dynamic stream.

\subsubsection{Algorithm for General Metric Spaces} Our starting point is the well-known reduction of Hochbaum and Shmoys \cite{hochbaum1986unified} from approximating $k$-center to computing a maximal independent set (MIS) in a collection of \textit{threshold} graphs. Formally, given a real $r>0$, the $r$-threshold graph of a point set $P$ is the graph $G_r = (V,E_r)$ with vertex set $V = P$, and where  $(x,y) \in E_r$ is an edge if and only if $d(x,y) \leq r$. One computes an MIS $\cI_r$ in the graph $G_r$ for each $r = (1+\eps)^i r_{\min}$ with $i=0,1,\dots,\lceil \log_{1+\eps} \Delta \rceil$. If $|\cI_r| \leq k$, then $\cI_r$ is a $k$-center solution of cost at most $r$. If $|\cI_r| > k+1$, then there are $k+1$ points whose pair-wise distance is at least $r$. Therefore, by the triangle inequality, the optimal cost is at least $r/2$. These facts together yield a $(2+\eps)$-approximation.

By the above, it suffices to maintain an MIS in $O(\eps^{-1} \log \Delta)$ threshold graphs. Now the problem of maintaining an MIS in a fully dynamic sequence of \textit{edge} insertions and deletions to a graph is very well studied \cite{assadi2019fully,gupta2018simple, onak2018fully,du2018improved, censor2016optimal, chechik2019fully, behnezhad2019fully}. Notably, this line of work has culminated with the algorithms of \cite{chechik2019fully,behnezhad2019fully}, which maintain an MIS in expected $\polylog n$ update time per edge insertion or deletion. Unfortunately, point insertions and deletions from a metric space correspond to \textit{vertex} insertions and deletions in a threshold graph. Since a single vertex update can change up to $O(n)$ edges in the graph at once, one cannot simply apply the prior algorithms for fully dynamic edge updates. Moreover, notice that in this vertex-update model, we are only given access to the graph via queries to the adjacency matrix. Thus, even finding a single neighbor of $v$ can be expensive.

On the other hand, observe that in the above reduction to MIS, one does not always need to compute the entire MIS; for a given threshold graph $G_r$, the algorithm can stop as soon as it obtains an independent set of size at least $k+1$. This motivates the following problem, which is to return either
an MIS of size at most $k$, or an independent set of size at least $k+1$. We refer to this as the $k$-Bounded MIS problem. Notice that given an MIS $\cI$ of size at most $k$ in a graph $G$, and given a new vertex $v$, if $v$ is not adjacent to any $u \in \cI$, then $\cI \cup \{v\}$ is an MIS, otherwise $\cI$ is still maximal. Thus, while an insertion of a vertex $v$ can add $\Omega(n)$ edges to $G$, for the $k$-Bounded MIS problem, one only needs to check the $O(k)$ potential edges between $v$ and $\cI$ to determine if $\cI$ is still maximal. Thus, our goal will be to design a fully dynamic algorithm for $k$-Bounded MIS with $\tilde{O}(k)$ amortized update time in the vertex-update model.

\paragraph{The Algorithm for $k$-Bounded MIS.}
To accomplish the above goal, we will adapt several of the technical tools employed by the algorithms for fully dynamic MIS in the edge-update model. Specifically, one of the main insights of this line of work is to maintain the \textit{Lexicographically First Maximal Independent Set} (LFMIS) with respect to a random permutation $\pi: V \to [0,1]$ of the vertices.\footnote{LFMIS with respects to random orderings were considered in \cite{censor2016optimal,assadi2019fully,chechik2019fully,behnezhad2019fully}.} The LFMIS is a natural object obtained by greedily adding the vertex with smallest $\pi(v)$ to the MIS, removing it and all its neighbors, and continuing iteratively until no vertices remain. Maintaining an LFMIS under a random ranking has several advantages from the perspective of dynamic algorithms. Firstly, it is \textit{history-independent}, namely, once $\pi$ is fixed, the current LFMIS depends only on the current graph, and not the order of insertions and deletions which led to that graph. Secondly, given a new vertex $v$, the probability that adding $v$ to the graph causes a large number of changes to be made to the LFMIS is small, since $\pi(v)$ must have been similarly small for this to occur.

Given the above advantages of an LFMIS, we will attempt to maintain the set $\LFMIS_{k+1}$ consisting of the first $\min\{k+1,|\LFMIS|\}$ vertices in the overall LFMIS with respect to a random ranking $\pi$; we refer to $\LFMIS_{k+1}$ as the top-$k$ LFMIS (see Definition \ref{def:topkLFMIS}). Notice that maintaining this set is sufficient to solve the $k$-Bounded MIS problem. The challenge in maintaining the set $\LFMIS_{k+1}$ will be to handle the ``excess'' vertices which are contained in the $\LFMIS$ but are not in $\LFMIS_{k+1}$, so that their membership in $\LFMIS_{k+1}$ can later be quickly determined when vertices with smaller rank in $\LFMIS_{k+1}$ are removed.
To handle these excess vertices, we make judicious use of a priority queue $\cQ$, with vertex priorities given by the ranking $\pi$.

Now the key difficulty in dynamically maintaining an MIS is that when a vertex $v$ in an MIS is deleted, potentially all of the neighbors of $v$ may need to be added to the MIS, resulting in a large update time. Firstly, in order to keep track of which vertices could possibly enter the LFMIS when a vertex is removed from it, we maintain a mapping $\ell:V \to \LFMIS$, such that for each $u \notin \LFMIS$, we have $\ell(u) \in \LFMIS$ and $(u,\ell(u))$ is an edge. The ``leader'' $\ell(u)$ of $u$ serves as a certificate that $u$ cannot be added to the MIS. When a vertex $v \in \LFMIS$ is removed from the LFMIS, we only need to search through the set $\cF_v = \{u \in V \;| \; \ell(u) = v\}$ to see which vertices should be added to the LFMIS. Note that this can occur when $v$ is deleted, or when a neighbor of $v$ with smaller rank is added to the LFMIS. Consequentially, the update time of the algorithm is a function of the number points $u$ whose leader $\ell(u)$ changes on that step. For each such $u$, we can check in $O(k)$ time if it should be added to $\LFMIS_{k+1}$ by querying the edges between $u$ and the vertices in $\LFMIS_{k+1}$. By a careful amortized analysis, we can prove that the total runtime of this algorithm is indeed at most an $O(k)$ factor larger than the total number of leader changes. This leaves the primary challenge of designing and maintaining a leader mapping which changes infrequently.

A natural choice for such a leader function is to set $\ell(u)$ to be the \textit{eliminator} of $v$ in the LFMIS. Here, for any vertex $u$ not in the LFMIS, the eliminator $\elim_\pi(u)$ of $u$ is defined to be its neighbor with lowest rank  that belongs to the LFMIS. 
The eliminators have the desirable property that they are also history-independent, and therefore the number of changes to the eliminators on a given update depends only on the current graph and the update being made. An important key result of \cite{behnezhad2019fully} is that the expected number of changes to the eliminators of the graph, even after the insertion or removal of an entire vertex, is at most $O(\log n)$. Therefore, if we could maintain the mapping $\ell(v) = \elim_{\pi}(v)$ by keeping track of the eliminators, our task would be complete.

Unfortunately, keeping track of the eliminators will not be possible in the vertex-update model, since we can only query a small fraction of the adjacency matrix after each update. In particular, when a vertex $v$ is inserted, it may change the eliminators of many of its neighbors, but we cannot afford to query all $\Omega(n)$ potential neighbors of $v$ to check which eliminators have changed. Instead, our solution is to maintain a leader mapping $\ell(v)$ which is an ``out-of-date'' version of the eliminator mapping. Each time we check if a vertex $v$ can be added to $\LFMIS_{k+1}$, by searching through its neighbors in $\LFMIS_{k+1}$, we ensure that either $v$ is added to $\LFMIS_{k+1}$ or its leader $\ell(v)$ is updated to the current eliminator of $v$, thereby aligning $\ell(v)$ with $\elim_\pi(v)$. However, thereafter, 
the values of $\ell(v)$ and $\elim_{\pi}(v)$ can become misaligned in several circumstances. In particular, the vertex $v$ may be moved into the queue $\cQ$ due to its leader $\ell(v)$ either leaving the LFMIS, or being pushed out of the top $k+1$ vertices in the LFMIS. In the second case, we show that $v$ can follow its leader to $\cQ$ without changing $\ell(v)$, however, in the first case $\ell(v)$ is necessarily modified. On the other hand, as noted, the eliminator of $v$ can also later change without the algorithm having to change $\ell(v)$.
Our analysis proceeds by a careful accounting, in which we demonstrate that each change in an eliminator can result in at most a constant number of changes to the leaders $\ell$, from which an amortized bound of $O(\log n)$ leader changes follows. %

\paragraph{Comparison to the prior $k$-center algorithm of \cite{ChaFul18}.}
The prior fully dynamic $k$-center algorithm of Chan, Gourqin, and Sozio \cite{ChaFul18}, which obtained an amortized $O(\eps^{-1}\log \Delta \cdot k^2)$ update time, also partially employed the idea of maintaining an LFMIS (although the connection to MIS under lexicographical orderings was not made explicit in that work). However, instead of consistently maintaining the LFMIS with respect to a random ranking $\pi$, they begin by maintaining an LFMIS with respect to the ordering $\pi'$ in which the points were originally inserted into the stream. Since this ordering is adversarial, deletions in the stream can initially be very expensive to handle. To prevent bad deletions from repeatedly occurring, whenever a deletion to a center $c$ occurs, the algorithm of \cite{ChaFul18} randomly reorders all points which are contained in clusters that come after $c$ in the current ordering being used. %
In this way, the algorithm of \cite{ChaFul18} gradually converts the adversarial ordering $\pi'$ into a random ordering $\pi$. However, by reordering \textit{all} points which occurred after a deleted center $c$, instead of just the set of points which were led by that center (via a mapping $\ell$), the amortized update time of the algorithm becomes $O(k^2)$.\footnote{Consider the stream which inserts $k$ clusters of equal size $n/k$, and then begins randomly deleting half of each cluster in reverse order. By the time a constant fraction of all the points are deleted, for each deletion the probability a leader is deleted is $\Omega(k/n)$, but such a deletion causes $O(nk)$ work to be done by the algorithm.} 
In contrast, one of our key insights is to update the entire clustering to immediately reflect a random LFMIS ordering after each update. 

\subsubsection{Algorithm for LSH Spaces}
The extension of our algorithm to LSH spaces is based on the following observation: each time we attempt to add a vertex $v$ to $\LFMIS_{k+1}$, we can determine the fate of $v$ solely by finding the vertex $u \in \LFMIS_{k+1}$ in the neighborhood of $v$ of minimal rank (i.e., the eliminator of $v$, if it is contained in $\LFMIS_{k+1}$). If $\pi(u) < \pi(v)$, we simply set $\ell(v) = u$ and proceed. Otherwise, we must find all other neighbors $w$ of $v$ in $\LFMIS_{k+1}$, remove them from the LFMIS, and set $\ell(w) = u$. Finding the vertex $u$ can therefore be cast as an $r$-\textit{near neighbor search} problem: here, one wants to return any $u \in \LFMIS_{k+1}$ which is at distance at most $r$ from $u$, with the caveat that we need to return such vertices in order based on their ranking. 
Since, whenever $u$ enters the LFMIS, each point $w$ that we search through which leaves $\LFMIS_{k+1}$ had its leader change, if we can find each consecutive neighbor of $u$ in $\LFMIS_{k+1}$ in time $\alpha$, we could hope to bound the total runtime of the algorithm by an $O(\alpha)$ factor more than the total number of leader changes, which we know to be small by analysis of the general metric space algorithm. 

To achieve values of $\alpha$ which are sublinear in $k$, we must necessarily settle for an \textit{approximate near neighbor search} (ANN) algorithm. A randomized, approximate $(r,cr)$-nearest neighbor data structure will return any point in $\LFMIS_{k+1}$ which is at distance at most $cr$, assuming there is at least one point at distance at most $r$ in $\LFMIS_{k+1}$. In other words, such an algorithm can be used to find all edges in $G_r$, with the addition of any arbitrary subset of edges in $G_{cr}$. By relaxing the notion of a threshold graph to allow for such a $c$-approximation, one can hope to obtain a $c(2+\eps)$-approximation to $k$-center via solving the $k$-Bounded MIS problem on each relaxed threshold graph. 

The key issue above is that, when using an ANN data structure, the underlying relaxed threshold graph is no longer a deterministic function of the point set $P$, and is instead ``revealed'' as queries are made to the ANN data structure. We handle this issue by demonstrating that, for the class of ANN algorithms based on locally sensitive hash functions, one can define a graph $G$ which is only a function of the randomness in the ANN data structure, and not the ordering $\pi$. The edges of this graph are defined in a natural way --- two points are adjacent if they collide in at least one of the hash buckets. By an appropriate setting of parameters, the number of collisions between points at distance larger than $cr$ can be made small. By simply ignoring such erroneous edges as they are queried, the runtime increases by a factor of the number of such collisions. %

\subsubsection{Lower bounds against an oblivious adversary}
\awr{commented out description of $\Omega(k^2)$ bound which is no longer mentioned above}
\aw{For proving the lower bound of $\Omega(k)$ we use the following hard distribution:}
in one case we randomly plant $k$ clusters each of size roughly $n/k$, where points within a cluster are close and points in separate clusters are far. In the second case, we do the same, and subsequently choose a point $i \sim [n]$ randomly and move it very far from all points (including its own cluster). Adaptive algorithms can gradually winnow the set of possible locations for $i$ by discovering connected components in the clusters, and eliminating the points in those components. Our proof follows by demonstrating that a large fraction of the input distribution results in any randomized algorithm making a sequence of distance queries which eliminates few data points, and therefore gives only a small advantage in discovering the planted point $i$.

\subsubsection{Lower bounds against an adaptive adversary}

\aw{We sketch our main ideas behind the lower bounds against an adaptive adversary.
For ease of presentation, }we discuss here a simplified setting in which the algorithm \begin{inenum}
\item can query only the distances between points that are currently
in $P$ (i.e., that have been introduced already but not yet removed)
and
\item has a worst-case update time of $f(k,n)$ (rather than amortized
update time) for some function $f$\end{inenum}. Note that these assumptions are not needed for the full proof presented in \Cref{sec:LBadap}. \aw{Our goal is to prove a lower bound on the 
approximation ratio of such an algorithm.}

The adversary starts by adding points into $P$ and maintains an auxiliary
graph $G=(V,E)$ with one vertex $v_{p}$ for each point $p$. Whenever
the algorithm queries the distance between two points $p,p'\in P$,
the adversary reports that they have a distance of 1 and adds an edge
$\{v_{p},v_{p'}\}$ of length 1 to $G$. Intuitively, the adversary
uses $G$ to keep track of the previously reported distances.
Whenever there is a vertex $v_{p}$ with a degree of at least $100f(k,n)$,
in the next operation the adversary deletes the corresponding point
$p$. There could be several such vertices, and then the adversary
deletes them one after the other. Thanks to assumption (i), once a point $p$ is deleted, the
degree of $v_{p}$ cannot increase further. Hence, the degrees of
the vertices in $G$ cannot grow arbitrarily. More precisely, one
can show that the degree of each vertex can grow to at most $O(\log n\cdot f(k,n))$.
 In particular, at least half of the
vertices in $G$ are at distance at least $\Omega(\log_{O(\log n\cdot f(k,n))}n)=\Omega\left(\log n\;/\;[\,\log\log n\cdot\log(f(k,n))\,]\right)$
to $v_{p}$.

Now observe that the algorithm knows only the point distances that
it queried, i.e., those that correspond to edges in $G$. For all other distances,
the triangle inequality imposes only an upper bound for any pairs of points whose corresponding
vertices are connected in $G$. Thus,
the algorithm cannot distinguish the setting where the underlying metric is
the shortest path metric in $G$ from the setting where all points
are at distance 1 to each other. If $k=1$, for any of the problems
under consideration, then for the selected center $c$ the algorithm
cannot distinguish whether all points are at distance 1 to $c$ or
if half of the points in $P$ are at distance $\Omega\left(\log n\;/\;[\,\log\log n\cdot\log(f(1,n))\,]\right)$
to $c$. Therefore, any algorithm which is correct with constant probability will suffer an approximation of at least $\Omega\left(\log n\;/\;[\,\log\log n\cdot\log(f(1,n))\,]\right)$.

We improve the above construction so that it also works for algorithms
that have amortized update times, are allowed to query
distances to points that are already deleted, and may output $O(k)$ instead of $k$ centers. 
To this end, we adjust
the construction so that for points with degree of at least $100f(k,n)$, where $f(k,n)$ is the (amortized) number of distance queries per operation,
we report distances that might be larger than~1 and that still allow
us to use the same argumentation as above. At the same time, we remove
the factor of $\log\log n$ in the denominator.

\subsubsection{Algorithms for $k$-Sum-of-radii and $k$-Sum-of-diameters}

Our algorithms for $k$-sum-of-radii and
$k$-sum-of-diameters against an oblivious adversary use the following paradigm:
we maintain bi-criteria approximations, i.e., solutions with a small
approximation ratio that might use more than $k$ centers, i.e., up to $O(k/\eps)$ centers.
In a second step, we use the centers of this solution as the input
to an auxiliary dynamic instance for which we maintain a solution
with only $k$ centers. These centers then form our solution to the
actual problem, by increasing their radii appropriately. Since the
input to our auxiliary instance is much smaller than the input to
the original instance, we can afford to use algorithms for it whose
update times have a much higher dependence on $n$, \aw{and in fact we can 
even recompute the whole solution from scratch. }

More precisely, recall that in the static offline setting there is a $(3.504+\eps)$-approximation
algorithm with running time $n^{O(1/\eps)}$ for $k$-sum-of-radii~\cite{CharikarP04}.
We provide a black-box reduction that transforms any such offline
algorithm with running time of $g(n)$ and an approximation ratio of
$\alpha$ into a fully dynamic $(6+2\alpha+\eps)$-approximation algorithm
with an update time of $O(((k/\epsilon)^{5}+g(\mathrm{poly}(k/\epsilon)))\log\Delta)$.
To this end, we introduce a dynamic primal-dual algorithm that maintains
a bi-criteria approximation with up to $O(k/\epsilon)$ centers. After
each update, we use these $O(k/\epsilon)$ centers as input for the
offline $\alpha$-approximation algorithm, run it from scratch, and
increase the radii of the computed centers appropriately such that
they cover all points of the original input instance.

For computing the needed bi-criteria approximation, there is a polynomial-time
offline primal-dual $O(1)$-approximation algorithm with a running
time of $\Omega(n^{2}k)$~\cite{CharikarP04} from which we borrow
ideas. Note that a dynamic algorithm with an update time of $(k\log n)^{O(1)}$
yields an offline algorithm with a running time of $n(k\log n)^{O(1)}$
(by inserting all points one after another) and no offline algorithm
is known for the problem that is that fast. Hence, we need new ideas.
First, we formulate the problem as an LP $(P)$ which has a variable
$x_{p}^{(r)}$ for each combination of a point $p$ and a radius~$r$
from a set of (suitably discretized) radii $R$, and a constraint
for each point $p$. Let $(D)$ denote its dual LP, see below, where
$z:=\epsilon\OPT'/k$ and $\OPT'$ is a $(1+\epsilon)$-approximate estimate on the value $\OPT$ of the optimal solution.
{
\thickmuskip=4mu
\medmuskip=3mu
\thinmuskip=2mu
\begin{alignat*}{9}
\min & \,\,\, & \sum_{p\in P}\sum_{r\in R}x_{p}^{(r)}(r+z) &  &  &  &  &  &  &  & \max & \,\,\, & \sum_{p\in P}y_{p}\,\,\,\,\,\,\,\,\,\,\,\\
\mathrm{s.t.} &  & \sum_{p'\in P}\sum_{r:d(p,p')\le r}x_{p'}^{(r)} & \ge1 & \,\,\, & \forall p\in P & \,\,\,\,\,\,\,\: & (P) & \,\,\,\,\,\,\,\,\,\,\,\,\,\,\, &  & \mathrm{s.t.} &  & \sum_{p'\in P:d(p,p')\le r}y_{p'} & \le r+z & \,\,\, & \forall p\in P,r\in R & \,\,\,\,\,\,\,\:(D)\\
 &  & x_{p}^{(r)} & \ge0 &  & \forall p\in P\,\,\forall r\in R &  &  &  &  &  &  & y_{p} & \ge0 &  & \forall p\in P
\end{alignat*}
}

We select a point $c$ randomly and raise its dual variable $y_{c}$.
Unlike \cite{CharikarP04}, we raise $y_{c}$ only until the constraint
for~$c$ and some radius $r$ becomes \emph{half-tight, }i.e., $\sum_{p':d(c,p')\le r}y_{p'}=r/2+z$.\emph{
}We show that then we can guarantee that no other dual constraint
is violated, which saves running time since we need to check only
the (few, i.e. $|R|$ many) dual constraints for~$c$, and not the constraints for
all other input points when we raise $y_{c}$. We add $c$ to our
solution and assign it a radius of~$2r$. Since the constraint for
$(c,r)$ is half-tight, our dual can still ``pay'' for including $c$
with radius $2r$ in our solution. %

In primal-dual algorithms, one often raises all primal variables whose
constraints have become tight, in particular in the corresponding
routine in~\cite{CharikarP04}. However, since we assign to $c$
a radius of $2r$, we argue that we do not need to raise the up to
$\Omega(n)$ many primal variables corresponding to dual constraints
that have become (half-)tight. Again, this saves a considerable amount
of running time. Then, we consider the points that are not covered
by $c$ (with radius $2r$), select one of these points uniformly
at random, and iterate. After a suitable pruning routine (which is
faster than the corresponding routine in~\cite{CharikarP04}) this
process opens $k'= O(k/\eps)$ centers at cost $(6+\eps)\OPT'$,
or asserts that $\OPT>\OPT'$. We show that we can maintain this solution dynamically.

As mentioned above, after each update we feed the centers of the bi-criteria
approximation as input points to our (arbitrary) static offline $\alpha$-approximation
algorithm for $k$-sum-of-radii and run it. Finally, we translate
its solution to a $(6+2\alpha+\eps)$-approximate solution to the original
instance. For the best known offline approximation algorithm \cite{CharikarP04}
it holds that $\alpha=3.504$, and thus we obtain a ratio of $13.008+\eps$
overall. For $k$-sum-of-diameters the same solution yields a $(26.016+2\epsilon)$-approximation.
\vspace{-.2cm}

\subsection{Other Related Work}

\paragraph{$k$-means and $k$-median.}
For $k$-means~\cite{HenzingerK20} gave a conditional lower bound against an oblivious adversary for any dynamic better-than-4 approximate $k$-means algorithm showing that for any $\gamma > 0$ no algorithm with $O(k^{1-\gamma})$ update time and $O(k^{2-\gamma})$ query time exists. 
There exists a fully dynamic constant approximation algorithm for $k$-means and $k$-median with expected amortized update time $\tilde O(n + k^2)$ that tries to minimize the number of center changes~\cite{cohen2019fully}. 
For $d$-dimensional Euclidean metric space dynamic algorithms for $k$-median and $k$-means exist by combining dynamic coreset algorithms for that metric space with a static algorithm, but their running time is in $O(poly(\epsilon^{-1}, k, \log \Delta, d))$~\cite{FrahlingS05,FeldmanSS13,BravermanFLSY17,DBLP:conf/focs/SohlerW18}.

\paragraph{$k$-sum-of-radii.} For a variant of the sum-of-radii problem~\cite{henzinger2020dynamic} gives a fully dynamic algorithm in metric spaces with doubling dimension $\kappa$ that achieves a $O(2^{2\kappa})$ approximation in time $O(2^{6\kappa}\log \Delta)$.

\section{Preliminaries}\label{sec:prelims}

We begin with basic notation and definitions. For any positive integer $n$, we write $[n]$ to denote the set $ \{1,2,\dots,n\}$. In what follows, we will fix any metric space $(\cX,d)$. 
A fully dynamic stream is a sequence $(p_1,\sigma_1),\dots,(p_M,\sigma_M)$ of $M$ updates such that $p_i \in \cX$ is a point, and $\sigma_i \in \{+,-\}$ signifies either an insertion or deletion of a point. Naturally, we assume that a point can only be deleted if it was previously inserted. Moreover, we assume that each point is inserted at most once before being deleted.
\awr{rephrased this sentence, note that now we have several objectives. The old version is commented}

We call a point $p \in \cX$ active at time $t$ if $p$ was inserted at some time $t' < t$, and not deleted anytime between $t'$ and $t$. We write $P^{t} \subset \cX$ to denote the set of active points at time $t$. We let $r_{\min},r_{\max}$ be reals such that for all $t \in [M]$ and $x,y \in P^t$, we have $r_{\min} \leq d(x,y) \leq r_{\max}$, and set $\Delta = r_{\max} / r_{\min}$ to be the aspect ratio of the point set. As in prior works \cite{ChaFul18,schmidt2019fully}, we assume that an upper bound on $\Delta$ is known.%

\paragraph{Clustering Objectives.}
In the $k$-center problem, given a point set $P$ living in a metric space $(\cX,d)$, the goal is to output a set of $k$ \textit{centers} $\cC = \{c_1,\dots,c_k\} \subset \cX$, along with a mapping $\ell: P \to \{c_1,\dots,c_k\}$, 
such that the following objective is minimized:
\[ \cost[P]{\cC}[k][\infty] =   \max_{p \in P} d(p,\ell(p))    \]
In other words, one would like for the maximum distance from $p$ to the $\ell(p)$, over all $p \in P$, to be minimized. A related objective is $k$-sum-of-radii, where the maximum distance to a point is considered for each cluster and added instead of taking the maximum:
\[ \cost[P]{\cC}[k][R\Sigma] =   \sum_{c \in \mathcal{C}} \max_{p \in \ell^{-1}(c)} d(p,c)    \]
For the $k$-sum-of-diameters objective, one considers the pairwise distances of points assigned to the same clusters instead of the distances from each point to its center.
\[ \cost[P]{\cC}[k][D\Sigma] =   \sum_{c \in \mathcal{C}} \max_{p,q \in \ell^{-1}(c)} d(p,q)    \]
Additionally, for any real $z >0$, we introduce the $(k,z)$-clustering problem, which is to minimize 
\[ \cost[P]{\cC}[k][z] =   \sum_{p \in P} d^z(p,\ell(p))    \]

We will be primarily concerned the the $k$-center objective, but we introduce the more general $(k,z)$-clustering objective, which includes both $k$-median (for $z=1$) and $k$-means (for $z=2$), as our lower bounds from Section \ref{sec:LB} will hold for these objectives as well.

We remark that while $\ell(p)$ is usually fixed by definition to be the closest point to $p$ in $C$, the closest point may not necessarily be easy to maintain in a fully dynamic setting. Therefore, we will evaluate the cost of our algorithms with respects to both the centers and the mapping $\ell$ from points to centers. For any $p \in P$, we will refer to $\ell(p)$ as the \textit{leader} of $p$ under the mapping $\ell$, and the set of all points lead by a given $c_i$ is the cluster led by $c_i$.

In addition to maintaining a clustering with approximately optimal cost, we would like for our algorithms to be able to quickly answer queries related to cluster membership, and enumeration over all points in a cluster. Specifically, we ask that our algorithm be able to answer the following queries at any time step $t$: 

\begin{figure}[H]
\begin{Frame}[Queries to a Fully Dynamic Clustering Algorithm]
\begin{enumerate}
    \item \textbf{Membership Query:} Given a point $p \in P^{t}$, return the center $c = \ell(p)$ of the cluster $C$ containing~$p$.
    \item  \textbf{Cluster Enumeration:} Given a point $p \in P^{t}$, list all points in the cluster $C$ containing~$p$.
\end{enumerate}
\end{Frame}
\end{figure}

In particular, after processing any given update, our algorithms will be capable of responding to membership queries in $O(1)$-time, and to clustering enumeration queries in time $O(|C|)$, where $C$ is the clustering containing the query point.

\section{From Fully Dynamic $k$-Center to $k$-Bounded MIS} 
\label{sec:kCenters}

In this section, we describe our main results for fully dynamic $k$-center clustering, based on our main algorithmic contribution, which is presented in Section \ref{sec:generalMetric}. We begin by describing how
the problem of $k$-center clustering of $P$ can be reduced to maintaining a maximal independent set (MIS) in a graph. In particular, the reduction will only require us to solve a weaker version of MIS, where we need only return a MIS of size at most $k$, or an independent set of size at least $k+1$. Formally, this problem, which we refer to as the $k$-Bounded MIS problem, is defined as follows.

\begin{definition}[$k$-Bounded MIS]
Given a graph $G = (V,E)$ and an integer $k \geq 1$, the $k$-bounded MIS problem is to output a maximal independent set $\cI \subset V$ of size at most $k$, or return an independent set $\cI \subset V$ of size at least $k+1$.%
\end{definition}

\paragraph{Reduction from $k$-center to $k$-Bounded MIS.}
The reduction from $k$-center to computing a maximum independent set in a graph is well-known, and can be attributed to the work of Hochbaum and Shmoys \cite{hochbaum1986unified}. The reduction was described in Section \ref{sec:tech}, however, both for completeness and so that it is clear that only a $k$-Bounded MIS is required for the reduction, we spell out the full details here.

Fix a set of points $X$ in a metric space, such that $r_{\min} \leq d(x,y) \leq  r_{\max}$ for all $x,y \in \cX$. Then, for each $r= r_{\min}, (1+\eps/2) r_{\min} , (1+\eps/2)^2 r_{\min}, \dots, r_{\max}$, one creates the \textit{$r$-threshold graph} $G_r = (V,E_r)$, which is defined as the graph with vertex set $V = X$, and $(x,y) \in E_r$ if and only if $d(x,y) \leq r$. One then runs an algorithm for $k$-Bounded MIS on each graph $G_r$, and finds the smallest value of $r$ such that the output of the algorithm $\cI_r$ on $G_r$ satisfies $|\cI_r| \leq k$ --- in other words, $\cI_r$ must be a MIS of size at most $k$ in $G_r$. Observe that $\cI_r$ yields a solution to the $k$-center problem with cost at most $r$, since each point in $X$ is either in $\cI_r$ or is at distance at most $r$ from a point in $\cI_r$. Furthermore, since the independent set $\cI_{r/(1+\eps/2)}$ returned from the algorithm run on $G_{r/(1+\eps/2)}$ satisfies $|\cI_{r/(1+\eps/2)}| \geq k+1$ it follows that there are $k+1$ points in $X$ which are pair-wise distance at least $r/(1+\eps/2)$ apart. Hence, the cost of any $k$-center solution (which must cluster two of these $k+1$ points together) is at least $r/(2+\eps)$ by the triangle inequality. It follows that $\cI$ yields a $2+\eps$ approximation of the optimal $k$-center cost.

Note that, in addition to maintaining the centers $\cI$, for the purposes of answering membership queries, one would also like to be able to return in $O(1)$ time, given any $x \in \cX$, a fixed $y \in \cI$ such that $d(x,y) \leq r$. We will ensure that our algorithms, whenever they return a MIS $\cI$ with size at most $k$, also maintain a mapping $\ell: V \setminus \cI \to \cI$ which maps any $x$, which is not a center, to its corresponding center $\ell(x)$. 

Observe that in the context of $k$-clustering, insertions and deletions of points correspond to insertions and deletions of entire vertices into the graph $G$. This is known as the fully dynamic \textit{vertex update model}. Since one vertex update can cause as many as $O(n)$ edge updates, we will not be able to read all of the edges inserted into the stream. Instead, we assume our dynamic graph algorithms can query for whether $(u,v)$ is an edge in constant time (i.e., constant time oracle access to the adjacency matrix).\footnote{This is equivalent to assuming that distances in the metric space can be computed in constant time, however if such distances require $\alpha$ time to compute, this will only increase the runtime of our algorithms by a factor of $\alpha$.}

Our algorithm for $k$-Bounded MIS will return a very particular type of MIS. Specifically, we will attempt to return the first $k+1$ vertices in a \textit{Lexicographically First MIS} (LFMIS), under a random lexicographical ordering of the vertices.

\paragraph{Lexicographically First MIS (LFMIS).} The LFMIS of a graph $G = (V,E)$ according to a ranking of the vertices specified by a mapping $\pi:V \to [0,1]$ is a unique MIS defined as by the following process. Initially, every vertex $v \in V$ is alive. We then iteratively select the alive vertex with minimal rank $\pi(v)$, add it to the MIS, and then kill $v$ and all of its alive neighbors. We write $\LFMIS(G,\pi)$ to denote the LFMIS of $G$ under $\pi$. For each vertex $v$, we define the \textit{eliminator} of $v$, denoted $\elim_{G,\pi}(v)$ to be the vertex $u$ which kills $v$ in the above process; namely,  $\elim_{G,\pi}(v)$ is the vertex with smallest rank in the set $(N(v) \cup \{v\}) \cap \LFMIS(G,\pi)$.

\begin{definition}\label{def:topkLFMIS}
Given a graph $G = (V,E)$, $\pi:V \to [0,1]$, and an integer $k \geq 1$, we define the top-$k$ LFMIS of $G$ with respect to $\pi$, denoted $\LFMIS_k(G,\pi)$ to be the set consisting of the first $\min\{k,|\LFMIS(G,\pi)|\}$ vertices in $\LFMIS(G,\pi)$ (where the ordering is with respect to $\pi$). When $G,\pi$ are given by context, we simply write $\LFMIS_k$. 
\end{definition}

It is clear that returning a top-$(k+1)$ LFMIS of $G$ with respect to any ordering will solve the $k$-Bounded MIS problem. In order to also obtain a mapping $\ell$ from points to their centers in the independent set, we define the following augmented version of the top-$k$ LFMIS problem, which we refer to as a \textit{top-$k$ LFMIS with leaders}.

\begin{definition}\label{def:LFMISLead}
A top-$k$ LFMIS with leaders consists of the set $\LFMIS_k(G,\pi)$, along with a \textit{leader mapping function} $\ell:V \to V \cup \{\bot\}$, such that $(v,\ell(v)) \in E$ whenever $\ell(v) \neq \bot$, and such that if $\LFMIS_k(G,\pi) = \LFMIS(G,\pi)$, then $\ell(v) \in \LFMIS_k(G,\pi)$ for all $v \in V \setminus \LFMIS_k(G,\pi)$, and $\ell(v) = \bot$ for all $v \in \LFMIS_k(G,\pi)$. 
\end{definition}

Within our algorithm, the event that $\ell(v) = \bot$ will occur only if $v$ itself is a leader in $\LFMIS_{k}$, or if $v$ is among a set of ``unclustered'' points whose leader mapping to a point in $\LFMIS$ may be out of date. We choose to set $\ell(v) = \bot$ when $v \in \LFMIS_k$ is a leader, rather than setting $\ell(v) = v$, to maintain the invariant that $(v,\ell(v)) \in E$ whenever $\ell(v) \neq \bot$. 

The main goal of the following Section \ref{sec:generalMetric} will be to prove the existence of a $\tilde{O}(k)$ amortized update time algorithm for maintaining a top-$k$ LFMIS with leaders of a graph $G$ under a fully dynamic sequence of insertions and deletions of vertices from $G$. Specifically, we will prove the following theorem.

\noindent \textbf{Theorem} \ref{thm:LFMISMain}. {\it There is a algorithm which, on a fully dynamic stream of insertions and deletions of vertices to a graph $G$, maintains at all time steps a top-$k$ LFMIS of $G$ with leaders under a random ranking $\pi: V \to [0,1]$. The expected amortized per-update time of the algorithm is $O(k \log n + \log^2 n)$, where $n$ is the maximum number active of vertices at any time. Moreover, the algorithm does not need to know $n$ in advance. }

Next, we demonstrate how Theorem \ref{thm:LFMISMain} immediately implies the main result of this work (Theorem \ref{thm:main}). Firstly, we prove a proposition which demonstrates that any fully dynamic Las Vegas algorithm with small runtime in expectation can be converted into a fully dynamic algorithm with small runtime with high probability.

\begin{proposition}\label{prop:highProb}
Let $\cA$ be any fully dynamic randomized algorithm that correctly maintains a solution for a problem $\cP$ at all time steps, and runs in amortized time at most $\Run(\cA)$ in expectation. Then there is a fully dynamic algorithm for $\cP$ which runs in amortized time at most $O(\Run(\cA) \log\delta^{-1})$ with probability $1-\delta$ for all $\delta \in (0,1/2)$.
\end{proposition}
\begin{proof}
The algorithm is as follows: we maintain at all time steps a single instance of $\cA$ running on the dynamic stream. If, whenever the current time step is $t$, the total runtime of the algorithm exceeds $4 t \Run(\cA)$, we delete $\cA$ and re-instantiate it with fresh randomness. We then run the re-instantiated version of $\cA$ from the beginning of the stream until either we reach the current time step $t$, or the total runtime again exceeds $4 t \Run(\cA)$, in which case we re-instantiate again. 

Let $M$ be the total length of the stream. For each $i=0,1,2,\dots,\lceil \log M \rceil$, let $\bZ_i$ be the number of times that $\cA$ is re-instantiated while the current time step is between $2^i$ and $2^{i+1}$. Note that the total runtime is then at most 
\[ 4 \Run(\cA) \cdot \left(M +  \sum_{i=0}^{\lceil \log M \rceil} 2^{i+1} \bZ_i \right) \]
Fix any $i \in \{0,1,\dots,\lceil \log M \rceil\}$, and let us bound the value $\bZ_i$. 
Each time that $\cA$ is restarted when the current time $t$ step is between $2^i$ and $2^{i+1}$, the probability that the new algorithm runs in time more than $2^{i+2}\Run(\cA) \geq 4  t \Run(\cA) $ on the first $2^{i+1}$ updates is at most $1/2$ by Markov's inequality. Thu, the probability that $\bZ_i > T_i + 1$ is at most $2^{-T_i}$, for any $T_i \geq 0$. Setting $T_i = \log(2/\delta) + \lceil \log M \rceil - i $, we have
\begin{equation}
    \begin{split}
        \sum_{i=0}^{\lceil \log M \rceil} \pr{\bZ_i > T_i + 1 }& \leq \frac{\delta}{2}   \sum_{i=0}^{\lceil \log M \rceil} \left(\frac{1}{2}\right)^{\lceil \log M \rceil - i} \\ 
        & < \delta
    \end{split}
\end{equation}
In other words, by a union bound, we have $\bZ_i \leq T_i+1$ for all $i=0,1,2,\dots,\lceil \log M \rceil$, with probability at least $1-\delta$. Conditioned on this, the total runtime is at most 

\begin{equation}
    \begin{split}
         4 \Run(\cA) \cdot \left(M +  \sum_{i=0}^{\lceil \log M \rceil} 2^{i+1} (T_i+1) \right) &= O\left(\Run(\cA)\left( \sum_{i=0}^{\lceil \log M \rceil} 2^{i} (  \log \delta^{-1} + \lceil \log M \rceil - i )  \right)\right) \\
         & = O\left(\Run(\cA)\left( M \log \delta^{-1} + M \sum_{i=0}^{\lceil \log M \rceil} \frac{i }{2^{i}}  \right)\right) \\
          & = O\left(\Run(\cA)\cdot  M \log \delta^{-1} \right) \\
    \end{split}
\end{equation}
which is the desired total runtime.
\end{proof}

Given Theorem \ref{thm:LFMISMain}, along with the reduction to top-$k$ LFMIS from $k$-center described in this section, which runs $O(\eps^{-1} \log \Delta)$ copies of a top-$k$ LFMIS algorithm, we immediately obtain a fully dynamic $k$-center algorithm with $\tilde{O}(k)$ expected amortized update time. By then applying Proposition \ref{prop:highProb}, we obtain our main theorem, stated below. 

	\noindent \textbf{Theorem} \ref{thm:main}. {\it 
    There is a fully dynamic algorithm which, on a sequence of insertions and deletions of points from a metric space $\cX$, maintains a $(2+\eps)$-approximation to the optimal $k$-center clustering. The amortized update time of the algorithm is $O(\frac{\log \Delta \log n}{\eps}(k  + \log n))$ in expectation, and $O(\frac{\log \Delta \log n}{\eps}(k  + \log n) \log \delta^{-1})$  with probability $1-\delta$ for any $\delta \in (0,\frac{1}{2})$, where $n$ is the maximum number of active points at any time step. 
    
    The algorithm can answer membership queries in $O(1)$-time, and enumerate over a cluster $C$ in time $O(|C_i|)$. 
}

\section{Fully Dynamic $k$-Bounded MIS with Vertex Updates}
\label{sec:generalMetric}

Given the reduction from Section \ref{sec:kCenters}, the goal of this section will be to design an algorithm which maintains a top-$k$ LFMIS with leaders  (Definition \ref{def:LFMISLead}) under a graph which receives a fully dynamic sequence of \textit{vertex} insertions and deletions. As noted, maintaining a top-$(k+1)$ LFMIS immediately results in a solution to the $k$-Bounded MIS problem. We begin by formalizing the model of vertex-valued updates to dynamic graphs.

\paragraph{Fully Dynamic Graphs with Vertex Updates}
 In the vertex-update fully dynamic setting, at each time step a vertex $v$ is either inserted into the current graph $G$, or deleted from $G$, along with all edges incident to $v$. This defines a sequence of graphs $G^1, G^2,\dots,G^M$, where $G^t= (V^t,E^t)$ is the state of the graph after the $t$-th update. Equivalently, we can think of there being an ``underlying'' graph $G = (V,E)$, where at the beginning all vertices are \textit{inactive}. At each time step, either an active vertex is made inactive, or vice-versa, and $G^{t}$ is defined as the subgraph induced by the active vertices at time $t$. The latter is the interpretation which will be used for this section. 

Since the degree of $v$ may be as large as the number of active vertices in $G$, our algorithm will be unable to read all of the edges incident to $v$ when it arrives. Instead, we require only query access to the adjacency matrix of the underlying graph $G$. Namely, we assume that we can test in constant time whether $(u,v) \in E$ for any two vertices $u,v$.

For the purpose of $k$-center clustering, we will need to maintain a top-$k$ LFMIS with leaders $(\LFMIS_k(G,\pi),\ell)$, along with a Boolean value indicating whether 
$\LFMIS_k(G,\pi) = \LFMIS(G,\pi)$. To do this, we can instead attempt to maintain the set $\LFMIS_{k+1}(G,\pi)$, as well as a leader mapping function $\ell:V \to V \cup \{\bot\}$, with the relaxed property that if $\LFMIS_k(G,\pi) = \LFMIS(G,\pi)$, then $\ell(v) \in \LFMIS(G,\pi)$ for all $v \in V \setminus \LFMIS(G,\pi)$ and $\ell(v) = \bot$ for all $v \in \LFMIS(G,\pi)$. We call such a leader function $\ell$ with this relaxed property a \textit{modified} leader mapping. Thus, in what follows, we will focus on maintaining a top-$k$ LFMIS with this modified leader mapping.

\subsection{The Data Structure}
We now describe the main data structure and algorithm which will maintain a top-$k$ LFMIS with leaders in the dynamic graph $G$.
We begin by fixing a random mapping $\pi:V \to [0,1]$, which we will use as the ranking for our lexicographical ordering over the vertices. It is easy to see that if $|V| = n$, then by discretizing $[0,1]$ so that $\pi(v)$ can be represented in $O(\log n)$ bits we will avoid collisions with high probability. At every time step, the algorithm will maintain an ordered set $\ALG $ of vertices in a linked list, sorted by the ranking $\pi$, with $|\ALG| \leq k+1$. We will prove that, after every update $t$, we have $\ALG = \LFMIS_{k+1}(G^t,\pi)$.\footnote{We use a separate notation $\ALG$, instead of $\LFMIS_{k+1}$, to represent the set maintained by the algorithm, until we have demonstrated that we indeed have $\ALG = \LFMIS_{k+1}(G^t,\pi)$ at all time steps $t$.}  We will also maintain a mapping $\ell:V \to V \cup \{\bot\}$ which will be our leader mapping function. Initially, we set $\ell(v) = \bot$ for all $v$. 
Lastly, we will maintain a (potentially empty) priority queue $\cQ$ of \textit{unclustered} vertices, where the priority is similarly given by $\pi$.

Each vertex $v$ in $G_t$ will be classified as either a \textit{leader}, a \textit{follower}, or \textit{unclustered}. Intuitively, when $\LFMIS_k(G,\pi) = \LFMIS(G,\pi)$, the leaders will be exactly the points in $\LFMIS_k(G,\pi)$, the followers will be all other points $v$ which are mapped to some $\ell(v) \in \LFMIS_k(G,\pi)$ (in other words, $v$ ``follows'' $\ell(v)$), and there will be no unclustered points. At intermediate steps, however, when $|\LFMIS(G,\pi)| \geq k + 1$, we will be unable to maintain the entire set $\LFMIS(G,\pi)$ and, therefore, we will store the set of all vertices which are not in $\LFMIS_{k+1}$ in the priority queue $\cQ$ of unclustered vertices. The formal definitions of leaders, followers, and unclustered points follow.

Every vertex currently maintained in $\ALG$ is be a leader. Each leader $v$ may have a set of follower vertices, which are vertices $u$ with $\ell(u) = v$, in which case we say that $u$ follows $v$. By construction of the leader function $\ell$, every follower-leader pair $(u,\ell(u))$ will be an edge of $G$.  We write $\cF_v = \{u \in V : \ell(u) = v\}$ to denote the (possibly empty) set of followers of a leader $v$. For each leader, the set $\cF_v$ will be maintained as part of the data structure at the vertex $v$.

Now when the size of $\LFMIS$ exceeds $k+1$, we will have to remove the leader $v$ in $\LFMIS$ with the largest rank, so as to keep the size of $\ALG$ at most $k+1$. The vertex $v$ will then be moved to the queue $\cQ$, along with its priority $\pi(v)$. The set $\cF_v$ of followers of $v$ will continue to be followers of $v$ --- their status remains unchanged. In this case, the vertex $v$ is now said to be an \textit{inactive} leader, whereas each leader currently in $\ALG$ is called an \textit{active} leader. If, at a later time, we have $\pi(v) < \max_{u \in \ALG} \pi(u)$, then it is possible that $v$ may be part of $\LFMIS_{k+1}$, in which case we will attempt to reinsert the inactive leader $v$ from $\cQ$ back into $\ALG$. Note, importantly, that whenever $\pi(v) < \max_{u \in \ALG} \pi(u)$ occurs at a future time step for a vertex $v \in \cQ$, then either $v$ is part of $\LFMIS_{k+1}$, or it is a neighbor of some vertex $u \in \LFMIS_{k+1}$ of lower rank. In both cases, we can remove $v$ from $\cQ$ and attempt to reinsert it, with the guarantee that after this reinsertion $v$ will either be an active leader, or a follower of an active leader.

\paragraph{The Unclustered Queue.} We now describe the purpose and function of the priority queue $\cQ$. 
Whenever either a vertex $v$ is inserted into the stream, or it is a follower of a leader $\ell(v)$ who is removed from the $\LFMIS$, we must attempt to reinsert $v$, to see if it should be added to $\LFMIS_{k+1}$. However, if $|\ALG| = k+1$, then the only way that $v$ should be a part of $\LFMIS_{k+1}$ (and therefore added to $\ALG$) is if $\pi(v) < \max_{u \in \ALG} \pi(u)$. If this does not occur, then we do not need to insert $v$ right away, and instead can defer it to a later time when either $|\ALG| < k+1$ or $\pi(v) < \max_{u \in \ALG} \pi(u)$ holds. Moreover, by definition of the modified leader mapping $\ell$, we only need to set $\ell(v)$ when $|\LFMIS_{k+1}| < k+1$. We can therefore add $v$ to the priority queue $\cQ$. 

Every point in the priority queue is called an \textit{unclustered point}, as they are not currently part of a valid $k$-clustering in the graph. By checking the top of the priority queue at the end of processing each update, we can determine whenever either of the events $|\ALG| \leq k$ or $\pi(v) < \max_{u \in \ALG} \pi(u)$ holds; if either is true, we iteratively attempt to reinsert the top of the queue until the queue is empty or both events no longer hold. This will ensure that either all points are clustered (so $\cQ = \emptyset$), or $|\ALG| = k+1$ and $\ALG = \LFMIS_{k+1}$ (since no point in the queue could have been a part of $\LFMIS_{k+1}$).

\paragraph{The Leader Mapping.}  
Notice that given a top-$k$ LFMIS, a valid leader assignment is always given by $\ell(v) = \elim_{G,\pi}(v)$, where $\elim_{G,\pi}(v)$ is the eliminator of $v$ via $\pi$ as defined in Section \ref{sec:kCenters}; this is the case since if  $\LFMIS_k(G,\pi) = \LFMIS(G,\pi)$ then each vertex is either in $\LFMIS_k(G,\pi)$ or eliminated by one of the vertices in $\LFMIS_k(G,\pi)$.
Thus, intuitively, our goal should be to attempt to maintain that $\ell(v) = \elim_{G,\pi}(v)$ for all $v \notin \LFMIS_k(G,\pi)$.
However, the addition of a new vertex which enters $\LFMIS_{k}(G,\pi)$ can change the eliminators of many other vertices \textit{not} in $\LFMIS_k(G,\pi)$. Discovering which points have had their eliminator changed immediately on this time step would be expensive, as one would have to search through the followers of all active leaders to see if any of their eliminators changed. Instead, we note that at this moment, so long as the new vertex does not share an edge with any other active leader, we do not need to modify our leader mapping. Instead, we can \textit{defer} the reassignment of the leaders of vertices $v$ whose eliminator changed on this step, to a later step when their leaders are removed from $\LFMIS_{k+1}(G,\pi)$. Demonstrating that the number of changes to the leader mapping function $\ell$, defined in this way, is not too much larger than the number of changes to the eliminators of all vertices, will be a major component of our analysis.

\paragraph{The Algorithm and Roadmap.} Our main algorithm is described in three routines: Algorithms \ref{alg:main}, \ref{alg:ins}, and \ref{alg:del}. Algorithm \ref{alg:main} handles the inital insertion or deletion of a vertex in the fully dynamic stream, and then calls at least one of Algorithms \ref{alg:ins} or \ref{alg:del}. Algorithm \ref{alg:ins} handles insertions of vertices in the data structure, and Algorithm \ref{alg:del} handles deletions of vertices from the data structure. We begin in Section \ref{sec:correctness} by proving that our algorithm does indeed solve the top-$k$ LFMIS problem with the desired modified leader mapping. Then, in Section \ref{sec:amortized}, we analyze the amortized runtime of the algorithm.

\begin{algorithm}[ht]
\DontPrintSemicolon
	\caption{Process Update}\label{alg:main}
	\KwData{An update $(v,\sigma)$, where $\sigma \in \{+,-\}$.}
 \If{$\sigma = +$ is an insertion of $v$}{
 Generate $\pi(v)$, and set $\ell(v) = \bot$.\;
 Call $\ins(v,\pi(v))$.\;}
\If{$\sigma = -$ is a deletion of $v$}{
Call $\delete(v)$.\;
}

\While{$|\cQ| \neq \emptyset \boldsymbol{\wedge} \left( |\ALG| \leq  k+1 \boldsymbol{\vee} \min_{w \in \cQ} \pi(w) < \max_{w \in \ALG} \pi(w)  \right)$}{ \label{line:whileMain}  
$u \leftarrow \arg \min_{w \in \cQ} \pi(w)$.\;
Delete $u$ from $\cQ$, and call $\ins(u,\pi(u))$. \;
}
\end{algorithm}

\begin{algorithm}[ht]
\DontPrintSemicolon
	\caption{$\ins(v,\pi(v))$} \label{alg:ins}
 \If{$|\ALG| = k+1$ $\boldsymbol{\wedge}$ $\pi(v) > \max_{u \in \ALG} \pi(u)$}{\label{line:firstIfIns}
 Insert $(v,\pi(v))$ into $\cQ$. \;}
\Else{
Compute $S =  \ALG \cap N(v)$\;
\If{$S = \emptyset$}{
Add $v$ to $\ALG$. \;
\If{ $|\ALG| = k+2$}{
Let $u = \arg \max_{u' \in \ALG} \pi(u') $. \;
Remove $u$ from $\ALG$, and insert $(u,\pi(u))$ into $\cQ$. \label{line:delOverflow} \;
}
}
\Else{
    $u^* = \arg \min_{u' \in S} \pi(u')$\;
    \If{$\pi(u^*) < \pi(v)$}{
    \If{$v$ is a leader}{
    For each $w \in \cF_v$, insert $(w,\pi(w))$ to $\cQ$, and set $\ell(w) = \bot$ \label{line:queue1Ins} \;
    Delete the list $\cF_v$. \;
    }
    Add $v$ to $\cF_{u^*}$ as a follower of $u^*$, set $\ell(v) = u^*$. \;
    }
    \Else{
    For each $w \in \cup_{u \in S} \cF_u$, add $(w,\pi(w))$ to $\cQ$, and set $\ell(w) = \bot$. \label{line:queue2Ins} \;
    For each $u \in S$, set $\ell(u) = v$ to be a follower of $v$, remove $u$ from $\ALG$, and delete the list $\cF_u$. \;
    }
 }  
}
\end{algorithm}

\begin{algorithm}[!ht]
\DontPrintSemicolon
	\caption{$\delete(v)$ }\label{alg:del}
	
	\If{$v$ is a follower}{
	Delete $v$ from $\cF_{\ell(v)}$, and remove $v$ from the set of vertices.  \; 
	}
	\ElseIf{$v \in \cQ$}{
	\If{$v$ is a leader}{
	 For each $w \in \cF_v$, insert $(w,\pi(w))$ to $\cQ$, and set $\ell(w) = \bot$ \label{line:queue1Del} \;
	 Delete the list $\cF_v$, and remove $v$ from $\cQ$ and the set of vertices. \;
	}
	\Else{
	Delete $v$ from $\cQ$ and the set of vertices. \;
	}
	}
	\Else(\tcc*[r]{Must have $v \in \ALG$}){
	 For each $w \in \cF_v$, insert $(w,\pi(w))$ to $\cQ$, and set $\ell(w) = \bot$\label{line:queue2Del} \;
	 	 Delete the list $\cF_v$, and remove $v$ from $\ALG$ and the set of vertices. \;
	}
\end{algorithm}

\subsection{Correctness of the Algorithm}\label{sec:correctness}
We will now demonstrate the correctness of the algorithm, by first proving two Propositions. 

\begin{proposition}\label{prop:IS}
After every time step $t$, the set $\ALG$ stored by the algorithm is an independent set in $G^{(t)}$.
\end{proposition}
\begin{proof}
Suppose otherwise, and let $v,u \in \ALG$ be vertices with $(v,u) \in E$. WLOG we have that $v$ is the vertex which entered $\ALG$ most recently of the two. Then we had $ u \in ALG$ at the moment that $\ins(v,\pi(v))$ was most recently called. Since on the step that $\ins(v,\pi(v))$ was most recently called the vertex $p$ was added to $\ALG$, it must have been that $\min_{w \in N(v) \cap \ALG} \pi(w) > \pi(v)$, thus, in particular,  we must have had $\pi(v) < \pi(u)$. However, in this case we would have made $u$ a follower of $v$ at this step and set $\ell(u) = v$, which could not have occurred since then $u$ would have been removed from $\ALG$, which completes the proof. 
\end{proof}

\begin{claim}\label{claim:1}
We always have $|\ALG| \leq k+1$ at all time steps.
\end{claim}
\begin{proof}
After the first insertion the result is clear. We demonstrate that the claim holds inductively after each insertion. Only a call to $\ins(v)$ can increase the size of $\ALG$, so consider any such call. If $N(v) \cap \ALG$ is empty, then we add $v$ to $\ALG$, which can possibly increase its size to $k+2$ if it previously had $k+1$ elements. In this case, we remove the element with largest rank and add it to $\cQ$, maintaining the invariant. If there exists a $u^* \in N(v) \cap \ALG$ with smaller rank than $v$, we make $v$ the follower of the vertex in $ N(v) \cap \ALG$ with smallest rank, in which case the size of $\ALG$ is unaffected. In the final case, all points in $N(v) \cap \ALG$ have larger rank than $v$, in which case all of $N(v) \cap \ALG$ (which is non-empty) is made a follower of $v$ and removed from $\ALG$, thereby decreasing or not affecting the size of $\ALG$, which completes the proof. 
\end{proof}

\begin{proposition}[Correctness of Leader Mapping]
At any time step, if $\ALG \leq k$ then every vertex $v \in V$ is either contained in $\ALG$, or has a leader $\ell(v) \in \ALG$ with $(v,\ell(v)) \in E$.
\end{proposition}
\begin{proof}
After processing any update, we first claim that if $\ALG \leq k$ we have $\cQ = \emptyset$. This follows from the fact that the while loop in Line \ref{line:whileMain} of Algorithm \ref{alg:main} does not terminate until one of these two conditions fails to hold. Thus if $\ALG \leq k$, every vertex $v \in V$ is either contained in $\ALG$ (i.e., an active leader), or is a follower of such an active leader, which completes the proof of the proposition, after noting that we only set $\ell(u) = v$ when $(u,v) \in E$ is an edge.
\end{proof}

\begin{lemma}[Correctness of the top-$k$ LFMIS]\label{lem:correctness}
After every time step, we have $\ALG = \LFMIS_{k+1}(G,\pi)$. 
\end{lemma}
\begin{proof}
Order the points in $\ALG = (v_1,\dots,v_r)$ and $\LFMIS_{k+1}(G,\pi) = (u_1,\dots,u_s)$ by rank. We prove inductively that $v_i = u_i$. Firstly, note that $u_1$ is the vertex with minimal rank in $G$. As a result, $u_1$ could not be a follower of any point, since we only set $\ell(u) = v$ when $\pi(v) < \pi(u)$. Thus $u_1$ must either be an inactive leader (as it cannot be equal to $v_j$ for $j > 1$) or an unclustered point. In both cases, one has $u_1 \in \cQ$, which we argue cannot occur. To see this, note that at the end of processing the update, the while loop in  Line \ref{line:whileMain} of Algorithm \ref{alg:main} would necessarily remove $u_1$ from $\cQ$ and insert it. It follows that we must have $u_1 = v_1$.

In general, suppose we have $v_i = u_i$ for all $i \leq j$ for some integer $j < s$. We will prove $v_{j+1} = u_{j+1}$. First suppose $r , s \geq j+1$.
Now by definition of the LSFMIS, the vertex $u_{j+1}$ is the smallest ranked vertex in $V \setminus  \cup_{i \leq j} N(v_i) \cup \{u_i\}$. Since we only set $\ell(u) = v$ when $(u,v) \in E$ is an edge, it follows that $u_{j+1}$ cannot be a follower of $u_{i} = v_i$ for any $i \leq j$. Moreover, since we only set $\ell(u) = v$ when $\pi(v) < \pi(u)$, it follows that $u_{j+1}$ cannot be a follower of $v_{i}$ for any $i > j$, since $\pi(v_i) \geq \pi(u_{j+1})$ for all $i > j$. 
Thus, if $v_{j+1} \neq u_{j+1}$, it follows that either $u_{j+1} \in \cQ$, or $u_{j+1}$ is a follower of some vertex $u' \in \cQ$ with smaller rank than $u_{j+1}$. Then, similarly as above, in both cases the while loop in  Line \ref{line:whileMain} of Algorithm \ref{alg:main} would necessarily remove $u_{j+1}$ (or $u'$ in the latter case) from $\cQ$ and insert it, because $\pi(u_{j+1}) < \pi(v_r)$, and in the latter case if such a $u'$ existed we would have $\pi(u') < \pi(u_{j+1}) < \pi(u_r)$. We conclude that $v_{j+1} = u_{j+1}$.

The only remaining possibility is $r \neq s$. First, if $r > s$, by Claim \ref{claim:1} we have $r \leq k+1$, and by  Proposition \ref{prop:IS} $\ALG$ forms a independent set. Thus $v_1,v_2,\dots,v_s,v_{s+1}$ is an independent set, but since $v_i = u_i$ for $i \leq s$ and $\LFMIS_{k+1}(G,\pi) = \{u_1,\dots,u_s\}$ is a maximal independent set whenever $s \leq k$, this yields a contradiction. Finally, if $r < s$, consider the vertex $u_{r+1}$. Since $u_i = v_i$ for all $i \leq r$, $u_{r+1}$ cannot be a follower of $v_i$ for any $i \in [r]$. As a result, it must be that either $u_{r+1} \in \cQ$ or $u_{r+1}$ is a follower of a vertex in $\cQ$. In both cases, 
at the end of the last update, we had $|\ALG| = r \leq k$ and $\cQ \neq \emptyset$, which cannot occur as the while loop in  Line \ref{line:whileMain} of Algorithm \ref{alg:main} would not have terminated. It follows that $r=s$, which completes the proof.
\end{proof}

\subsection{Amortized Update Time Analysis}\label{sec:amortized}
We now demonstrate that the above algorithm runs in amortized $\tilde{O}(k)$-time per update. We begin by proving a structural result about the behavior of our algorithm. In what follows, let $G^t$ be the state of the graph \textit{after} the $t$-th update. Similarly, let $\ell_t(v) \in V \cup \{\bot\}$ be the value of $\ell(v)$ after the $t$-th update. 

\begin{proposition}\label{prop:orderedPi}
Let $\ins(v_1,\pi(v_1)),\ins(v_2,\pi(v_2)),\dots,\ins(v_r,\pi(v_r))$ be the ordered sequence of calls to the $\ins$ function (Algorithm \ref{alg:ins}) which take places during the processing of any individual update in the stream. Then we have $\pi(v_1) < \pi(v_2) < \cdots < \pi(v_r)$. As a corollary, for any vertex $v$ the function $\ins(v,\pi(v))$ is called at most once per time step. 
\end{proposition}
\begin{proof}
Assume $r>1$, since otherwise the claim is trivial. 
To prove the proposition it will suffice to show two facts: $(1)$ whenever a call to $\ins(v_i,\pi(v_i))$ is made, $\pi(v_i)$ is smaller than the rank of all vertices in the queue $\cQ$, and $(2)$ a call to $\ins(v_i,\pi(v_i))$ can only result in vertices with larger rank being added to $\cQ$. 

To prove $(1)$, note that after the first call to $\ins(v_1,\pi(v_1))$, which may have been triggered directly as a result of $v_1$ being added to the stream at that time step, all subsequent calls to $\ins$ can only be made via the while loop of Line \ref{line:whileMain} in Algorithm \ref{alg:main}, where the point with smallest rank is iteratively removed from $\cQ$ and inserted. Thus, fact $(1)$ trivially holds for all calls to $\ins$ made in this while loop, and it suffices to prove it for  $\ins(v_1,\pi(v_1))$ in the case that $v_1$ is added to the stream at the current update (if $v_1$ was added from the queue, the result is again clear). Now if $\cQ \neq \emptyset$ at the moment $\ins(v_1,\pi(v_1))$ is called, it must be the case that $|\ALG| = k+1$ and $\min_{w \in \cQ} \pi(w) > \max_{w \in \ALG} \pi(w)$ (otherwise the queue would have been emptied at the end of the prior update). Thus, if it were in fact the case that $\pi(v_1) >\min_{w \in \cQ} \pi(w)$, then we also have $\pi(v_1) > \max_{w \in \ALG} \pi(w)$, and therefore the call to $\ins(v_1,\pi(v_1))$ would result in inserting $v_1$ into $\cQ$ on Line \ref{line:firstIfIns} of Algorithm \ref{alg:ins}. Such an update does not modify $\ALG$, and does not change the fact that $\min_{w \in \cQ} \pi(w) > \max_{w \in \ALG} \pi(w)$, thus the processing of the update will terminate after the call to $\ins(v_1,\pi(v_1))$ (contradicting the assumption that $r>1$), which completes the proof of $(1)$.

To prove $(2)$, note that there are only three ways for a call to $\ins(v_i,\pi(v_i))$ to result in a vertex $u$ being added to $\cQ$. In the first case, if $\ell(u)$ was an active leader which was made a follower of $v_i$ as a result of $\ins(v_i,\pi(v_i))$, then we have $\pi(u) < \pi(\ell(u)) < \pi(v_i)$. Next, we could have had $\ell(u) = v_i$ (in the event that $v_i$ was an inactive leader being reinserted from $\cQ$), in which case $\pi(u) < \pi(v_i)$. Finally, it could be the case that $u$ was the active leader in $\ALG$ prior to the call to $\ins(v_i,\pi(v_i))$, and was then removed from $\ALG$ as a result of the size of $\ALG$ exceeding $k+1$ and $u$ having the largest rank in $\ALG$. This can only occur if $v_i$ was added to $\ALG$ and had smaller rank than $u$, which completes the proof of $(2)$.

Since by $(1)$ every time $\ins(v_i,\pi(v_i))$ is called $\pi(v_i)$ is smaller than the rank of all vertices in the queue, and by $(2)$ the rank of all new vertices added to the queue as a result of $\ins(v_i,\pi(v_i))$ will continue to be larger than $\pi(v_i)$, it follows that $v_{i+1}$, which by construction must be the vertex with smallest rank in $\cQ$ after the call to $\ins(v_i,\pi(v_i))$, must have strictly larger rank than $v_i$, which completes the proof of the proposition. 
\end{proof}

The following proposition is more or less immediate. It implies, in particular, that a point can only be added to $\cQ$ once per time step (similarly, $v$ can be removed from $\cQ$ once per time step). 

\begin{proposition}\label{prop:afterQ}
Whenever a vertex $v$ in the queue $\cQ$ is removed and $\ins(v,\pi(v))$ is called, the vertex $v$ either becomes a follower of an active leader, or an active leader itself. 
\end{proposition}
\begin{proof}
If $v$ shares an edge with a vertex in $\ALG$ with smaller rank, it becomes a follower of such a vertex. Otherwise, all vertices in $N(v) \cap \ALG$ become followers of $v$, and $v$ becomes an active leader by construction (possibly resulting in an active leader of larger rank to be removed from $\ALG$ as a result of it no longer being contained in $\LFMIS_{k+1}$).
\end{proof}

Equipped with the prior structural propositions, our approach for bounding the amortized update time is to first observe that, on any time step $t$, our algorithm only attempts to insert a vertex~$v$, thereby spending $O(k)$ time to search for edges between $v$ and all members of $\ALG$, if either $v$ was the actual vertex added to the stream on step $t$, or when $v$ was added to $\cQ$ on step $t$ or before. Thus, it will suffice to bound the total number of vertices which are ever added into $\cQ$ --- by paying a cost $O(k + \log n)$ for each vertex $v$ which is added to the queue, we can afford both the initial $O(\log n)$ cost of adding it to the priority queue, as well as the $O(k)$ runtime cost of possibly later reinserting $v$ during the while loop in Line \ref{line:whileMain} of Algorithm \ref{alg:main}. We formalize this in the following proposition.

\begin{proposition}\label{prop:T}
Let $T$ be the total number of times that a vertex is inserted into the queue $\cQ$ over the entire execution of the algorithm, where two insertions of the same vertex $v$ on separate time steps are counted as distinct insertions.  Then the total runtime of the algorithm, over a sequence of $M$ insertions and deletions, is at most $O(T (k + \log n) + Mk)$, where $n$ is the maximum number of vertices active at any given time.
\end{proposition}
\begin{proof}
Note that the only actions taken by the algorithm consist of adding and removing vertices $v$ from $\cQ$ (modifying the value of $\ell(v)$ in the process, and possibly deleting $\cF_v$), and computing $\ALG \cap N(v)$ for some vertex $v$. The latter requires $O(k)$ time since we have $|\ALG| \leq k+1$ at all time steps. Given $O(\log n)$ time to insert or query from a priority queue with at most $n$ items, we have that $O(T \log n)$ upper bounds the cost of all insertions and deletions of points to $\cQ$. Moreover, all calls to compute $\ALG \cap N(v)$ for some vertex $v$ either occur when $v$ is the vertex added to the stream on that time step (of which there is at most one), or when $v$ is inserted after previously having been in $\cQ$. By paying each vertex $v$ a sum of $O(k)$ when it is added to $\cQ$, and paying $O(k)$ to each vertex when it is first added to the stream, it can afford the cost of later computing $\ALG \cap N(v)$  when it is removed. This results in a total cost of $O(T(k+\log n) + Mk)$, which completes the proof. 
\end{proof}

In what follows, we focus on bounding the quantity $T$. To accomplish this, observe that a vertex $v$ can be added to $\cQ$ on a given time step $t$ in one of three ways:
\begin{figure}[H]
    \centering
  
\begin{Frame}[Scenarios where $v$ is added to $\cQ$]
    \begin{enumerate}
\item The vertex $v$ was added in the stream on time step $t$. In this case, $v$ is added to $\cQ$ when the if statement on Line \ref{line:firstIfIns} of Algorithm \ref{alg:ins} executes. 
\item The vertex $v$ is added to $\cQ$ when it was previously led by $\ell_{t-1}(v) \in V$, and either $\ell_{t-1}(v)$ becomes a follower of another leader during time step $t$ (resulting in $v$ being added to $\cQ$), or $\ell_{t-1}(v)$ is deleted. This can occur in either Lines \ref{line:queue1Ins} or \ref{line:queue2Ins}  of Algorithm \ref{alg:ins} for the first case, or in Lines \ref{line:queue1Del} or \ref{line:queue2Del} of Algorithm \ref{alg:del} in the case of  $\ell_{t-1}(v)$  being deleted. 
\item The vertex $v$ was previously in $\LFMIS_{k+1}$, and subsequently left $\LFMIS_{k+1}$ because $|\LFMIS_{k+1}| = k+1$ and a new vertex $u$ was added to $\LFMIS_{k+1}$ with smaller rank. This occurs in Line \ref{line:delOverflow} of Algorithm \ref{alg:ins}. 
    \end{enumerate}
\end{Frame}

\end{figure}

Obviously, the first case can occur at most once per stream update, so we will focus on bounding the latter two types of additions to $\cQ$. For any step $t$, define $\cA_\pi^t$ to be the number of vertices that are added to $\cQ$ as a result of the second form of insertions above. Namely, $\cA_\pi^t = |\{v \in G^{t} : \ell_{t-1}(v) \in V, \text{ and } \ell_t(v) \neq \ell_{t-1}(v) \}|$. Next, define $\cB_\pi^t$ to be the number of leaders which were removed from $\ALG$ Line \ref{line:delOverflow} of Algorithm \ref{alg:ins} (i.e., insertions into $\cQ$ of the third kind above). Letting $T$ be as in Proposition \ref{prop:T}, we have $T \leq  M+  \sum_t \cA_\pi^t + \cB_\pi^t $.

To handle $\sum_t \cA_\pi^t$ and $\sum_t \cB_\pi^t$, we demonstrate that each quantity can be bounded by the total number of times that the \textit{eliminator} of a vertex changes. Recall from Section \ref{sec:kCenters} that, given a graph $G=(V,E)$, $v \in V$, and ranking $\pi:V \to [0,1]$, the eliminator of $v$, denoted $\elim_{G,\pi}(v)$, is defined as the vertex of smallest rank in the set $(N(v) \cup \{v\} ) \cap \LFMIS(G,\pi)$. Now define $\cC_\pi^t$ to be the number of vertices whose eliminator changes after time step $t$. Formally, for any two graphs $G,G'$ differing in at most once vertex, we define 
$\cC_\pi(G,G') = \{v \in V | \elim_{G,\pi}(v) \neq \elim_{G',\pi}(v)\}$, and set $\cC_\pi^t = |\cC_\pi(G^{t-1},G^{t})|$. %
We now demonstrate that $\sum_t \cC_\pi^t$ deterministically upper bounds both $\sum_t \cA_\pi^t$ and $\sum_t \cB_\pi^t$.

\begin{lemma}\label{lem:main}
Fix any ranking $\pi: V \to [0,1]$. Then we have $\sum_t \cA_\pi^t \leq 5 \sum_t \cC_\pi^t$, and moreover for any time step $t$ we have $\cB_\pi^t \leq \cC_\pi^t$.
\end{lemma}
\begin{proof}
We first prove the second statement. Fix any time step $t$, and let $v_1,\dots,v_r$ be the $r = |\cB_\pi^t|$ vertices which were removed from $\ALG$, ordered by the order in which they were removed from $\ALG$. For this to occur, we must have inserted at least $r$ vertices $u_1,\dots,u_r$ into $\ALG$ which were not previously in $\ALG$ on the prior step; in fact, Line \ref{line:delOverflow} of Algorithm \ref{alg:ins} induces a unique mapping from each $v_i$ to the vertex $u_i$ which forced it out of $\ALG$ during a call to $\ins(u_i,\pi(u_i))$. Note that, under this association, we have $\pi(u_i) < \pi(v_i)$ for each $i$. We claim that the eliminator of each such $u_i$ changed on time step $t$.

Now note that $\{u_1,\dots,u_r\}$ and $\{v_1,\dots,v_r\}$ are disjoint, since $u_i$ was inserted before $u_{i+1}$ during time step $t$ by the definition of the ordering, and so $\pi(u_1) < \pi(u_2) < \dots < \pi(u_r)$ by Proposition \ref{prop:orderedPi}, so no $u_i$ could be later kicked out of $\ALG$ by some $u_j$ with $j > i$. It follows that none of $u_1,\dots,u_r$ were contained in $\LFMIS_{k+1}(G^{t-1},\pi)$, but they are all in $\LFMIS_{k+1}(G^{t},\pi)$. Now note that it could not have been the case that $u_i \in \LFMIS(G^{t-1},\pi)$, since we had $v_i \in \LFMIS_{k+1}(G^{t-1},\pi)$ but $\pi(u_i) < \pi(v_i)$. Thus $u_i \notin \LFMIS(G^{t-1},\pi)$, and therefore the eliminator of $u_i$ changed on step $t$ from $\elim_{G^{t-1},\pi}(u_i) \neq u_i$ to $\elim_{G^t,\pi}(u_i) = u_i$, which completes the proof of the second statement.

We now prove the first claim that $\sum_t \cA_\pi^t \leq 5 \sum_t \cC_\pi^t$.  Because a vertex can be inserted into $\cQ$ at most once per time step (due to Proposition \ref{prop:afterQ}), each insertion $\cQ$ which contributes to $\sum_t \cA_\pi^t$ can be described as a  vertex-time step pair $(v,t)$, where we have $\ell_{t}(v) \neq \ell_{t-1}(v) \in V$ because either $\ell_{t-1}(v)$ became a follower of a vertex in $\LFMIS_{k+1}(G^t,\pi)$, or because $\ell_{t-1}(v)$ was deleted on time step $t$. %
We will now need two technical claims.

\begin{claim}\label{claim:2}
Consider any vertex-time step pair $(v,t)$ where $\ell_t(v) \in V$ and $\ell_{t-1}(v) \neq \ell_t(v)$. In other words, $v$ was made a follower of some vertex $\ell_t(v)$ during time step $t$. Then $\ell_t(v) = \elim_{G^t,\pi}(v)$.
\end{claim}
\begin{proof}
First note that the two statements of the claim are equivalent, since if $\ell(v)$ is set to $u \in V$ during time step $t$, then by Proposition \ref{prop:orderedPi} we have that $\ell(v)$ is not modified again during the processing of update $t$, so $u = \ell_{t}(v)$. Now the algorithm would only set $\ell_t(v) = u$ in one of two cases. In this first case, it occurs during a call to $\ins(v,\pi(v))$, in which case $\ell_t(v)$ is set to the vertex with smallest rank in $\LFMIS_{k+1}(G^t,\pi)\cap N(v)$, which by definition is $\elim_{G^t,\pi}(v)$. In the second case, $v$ was previously in $\LFMIS_{k+1}(G^{t-1},\pi)$, and $\ell(v)$ was changed to a vertex $w$ during a call to  $\ins(w,\pi(w))$, where $w \in N(v)$ and $\pi(w) < \pi(v)$. Since prior to this insertion $v$ was not a neighbor of any point in $\LFMIS_{k+1}(G^{t-1},\pi)$, and since by Proposition \ref{prop:orderedPi} the $\ins$ function will not be called again on time $t$ for a vertex with rank smaller than $w$, it follows that $w$ has the minimum rank of all neighbors of $v$ in $\LFMIS(G^{t},\pi)$, which completes the claim. 
\end{proof}

\begin{claim}\label{claim:3}
Consider any vertex-time step pair $(v,t)$ where $\ell_t(v) = \bot$ and $\ell_{t-1}(v) = \elim_{G^{t-1},\pi}(v) $. Then $\elim_{G^{t-1},\pi}(v) \neq \elim_{G^{t},\pi}(v) $.
\end{claim}
\begin{proof}
If $\ell_{t-1}(v) = \elim_{G^{t-1},\pi}(v)$, then $\elim_{G^{t-1},\pi}(v) \in v$  and $\ell_t(v)$ is changed to $\bot$ during time step $t$, then as in the prior claim, this can only occur if $\ell_{t-1}(v)$ is made a follower of another point in $\LFMIS_{k+1}(G^t,\pi)$, or if $\ell_{t-1}(v)$ is deleted on that time step. In both cases we have $\ell_{t-1}(v) \notin \LFMIS(G^t,\pi)$. Since $\elim_{G^{t},\pi}(v)$ is always in $\LFMIS(G^t,\pi)$, the claim follows.  
\end{proof}

Now fix any vertex $v$, and let $\sigma_1,\dots,\sigma_M$ be the sequence of eliminators of $v$, namely $\sigma_t = \elim_{G^t,\pi}(v)$ (note that $\sigma_t$ is either a vertex in $V$ or $\sigma_t = \emptyset$). Similarly define $\lambda_1,\dots,\lambda_M$ by $\lambda_t = \ell_t(v)$, and note that $\lambda_t \in V \cup \{\bot\}$. To summarize the prior two claims: each time $\lambda_{t-1} \neq \lambda_t$ and $\lambda_t \in V$, we have $\sigma_t = \lambda_t$; namely, the sequences become aligned at time step $t$. Moreover, whenever the two sequences are aligned at some time step $t$, namely $\sigma_t = \lambda_t$, and subsequently $\lambda_{t+1} = \bot$, we have that $\sigma_{t+1} \neq \sigma_t$. We now prove that every five subsequent changes in the value of $\lambda$ cause at least one unique change in $\sigma$.

To see this, let $t_1 < t_2 < t_3$ be three subsequent changes, so that $\lambda_{t_1} \neq \lambda_{t_1 - 1}$, $\lambda_{t_2} \neq \lambda_{t_2 - 1}$, $\lambda_{t_3} \neq \lambda_{t_3 - 1}$, and $\lambda_i$ does not change for all $i=t_1,\dots,t_2-1$ and $i= t_2,\dots,t_3-1$. First, if $\lambda_{t_1},\lambda_{t_2} \in V$, by Claim \ref{claim:2} we have $\sigma_{t_1} = \lambda_{t_1}$ and  $\sigma_{t_2} = \lambda_{t_2}$, and thus $\sigma_{t_1} \neq \sigma_{t_2}$, so $\sigma$ changes in the interval $[t_1,t_2]$. If $\lambda_{t_1} \in V$ and $\lambda_{t_2} = \bot$, we have $\sigma_{t_1} = \lambda_{t_1}$, and so if $\sigma$ does not change by time $t_2-1$ we have $\sigma_{t_2-1} = \lambda_{t_2-1}$, and thus $\sigma_{t_2} \neq \sigma_{t_2-1}$ by Claim \ref{claim:3}, so $\sigma$ changes in the interval $[t_1,t_2]$. Finally, if $\lambda_{t_1} = \bot$, then we must have $\lambda_{t_2} \in V$, and so $\lambda_{t_3} = \bot$. Then by the prior argument, $\sigma$ must change in the interval $[t_2,t_3]$. Thus, in each case, $\sigma$ must change in the interval $[t_1,t_3]$. To avoid double counting changes which occur on the boundary, letting $t_1,\dots,t_r$ be the sequence of all changes in $\lambda$, it follows that there is at least one change in $\sigma$ in each of the disjoint intervals $(t_{5i+1}, t_{5(i+1)})$ for $i=0,1,2,\dots,\lfloor r/5 \rfloor$. It follows that $\sum_t \cA_\pi^t \leq 5 \sum_t \cC_\pi^t$, which completes the proof. 
\end{proof}

The following theorem, due to \cite{behnezhad2019fully}, bounds the expected number of changes of eliminators which occur when a vertex is entirely removed or added to a graph.

\begin{theorem}[Theorem 3 of \cite{behnezhad2019fully}]\label{thm:elimBound}
Let $G = (V,E)$ be any graph on $n$ vertices, and let $G' = (V',E')$ be obtained by removing a single vertex from $V$ along with all incident edges. Let $\pi: V \to [0,1]$ be a random mapping. Let $\cC_\pi(G,G') = \{v \in V | \elim_{G,\pi}(v) \neq \elim_{G',\pi}(v)\}$. Then we have $\mathbb{E}_\pi \left[|\cC_\pi(G,G')|\right] = O(\log n)$.
\end{theorem}

\begin{theorem} \label{thm:LFMISMain} There is a algorithm which, on a fully dynamic stream of insertions and deletions of vertices to a graph $G$, maintains at all time steps a top-$k$ LFMIS of $G$ with leaders (Definition \ref{def:LFMISLead}) under a random ranking $\pi: V \to [0,1]$. The expected amortized per-update time of the algorithm is $O(k \log n + \log^2 n)$, where $n$ is the maximum number active of vertices at any time. Moreover, the algorithm does not need to know $n$ in advance. 
\end{theorem}
\begin{proof}
By the above discussion, letting $T$ be as in Proposition \ref{prop:T}, we have $T \leq  M+  \sum_t \cA_\pi^t + \cB_\pi^t $. By the same proposition, the total update time of the algorithm over a sequence of $n$ updates is at most $O(T(k + \log n) + kM)$.  By Lemma \ref{lem:main}, we have $\sum_t \cA_\pi^t + \cB_\pi^t  \leq 6 \sum_t \cC^t_\pi$, and by Theorem \ref{thm:elimBound} we have $\mathbb{E}_\pi \left[\sum_t \cC^t_\pi\right] = O(M \log n)$. It follows that $\ex{T} = O(M \log n)$, therefore the total update time is $O(kM \log n + M \log^2 n)$, which completes the proof. 
\end{proof}

\section{Fully Dynamic $k$-Centers via Locally Sensitive Hashing}\label{sec:LSH}

In this section, we demonstrate how the algorithm for general metric spaces of Section \ref{sec:generalMetric} can be improved to run in \textit{sublinear} in $n$ amortized update time, even when $k = \Theta(n)$, if the metric in question admits good locally sensitive hash functions (introduced below in Section \ref{sec:LSHAlg}). Roughly speaking, a locally sensitive hash function is a mapping $h: \cX \to U$, for an universe $U$, which has the property that points which are close in the metric should collide, and points which are far should not. Thus, when searching for points which are close to a given $x \in \cX$, one can first apply a LSH to quickly prune far points, and search only through the points in the hash bucket $h(x)$. We will use this approach to speed up the algorithm from Section \ref{sec:generalMetric}. 

There are several serious challenges when using an arbitrary ANN data structure to answer nearest-neighbor queries for our $k$-center algorithm. Firstly, the algorithm \textit{adaptively} queries the ANN data structure: the points which are inserted into $\LFMIS_{k+1}$, as well as the future edges which are reported by the data structure, depend on the prior edges which were returned by the data structure. Such adaptive reuse breaks down traditional guarantees of randomized algorithms, hence designing such algorithms which are robust to adaptivity is the subject of a growing body of research \cite{ben2020framework,cherapanamjeri2020adaptive,HassidimKMMS20,WoodruffZ20,ACSS21}. More nefariously, the adaptivity also goes in the other direction: namely, the random ordering $\pi$ influences which points will be added to the set $\LFMIS_{k+1}$, in turn influencing the future queries made to the ANN data structure, which in turn dictate the edges which exist in the graph (by means of queries to the ANN oracle). Thus, the graph itself cannot be assumed to be independent of $\pi$! However, we show that we can leverage LSH functions to define the graph independent of $\pi$.

\paragraph{Summary of the LSH-Based Algorithm. }
We now describe the approach of our algorithm for LSH spaces.
Specifically, first note that the factor of $k$ in the amortized update time in Proposition \ref{prop:T} comes from the time required to compute $S = \ALG \cap N(v)$ in Algorithm \ref{alg:ins}. However, to determine which of the three cases we are in for the execution of Algorithm \ref{alg:ins}, we only need to be given the value $u^* = \arg \min_{u' \in S} \pi(u')$ of the vertex in $S$ with smallest rank, or $S = \emptyset$ if none exists. If $S = \emptyset$, then the remainder of Algorithm \ref{alg:ins} runs in constant time. If $\pi(u^*) < \pi(v)$, where $v$ is the query point, then the remaining run-time is constant unless $v$ was a leader, in which case it is proportional to the number of followers of $v$, each of which are inserted into $\cQ$ at that time. Lastly, if $\pi(u^*) > \pi(v)$, then we search through each $u \in S$, make $u$ a follower of $v$, and add the followers of $u$ to $\cQ$. If $u$ was previously in $\ALG$, it must have also been in $\LFMIS_{k+1}$ on the prior time step. Thus it follows that each such $u \in S$ changes its eliminator on this step.

In summary, after the computation of $S = \ALG \cap N(v)$, the remaining runtime is bounded by the sum of the number of points added to $\cQ$, and the number of points that change their eliminator on that step. Since, ultimately, the approach in Section \ref{sec:amortized} was to bound $T$ by the total number of times a point's eliminator changes, our goal will be to obtain a more efficient data structure for returning $S = \ALG \cap N(v)$. Specifically, if after a small upfront runtime $R$, such a data structure can read off the entries of $S$ in the order of the rank, each in constant time, one could therefore replace the factor of $k$ at both parts of the sum in the running time bound of Proposition \ref{prop:T} by $R$. We begin by formalizing the guarantee that such a data structure should have. 

We will demonstrate that approximate nearest neighbor search algorithms based on locally sensitive hashing can be modified to have the above properties. However, since such an algorithm will only be approximate, it will sometimes returns points in $S$ which are farther than distance $r$ from $v$, where $r$ is the threshold. Thus, the resulting graph defined by the locally sensitive hashing procedure will now be an \textit{approximate threshold graph:}

\begin{definition}\label{def:approxThreshGraph}
Fix a point set $P$ from a metric space $(\cX,d)$, and real values $r>0$ and $c\geq 1$. A $(r,c,L)$-approximate threshold graph $G_{r,c} = (V(G_{r,c}),E(G_{r,c}))$ for $P$ is any graph with $V(G_{r,c}) = P$, and whose whose edges satisfy $E(G_r) \subseteq E(G_{r,c})$ and $|E(G_{r,c}) \setminus E(G_{cr})| \leq L$, where $E(G_r), E(G_{cr})$ are the edge set of the threshold graphs $G_r,G_{cr}$ respectively.

\end{definition}

If $L = 0$, it is straightforward to see that an algorithm for solving the top-$k$ LFMIS problem on a $(r,c,0)$-approximate threshold graph $G_{r,c}$ can be used to obtain a $c(2+\eps)$ approximation to $k$-center. When $L>0$, an algorithm can first check, for each edge $e \in E(G_{r,c})$ it considers, whether $e \in G_{rc}$, and discard it if it is not the case. We will see that the runtime of handling a $(r,c,L)$-approximate threshold graph will depend linearly on $L$. Moreover, we will set parameters so that $L$ is a constant in expectation.

\subsection{Locally Sensitive Hashing and the LSH Algorithm}\label{sec:LSHAlg}

We begin by introducing the standard definition of a locally sensitive hash family for a metric space \cite{indyk1998approximate}.

\begin{definition}[Locally sensitive hashing \cite{indyk1998approximate,har2012approximate}]\label{def:LSH}
Let $\cX$ be a metric space, let $U$ be a range space, and let $r \geq 0$ $c \geq 1$ and $0 \leq p_2 \leq p_1  \leq 1$ be reals.
A family $\cH = \{h : \cX \to U\}$ is called $(r,cr,p_1,p_2)$-sensitive if for any $q,p \in \cX$:

\begin{itemize}
    \item If $d(p,q) \leq r$, then $\prb{\cH}{h(q) = h(p)} \geq p_1$. 
    \item If $d(p,q) > cr$, then $\prb{\cH}{h(q) = h(p)} \leq p_2$. 
\end{itemize}
\end{definition}

Given a $(r,cr,p_1,p_2)$-sensitive family $\cH$, we can define $\cH^t$ to be the set of all functions $h^t: \cX \to U^t$ defined by $h^t(x) = (h_1(x),h_2(x),\dots,h_t(x))$, where  $h_1,\dots,h_k \in \cH$. In other words, a random function from $\cH^t$ is obtained by drawing $t$ independent hash functions from $\cH$ and concatenating the results. It is easy to see that the resulting hash family $\cH^t$ is $(r,cr,p_1^t,p_2^t)$-sensitive. 
We now demonstrate how a $(r,c)$-approximate threshold graph can be defined via a locally sensitive hash function. 

\begin{definition}\label{def:inducedGraph}
Fix a metric space $(\cX,d)$ and a finite point set $P \subset \cX$, as well as integers $t,s \geq 1$. Let $\cH: \cX \to U$ be a $(r,cr,p_1,p_2)$-sensitive family. Then a graph $G_{r,cr}(P,\cH,t,s) = ( V,E)$ induced by $\cH$ is a random graph which is generated via the following procedure. First, one randomly selects hash functions $h_1,h_2,\dots,h_s \sim \cH^t$. Then the vertex set is given by $V=P$, and then edges are defined via $(x,y) \in E$ if and only if $h_i(x) = h_i(y)$ for some $i \in [s]$. 
\end{definition}

We now demonstrate that, if $\cH$ is a sufficiently sensitive hash family, the random graph $G_{r,cr}(P,\cH,t,s)$ constitutes a $(r,c,L)$-approximate threshold graph with good probability, where $L$ is a constant in expectation. 

\begin{proposition}\label{prop:isAThreshold}
Fix a metric space $(\cX,d)$ and a point set $P \subset \cX$ of size $|P|= n$, and let $\cH: \cX \to U$ be a $(r,cr,p_1,p_2)$-sensitive family. Fix any $\delta \in (0,\frac{1}{2})$. Set $s = \ln(n^2/\delta) n^{2\rho}/p_1 $, where $\rho = \ln \frac{1}{p_1}/\ln \frac{1}{p_2}$, and  $t = \lceil 2 \log_{1/p_2} n  \rceil$. Then, with probability at least $1-\delta$, the the random graph $G_{r,cr}(P,\cH,t,s)$ is a $(r,c,L)$-approximate threshold graph, where $L$ is a random variable satisfying $\ex{L} < 2$ (Definition \ref{def:approxThreshGraph}). 
\end{proposition}
\begin{proof}

First, fix any $x,y \in P$ such that $d(x,y) > c r$. We have that $\prb{h \sim \cH^k}{h(x) = h(y)} \leq p_2^t < \frac{1}{ n^2}$. It follows that 
\[\exx{h_1,\dots,h_s \sim \cH^t}{|E(G_{r,cr}(P,\cH,t,s) \setminus E(G_{r })|} \leq \sum_{(x,y) \in P^2} \frac{1}{n^2} < 1\]
Namely, we have $\ex{L} < 1$ where $L = |E(G_{r,cr}(P,\cH,t,s) \setminus E(G_{r })|$.
Next, fix any $(x,y) \in E(G_r)$. We have 
\[\prb{h \sim \cH^k}{h(x) = h(y)} \geq p_1^t > p_1^{2 \log_{1/p_2} n  + 1}= p_1 (n^2)^{-\rho}\] Thus, the probability that at least one $h_i$ satisfies $h_i(x) = h_i(y)$ is at least
\[1-(1- p_1 n^{-2\rho})^s >1-(1/e)^{\ln(n^2/\delta)} =  1- \delta/n^2\] 
After a union bound over all such  possible pairs, it follows that $(x,y) \in E(G_{r,cr}(P,\cH,t,s)) $ for all $(x,y) \in E(G_r)$ with probability at least $1-\delta$. Note that since $\delta < 1/2$ and $L$ is a non-negative random variable, it follows that conditioning on the prior event can increase the expectation of $L$ by at most a factor of $2$, which completes the proof. 
\end{proof}

We will now describe a data structure which allows us to maintain a subset $\cL$ of vertices of the point set $P$, and quickly answer queries for neighboring edges of a vertex $v$ in the graph $G$ defined by the intersection of $G_{r,cr}(P,\cH,t,s)$ and $G_{cr}$. Note that if $G_{r,cr}(P,\cH,t,s)$ is a $(r,c,L)$-approximate threshold graph, then this intersection graph $G$ satisfies $G_r \subseteq G \subseteq G_{cr}$, and is therefore a  $(r,c,0)$-approximate threshold graph. It is precisely this graph $G$ which we will run our algorithm for top-$k$ LFMIS on. However, in addition to finding all neighbors of $v$ in $\cL$, we will also need to quickly return the neighbor with smallest rank $\pi$, where $\pi:V \to [0,1]$ is a random ranking as in Section \ref{sec:generalMetric}. Roughly, the data structure will hash all points in $\cL$ into the hash buckets given by $h_1,\dots,h_s$, and maintain each hash bucket via a binary search tree of depth $O(\log n)$, where the ordering is based on the ranking $\pi$. 

For the following Lemma and Theorem, we fix a metric space $(\cX,d)$ and a point set $P \subset \cX$ of size $|P|= n$, as well as a scale $r > 0$, and approximation factor $c$. Moreover, let $\cH: \cX \to U$ be a $(r,cr,p_1,p_2)$-sensitive hash family for the metric space $\cX$, and let $\Run(\cH)$ be the time required to evaluate a hash function $h \in \cH$. Furthermore, let $\pi: P \to [0,1]$ be any ranking over the points $P$, such that $\pi(x)$ is truncated to $O(\log n)$ bits, and such that $\pi(x) \neq \pi(y)$ for any distinct $x,y \in P$ (after truncation). Note that the latter holds with probability $1-1/\poly(n)$ if $\pi$ is chosen uniformly at random. Lastly, for a graph $G = (V(G),E(G))$, let $N_G(v)$ be the neighborhood of $v$ in $G$ (in order to avoid confusion when multiple graphs are present). 
\begin{lemma}\label{lem:binaryDataStructure}
Let $G_{r,cr}(P,\cH,t,s)$ be a draw of the random graph as in Definition \ref{def:inducedGraph}, where $r,c,P,\cH$ are as above, such that $G_{r,cr}(P,\cH,t,s)$ is a $(r,c,L)$-approximate threshold graph. Let $G = G_{r,cr}(P,\cH,t,s)\cap G_{cr}$. Then there is a fully dynamic data structure, which maintains a subset $\cL \subset V(G)$ and can perform the following operations:
\begin{itemize}
    \item $\ttx{Insert}(v)$: inserts a vertex $v$ into $\cL$ in time $O(s  \log n + ts \Run(\cH))$
    \item $\ttx{Delete}(v)$: deletes a vertex $v$ from $\cL$ in time $O(s \log n + ts \Run(\cH))$
    \item $\ttx{Query-Top}(v)$: returns $u^* = \arg \min_{u \in N_G \cap \cL} \pi(u)$, or \textsc{Empty} if $N_G(v) \cap \cL = \emptyset$, in time $O(sL \log n + ts \Run(\cH)) $
    \item $\ttx{Query-All}(v)$: returns the set $N_G(v) \cap \cL$, running in time $O(s(|N_G(v) \cap \cL| + L) \log n +  ts \Run(\cH))$. 
\end{itemize}
Moreover, given the hash functions $h_1,\dots,h_s \in \cH^t$ which define the graph $G_{r,cr}(P,\cH,t,s)$, the algorithm is deterministic, and therefore correct even against an adaptive adversary. 
\end{lemma}
\begin{proof}
For each $i \in [s]$ and hash bucket $b \in U$, we store a binary tree $T_{i,b}$ with depth at most $O(\log n)$, such that each node $z$ corresponds to an interval $[a,b] \subset [0,1]$, and the left and right children of $z$ correspond to the intervals $[a,(a+b)/2]$ and $[(a+b/2),b]$, respectively. Moreover, each node $z$ maintains a counter for the number of points in $\cL$ which are stored in its subtree. Given a vertex $v$ with rank $\pi(v)$, one can then insert $v$ into the unique leaf of $T_{i,b}$ corresponding to the $O(\log n)$ bit value $\pi(v)$ in time $O(\log n)$. Note that, since the keys $\pi(v)$ are unique, each leaf contains at most one vertex. Similarly, one can remove and search for a vertex from $T_{i,b}$ in time $O(\log n)$.

When processing any query for an input vertex $v \in P$, one first evaluates all $ts$ hash functions required to compute $h_1(v),\dots,h_s(v)$, which requires  $ts \Run(\cH)$ time. For insertions and deletions, one can insert $v$ from each of the $s$ resulting trees in time $O(\log n)$ per tree, which yields the bounds for $\ttx{Insert}(v)$ and $\ttx{Delete}(v)$. For $\ttx{Query-Top}(v)$, for each $i \in [s]$, one performs an in-order traversal of $T_{i,h_i(v)}$, ignoring nodes without any points stored in their subtree, and returns the first leaf corresponding to a vertex $u \in N_G(v)$, or \textsc{Empty} if all points in the tree are examined before finding such a neighbor. Each subsequent non-empty leaf in the traversal can be obtained in $O(\log n)$ time, and since by definition of a $(r,c,L)$-approximate threshold graph, $h_i(v) = h_i(u')$ for at most $u'$ vertices with $(v,u') \notin G_{cr}$, it follows that one must examine at most $L$ vertices $u'$ with $(v,u') \notin G_{cr}$ before one finds  $u^* = \arg \min_{u \in N_G \cap \cL} \pi(u)$ (or exhausts all points in the tree). Thus, the runtime is $O(L \log n)$ to search through each of the $s$ hash functions, which results in the desired bounds. 

Finally, for $\ttx{Query-All}(v)$, one performs the same search as above, but instead completes the full in-order traversal of each tree $T_{i,h_i(v)}$. By the  $(r,c,L)$-approximate threshold graph property, each tree $T_{i,h_i(v)}$ contains at most $|N_G(v) \cap \cL| + L$ vertices from $\cL$, after which the runtime follows by the argument in the prior paragraph. 
\end{proof}

Given the data structure from Lemma \ref{lem:binaryDataStructure}, we will now demonstrate how the algorithm from Section \ref{sec:generalMetric} can be implemented in sublinear in $k$ time, given a sufficently good LSH function for the metric. The following theorem summarizes the main consequences of this implementation, assuming the graph $G_{r,cr}(P,\cH,t,s)$ is a $(r,c,L)$-approximate threshold graph. 

\begin{theorem}\label{thm:LSHPreMain}
 Let $G_{r,cr}(P,\cH,t,s)$ be a draw of the random graph as in Definition \ref{def:inducedGraph}, where $r,c,P,\cH$ are as above, such that $G_{r,cr}(P,\cH,t,s)$ is a $(r,c,L)$-approximate threshold graph. Let $G = (V,E)$ be the graph with $V = P$ and $E  = E(G_{r,cr}(P,\cH,t,s))\cap E(G_{cr})$. Then there is a fully dynamic data structure which, under a sequence of vertex insertions and deletions from $G$, maintains a top-$n$ LFMIS with leaders (Definition \ref{def:LFMISLead}) of $G$ at all time steps. The expected amortized per-update runtime of the algorithm is $O( ( sL \log n + ts \Run(\cH)) \log n + \log^2 n )$, where the expectation is taken over the choice of $\pi$.
\end{theorem}
\begin{proof}
The algorithm is straightforward: we run the fully dynamic algorithm for top-$k$ LFMIS with leaders from Section \ref{sec:generalMetric}, however we utilize the data structure from Lemma \ref{lem:binaryDataStructure} to compute $S = \cL \cap N_G(v)$ in Algorithm \ref{alg:ins} (where $\cL = \cL_{n+1}$), as well as handle deletions from $\cL$ in Algorithm \ref{alg:del}. Note that to handle a call to $\ins(v)$ of Algorithm \ref{alg:ins}, one first calls  $\ttx{Query-Top}(v)$ in data structure from Lemma \ref{lem:binaryDataStructure}. If the result is $\emptyset$, or $u$ with $\pi(u) < \pi(v)$, then one can proceed as in Algorithm \ref{alg:ins} but by updating $\cL$ via Lemma \ref{lem:binaryDataStructure}. If the result is $u$ with $\pi(u) > \pi(v)$, one then calls  $\ttx{Query-All}(v)$ to obtain the entire set $S = N_G(v) \cap \cL$, each of which will subsequently be made a follower of $v$. 

Let $T$ be the total number of times that a vertex is inserted into the queue $\cQ$ over the entire execution of the algorithm (as in Proposition \ref{prop:T}), and as in 
Section \ref{sec:amortized}, we let $\cC^t_\pi$ denote the number of vertices whose eliminator changed after the $t$-th time step. We first prove the following claim, which is analogous to  Proposition \ref{prop:T}. 

\begin{claim}
The total runtime of the algorithm, over a sequence of $M$ insertions and deletions of vertices from $G$, is at most $O(T(\lambda + \log n) + \lambda (M+ \sum_{t \in [M}) \cC^t_\pi))$, where $\lambda = sL \log n + ts \Run(\cH)$. 
\end{claim}
\begin{proof}
First note that $\lambda$ upper bounds the cost of inserting and deleting from $\cL$, as well as calling $\ttx{Query-Top}(v)$. For every vertex, when it is first inserted into the stream, we pay it a cost of $\lambda$ to cover the call to $\ttx{Query-Top}(v)$. Moreover, whenever a vertex is added to the queue $\cQ$, we pay a cost of $\lambda$ to cover a subsequent call to $\ttx{Query-Top}(v)$ when it is removed from the queue and inserted again, plus an additional $O(\log n)$ required to insert and remove the top of a priority queue. The only cost of the algorithm which the above does not cover is the cost of calling $\ttx{Query-All}(v)$, which can be bounded by $O(\lambda \cdot |N_G(v) \cap \cL| )$. Note that, by correctness of the top-$k$ LFMIS algorithm (Lemma \ref{lem:correctness}), each vertex in $\cL$ is its own eliminator at the beginning of each time step. It follows that each vertex $u \in N_G(v) \cap \cL$ had its eliminator changed on step $t$, since $\ttx{Query-All}(v)$ is only called on time step $t$ in the third case of Algorithm \ref{alg:ins}, where all points in $|N_G(v) \cap \cL|$ will be made followers of $v$. Thus the total cost of all calls to $\ttx{Query-All}(v)$ can be bounded by $\lambda \sum_{t \in [M}) \cC^t_\pi$, which completes the proof of the claim. 
\end{proof}

Given the above, by Lemma \ref{lem:main} we have that $T \leq M \leq 6 \cC^t_\pi$, and by Theorem \ref{thm:elimBound} we have $\exx{\pi}{\sum_t \cC^t_\pi} = O(M \log n)$. It follows that, the expected total runtime of the algorithm, taken over the randomness used to generate $\pi$ (with $h_1,\dots,h_s$ previously fixed and conditioned on) is at most $O(M \log^2 n + M \lambda \log n)$ as needed.
\end{proof}

\begin{theorem}\label{thm:lshMain}
Let $(\cX,d)$ be a metric space, and fix $\delta \in (0,1/2)$. Suppose that for any $r \in (r_{\min}, r_{\max})$ there exists an $(r,cr,p_1,p_2)$-sensitive hash family $\cH_r: \cX \to U$, such that each $h \in \cH_r$ can be evaluated in time at most $\Run(\cH)$, and such that $p_2$ is bounded away from $1$. Then there is a fully dynamic algorithm that, on a sequence of $M$ insertions and deletions of points from $\cX$, given an upper bound $M \leq \hat{M} \leq \poly(M)$, with probability $1-\delta$, correctly maintains a $c(2+\eps)$-approximate $k$-center clustering to the active point set $P^t$ at all time steps $t$, simultaneously for all $k \geq 1$. The total runtime of the algorithm is at most 
\[\tilde{O}\left(M \cdot \frac{\log \Delta \log \delta^{-1}}{\eps p_1} n^{2 \rho} \cdot \Run(\cH) \right)\]
where $\rho = \frac{\ln p_1}{\ln p_2}$, and $n$ is an upper bound on the maximum number of points at any time step.
\end{theorem}
\begin{proof}

We first demonstrate that there exists an algorithm $\cA(\delta,M)$ which takes as input $\delta \in (0,1/2)$ and $M \geq 1$, and on a sequence of at most $M$ insertions and deletions of points in $\cX$, with probability $1-\delta$ correctly solves the top-$M$ LFMIS with leaders (Definition \ref{def:LFMISLead}) on a graph $G$ that is $(r,c,0)$-approximate threshold graph for $P$ (Definition \ref{def:approxThreshGraph}), where $P \subset \cX$ is the set of all points which were inserted during the sequence, and runs in total expected time $\alpha_{\delta,M}$, where 

\[\alpha_{\delta,M} = \tilde{O}\left(M^{1+2\rho} \log(M/\delta)  \Run(\cH)  \right)\]
The reduction from having such an algorithm to obtaining a $k$-center solution for every $k \geq 1$, incurring a blow-up of $\eps^{-1} \log \Delta$, and requiring one to scale down $\delta$ by a factor of $O(\eps^{-1} \log \Delta)$ so that all $\eps^{-1} \log \Delta$ instances are correct, is the same as in Section \ref{sec:kCenters}, with the modification that the clustering obtained by a MIS $\cL$ at scale $r \geq 0$ has cost at most $cr$, rather than $r$. Thus, in what follows, we focus on a fixed $r$. 

First, setting $s_{\delta,M} = O(\log (M/\delta ) M^{2\rho} / p_1)$ and $t_M = O(\log_{1/p_2} M)$, by Proposition \ref{prop:isAThreshold} it holds that with probability $1-\delta$ the graph $G_{r,cr}(P,\cH,t_M,s_{\delta,M})$ is a $(r,c,L)$-approximate threshold graph, with $\ex{L} < 2$. Then by Theorem \ref{thm:LSHPreMain}, there is an algorithm which maintains a top-$M$ LFMIS with leaders to the graph $G = G_{r,cr}(P,\cH,t_M,s_{\delta,M}) \cap G_{rc}$, which in particular is a $(r,c,0)$-approximate threshold graph for $P$, and runs in expected time at most  

\[O( M ( s_{\delta,M} L \log M+ t_M s_{\delta,M} \Run(\cH)) \log M+ M \log^2 M )\]
where the expectation is taken over the choice of the random ranking $\pi$. Taking expectation over $L$, which depends only on the hash functions $h_1,\dots,h_{s_{\delta,M}}$, and is therefore independent of $\pi$, the expected total runtime is at most $\alpha_{\delta,M}$ as needed. 

To go from the updated time holding in expectation to holding with probability $1-\delta$, we follow the same proof of Proposition \ref{prop:highProb}, except that we set the failure probability of each instance to be $O(\delta/\log^2(M/\delta))$. By the proof of Proposition \ref{prop:highProb}, the total number of copies ever run by the algorithm is at most $O(\log(1/\delta) \sum_{i=1}^{\log M} i) = O(\log(1/\delta) \log^2 M)$ with probability at most $1-\delta/2$, and thus with probability at least $1-\delta$ it holds that both at most $O(\log(1/\delta) \log^2 M)$  copies of the algorithm are run, and each of them is correct at all times. Note that whenever the runtime of the algorithm exceeds this bound, the algorithm can safely terminate, as the probability that this occurs is at most $\delta$ by Proposition \ref{prop:highProb}.

Put together, the above demonstrates the existence of an algorithm $\bar{\cA}(\delta,M)$ which takes as input $\delta \in (0,1/2)$ and $M \geq 1$, and on a sequence of at most $M$ insertions and deletions of points in $\cX$, with probability $1-\delta$ correctly maintains a $c(2+\eps)$-approximation to the optimal $k$-center clustering simultaneously for all $k \geq 1$, with total runtime at most $\tilde{O}(\eps^{-1} \log \Delta \alpha_{\delta,M})$. However, we would like to only run instances of $\bar{\cA}(\delta,M)$ with $M = O(n)$ at any given point in time, so that the amortized update time has a factor of $n^{2 \rho}$ instead of $M^{2\rho}$. To accomplish this, we greedily pick time steps $1 \leq t_1 < t_2 < \dots < M$ with the property that $t_i - t_{i-1} = \lceil |P^{t_{i-1}}|/2 \rceil$. Observe that for such time steps, we have $|P^{t_{i}}| <  2|P^{t_{i-1}}|$. We then define time steps $1 \leq t_1' < t_2' < \dots < M$ with the property that $t_i'$ is the first time step where the active point set size exceeds $2^i$. We then run an instance of $\bar{\cA}(\delta_0,2^i)$, starting with $i=1$, where $\delta_0 = \delta/M$. Whenever we reach the next time step $t_{i+1}'$ we restart the algorithm $\bar{\cA}$ with parameters $(\delta_0, 2^{i+1})$, except that instead of running $\bar{\cA}$ on the entire prefix of the stream up to time step $t_{i+1}'$, we only insert the points in $P^{t_{i+1}'}$ which are active at that time step. Similarly, when we reach a time step $t_{j}$, we restart $\bar{\cA}$ with the same parameters, and begin by isnerting the active point set $P^{t_{j}}$, before continuing with the stream. 

Note that $\bar{\cA}$ is only restarted $\log n$ times due to the active point set size doubling. Moreover, each time it is restarted due to the time step being equal to $t_j$ for some $j$, the at most $2|P^{t_{j-1}}|$ point insertions required to restart the algorithm can be amortized over the $t_j - t_{j-1} \geq |P^{t_{j-1}}| /2$ prior time steps. Since each instance is correct with probability $1-\delta_0$, by a union bound all instances that are ever restarted are correct with probability $1-\delta$. Note that, since the amortized runtime dependency of the overall algorithm on $M$ is polylogarithmic, substituting $M$ with an upper bound $\hat{M}$ satisfying  $M \leq \hat{M} \leq \poly(M)$ increases the runtime by at most a constant. Moreover, we never run $\bar{\cA}(\delta,t)$ with a value of $t$ larger than $2 n$, which yields the desired runtime.  
\end{proof}

\subsection{Corollaries of the LSH Algorithm to Specific Metric Spaces}\label{sec:LSHcorollaries}
We now state our results for specific metric spaces by applying Theorem \ref{thm:lshMain} in combination with known locally sensitive hash functions from the literature. 
The following corollary follows immediately by application of the locally sensitive hash functions of \cite{datar2004locality,har2012approximate}, along with the bounds from Theorem \ref{thm:lshMain}. 

\begin{corollary}\label{cor:Euclidean} Fix any $c \geq 1$. Then there is a fully dynamic algorithm which, on a sequence of $M$ insertions and deletions of points from $d$-dimensional Euclidean space $(\R^d , \ell_p)$, for $p \geq 1$ at most a constant,  with probability $1-\delta$, correctly maintains a $c(4+\eps)$-approximate $k$-center clustering to the active point set $P^t$ at all time steps $t \in [M]$, and simultaneously for all $k \geq 1$. The total runtime is at most 

\[\tilde{O}\left( M \frac{ \log \delta^{-1} \log \Delta  }{\eps} d n^{1/c}\right)\]
\end{corollary}

For the case of standard Euclidean space ($p=2$), one can used the improved ball carving technique of Andoni and Indyk \cite{andoni2006near} to obtain better locally sensitive hash functions, which result in the following:

\begin{corollary}\label{cor:Euclidean2} Fix any $c \geq 1$. Then there is a fully dynamic algorithm which, on a sequence of $M$ insertions and deletions of points from $d$-dimensional Euclidean space $(\R^d , \ell_2)$,  with probability $1-\delta$, correctly maintains a $c(\sqrt{8}+\eps)$-approximate $k$-center clustering to the active point set $P^t$ at all time steps $t \in [M]$, and simultaneously for all $k \geq 1$. The total runtime is at most 

\[\tilde{O}\left( M\frac{ \log \delta^{-1} \log \Delta  }{\eps} d n^{1/c^2 + o(1)}\right)\]
\end{corollary}

Additionally, one can use the well-known MinHash \cite{broder1997resemblance} to obtain a fully dynamic $k$-center algorithm for the \textit{Jaccard Distance}. Here, the metric space is the set of all subsets of a finite universe $X$, equipped with the distance $d(A,B) = 1- \frac{|A \cap B|}{|A \cup B|}$ for $A,B \subseteq X$. We begin stating a standard bound on the value of $\rho$ for the MinHash LSH family.

\begin{proposition}[\cite{indyk1998approximate}]
Let $\cH$ be the hash family given by 
\[ \cH = \{h_\pi :2^X \to X \; | \;  h_\pi(A) = \arg \min_{a \in A} \pi(a), \; \pi \text{ is a permutation of } X \} \] 
Then for any $c \geq 1$ and $r \in [0,1/(2c)]$, we have that $\cH$ is $(r,cr,1-r,1-cr)$-sensitive for the Jaccard Metric over $X$, where $\rho = 1/c$. 
\end{proposition}
\begin{proof}
If $d(A,B) = r$, for any $r \in [0,1]$, we have 
\[ \prb{h\sim \cH}{h(A) = h(B)} = \frac{|A \cap B|}{|A \cup B|} = 1-r \]
Thus, we have $p_1 = 1-r$ and $p_2 = 1-cr$. Now by Claim $3.11$ of \cite{har2012approximate}, we have that for all $x \in [0,1)$ and $c \geq 1$ such that $1-cx > 0$, the following inequality holds:
\[ \frac{\ln(1-x)}{\ln(1-cx)} \leq \frac{1}{c} \]
Thus we have the desired bound:

\[       \rho \leq \frac{\ln(1-r)}{\ln(1- cr)} \leq 1/c \qedhere ¸\]
\end{proof}

\begin{corollary}\label{cor:Jaccard}
Let $X$ be a finite set, and fix any $c \geq 1$. Then there is a fully dynamic algorithm which, on a sequence of $M$ insertions and deletions of subsets of $X$ equipped with the Jaccard Metric,  with probability $1-\delta$, correctly maintains a $c(4+\eps)$-approximate $k$-center clustering to the active point set $P^t$ at all time steps $t \in [M]$, and simultaneously for all $k \geq 1$. The total runtime is at most 
\[\tilde{O}\left(M \frac{ \log \delta^{-1} \log \Delta  }{\eps} |X| n^{1/c}\right)\]
\end{corollary}
\begin{proof}
Letting $r_{\min}$ be the minimum distance between points in the stream, we run a copy of the top-$M$ LFMIS algorithm of Theorem \ref{thm:LSHPreMain} for $r = r_{\min}, (1+\eps)r_{\min},\dots,1/(2c)$, where we set the value of the approximation factor $c$ in Theorem \ref{thm:LSHPreMain} to be scaled by a factor of $2$, so that each instance computes a LFMIS on a $(r,2c,0)$-approximate threshold graph. Note that for any time step $t$, if the copy of the algorithm for $r=1/c$ contains a LFMIS with at least $k+1$ vertices, it follows that the cost of the optimal clustering is at most $1/(4c)$. In this case, we can return an arbitrary vertex as a $1$-clustering of the entire dataset, which will have cost at most $1$, as the Jaccard metric is bounded by $1$, and thereby yielding a $4c$ approximation. Otherwise, the solution is at most a 
$c(4+2\eps)$-approximation, which is the desired result after a re-scaling of $\eps$. Note that $\Run(\cH) = \tilde{O}(|X|)$ to evaluate the MinHash, which completes the proof.  
\end{proof}

\section{Algorithm for $k$-Center Against an Adaptive Adversary}
\label{sec:app-det-upper}

\NewDocumentCommand{\subtree}{O{T} m}{%
	#1(#2)
}
\NewDocumentCommand{\subtreepoints}{m}{%
	\mathcal{P}(#1)
}
\NewDocumentCommand{\detcen}{m}{%
	C_{#1}%
}

\SetKwFunction{FnInsertNode}{InsertIntoNode}
\SetKwFunction{FnDeleteNode}{DeleteFromNode}
\SetKwFunction{FnTryCenter}{TryMakeCenter}
\SetKwFunction{FnInsert}{InsertPoint}
\SetKwFunction{FnDelete}{DeletePoint}
\SetKwData{iscen}{isCenter}
\SetKwData{islowb}{lowerBoundWitness}

We describe a deterministic algorithm for $k$-center that \aw{maintains an $O(\min\{\log(n/k) / \log \log (n + \Delta),k\})$-approximate solution
with an update time of $O(k \log n \cdot \log\Delta \cdot \log (n+\Delta))$.
Given some constant $\epsilon >0$, we maintain one data structure for each value
$\OPT'$ that is a power of $1+\epsilon$ in $[1,\Delta]$, i.e.,
$O(\log_{1+\epsilon}\Delta)$ many. The data structure for each such
value $\OPT'$ outputs a solution of cost at most $O(\min\{\log(n/k) / \log \log (n+\Delta),k\})\OPT'$
or asserts that $\OPT>\OPT'$. We output the solution with smallest
cost which is hence a $O(\min\{\log(n/k) / \log \log (n+\Delta),k\})$-approximation.

Given the value $\costbound$, our algorithm maintains} a hierarchy on the input represented by a $B$-ary tree $T$, which we call \emph{clustering tree}, where $B = \log (n + \Delta)$. The main property of the clustering tree is that the input points are stored in the leaves and each inner node stores a $k$-center set for the $Bk$ centers that are stored at its $B$ children.

\begin{definition}[clustering tree]
	\label{def:clustree}
	Let $\costbound > 0$, let $P$ be a set of points and let $T$ be a $B$-ary tree, where $B \geq 2$. We call $T$ a \emph{clustering tree} on $P$ with \emph{node-cost} $\costbound$ if the following conditions hold:
	\begin{enumerate}
		\item each node $u$ stores at most $Bk$ points from $P$, denoted $\points{u}$, \label{enum:ct-leaf-node}
		\item for each node $u$, at most $k$ points, denoted $\detcen{u}$, are marked as centers. Either, their $k$-center cost is at most $\costbound$ on all points stored in $u$, or $u$ is marked as a witness that there is no center set with cost at most $\costbound/2$, \label{enum:ct-cost}
		\item each inner node stores the at most $Bk$ centers of its children. \label{enum:ct-innernode}
	\end{enumerate}
\end{definition}

For each node $u$ in $T$, the algorithm maintains a corresponding graph on the at most $Bk$ points $\points{u}$ it stores, which is called \emph{blocking graph}. Without loss of generality, we assume that $T$ is a full $B$-ary tree with $n/(Bk)$ leaves. We explain in the proof of \cref{lem:det-main-kcenter} how to get rid of this assumption. For the sake of simplicity, we identify a node $u$ with its associated blocking graph $N=(V,E)$ in the following. For each node $N=(V,E)$, at most $k$ points are marked as centers (\iscen in \cref{alg:det-guess-node}), and the algorithm maintains the invariant that two centers $u,v \in V$ have distance at least $\costbound$ by keeping record of \emph{blocking} edges in the blocking graph between centers and points that have distance less than $\costbound$ to one of these centers. We say that a center $u$ \emph{blocks} a point $v$ (from being a center) if there is an edge $(u,v)$ in the blocking graph. In addition, the algorithm records whether $N$ contains more than $k$ points with pairwise distance greater than $\costbound$ (\islowb in \cref{alg:det-guess-node,alg:det-guess}).

\paragraph*{Insertions (see \FnInsert).} When a point $u$ is inserted into $T$, a node $N=(V,E)$ with less than $Bk$ points is selected and it is checked whether $d(v,u) \leq \costbound$ for any center $v \in V$. If this is the case, the algorithm inserts an edge $(u,v)$ for every such center $v$ into $E$ and terminates afterwards. Otherwise, the algorithm checks whether the number of centers is less than $k$. If this is the case, it marks $u$ as a center and inserts an edge $(u,w)$ for \emph{each point} $w \in V$ with $d(u,w) \leq \costbound$ and recurses on the parent of $N$. Otherwise, if there are more than $k-1$ centers in $u$, the algorithm marks $N$ as witness and terminates.

\paragraph*{Deletions (see \FnDelete).} When a point $u$ is deleted from $T$, the point is first removed from the leaf $N$ (and the blocking graph) where it is stored. If $u$ was not a center, the algorithm terminates. Otherwise, the algorithm checks whether any points were unblocked (have no adjacent node in the blocking graph) and, if this is the case, proceeds by attempting to mark these points as centers and inserting them into the parent of $N$ one by one (after marking the first point as center, the remaining points may be blocked again). Afterwards, the algorithm recurses on the parent of $N$.

\begin{algorithm}
	\SetKwData{neigh}{neighbors}
	\SetKwData{newcen}{newCenters}
	\KwData{$\iscen$ is a boolean array on the elements of $V$, $\islowb$ is a boolean array on the nodes of $T$}
	\Fn{\FnInsertNode{$N = (V,E), p, \costbound$}}{
		insert $p$ into $V$ \;
		\ForEach{$v \in V \setminus \{ p \}$}{
			\If{$\iscen[v] \wedge d(p,v) \leq \costbound$}{
				insert $(v,p)$ into $E$ \;
			}
		}
		\Return \FnTryCenter($G, p, \costbound$) \;
	}
	\Fn{\FnDeleteNode{$N = (V,E), p, \costbound$}}{
		$\newcen \gets \emptyset$ \;
		$\neigh \gets \ngh{p}$ \;
		delete $p$ from $N$ \;
		\If{$\iscen[p] = true$}{
			\ForEach{$u \in \neigh$}{
				$\newcen \gets \newcen \cup \FnTryCenter{N, u, \costbound}$ \;
			}
		}
		\Return{$\newcen$} \;
	}
	\Fn{\FnTryCenter{$N = (V,E), p, \costbound$}}{
		\If{$\dg{p} = 0 \wedge \lvert \{ v \mid v \in V \wedge \iscen[v] \} \rvert < k$}{
			$\iscen[p] \gets true$ \;
			\ForEach{$v \in V \setminus \{ p \}$}{
				\If{$d(p,v) \leq \costbound$}{
					insert $(p,v)$ into $E$ \;
				}
			}
			$\islowb[N] \gets false$ \;
			\Return{$\{ p \}$} \;
		}
		\ElseIf{$\dg{p} = 0$}{
			$\islowb[N] \gets true$ \;
		}
		\Return{$\emptyset$} \;
	}
	\caption{\label{alg:det-guess-node} Insertion and deletion of a point $P$ in a node of the clustering tree $T$ (represented by a blocking graph $G$).}
\end{algorithm}

\begin{algorithm}
	\SetKwData{cen}{centers}
	\SetKwData{newcen}{newCenters}
	\SetKwData{failed}{failed}
	\KwData{$\islowb$ is a boolean array on the nodes of $T$}
	\Fn{\FnInsert{$T, p, \costbound$}}{
		$N \gets$ leaf in $T$ that contains less than $Bk$ elements \;
		\Do{$\cen = \{ p \}$}{
			$\cen \gets \FnInsertNode{$N, p$}$ \;
			$N \gets N.parent$ \;
		}
	}
	\Fn{\FnDelete{$T, p, \costbound$}}{
		$N \gets$ leaf in $T$ that contains $p$ \;
		$\cen \gets \emptyset; \failed \gets false$ \;
		\Do{$N \neq null$}{
			$\cen \gets \cen \cup \FnDeleteNode{$N, p$}$ \;
			$\newcen \gets \emptyset$ \;
			\ForEach{$v \in \cen$}{
				$\newcen \gets \newcen \cup \FnInsertNode{N.parent, v}$ \;
			}
			$\cen \gets \newcen$ \;
			$N \gets N.parent$ \;
		}
	}
	\caption{\label{alg:det-guess} Insertion and deletion of a point $p$ in the clustering tree $T$.}
\end{algorithm}

Given a rooted tree $T$ and a node $u$ of $T$, we denote the subtree of $T$ that is rooted at $u$ by $\subtree{u}$. For a clustering tree $T$, we denote the set of all points stored at the leaves of $\subtree{u}$ by $\subtreepoints{u}$. Recall that the points directly stored at $u$ are denoted by $\points{u}$. Observe that for each node $u$ in a clustering tree, $\subtree{u}$ is a clustering tree of $\subtreepoints{u}$.

\subsection{Structural Properties of the Algorithm}

We show that \cref{alg:det-guess} maintains a clustering tree.

\begin{lemma}
	\label{lem:det-clustree-insert}
	Let $T$ be a clustering tree on a point set $P$. After calling $\FnInsert(T, p)$ (see \cref{alg:det-guess}) for some point $p \notin P$, $T$ is a clustering tree on $P \cup \{ p \}$.
\end{lemma}
\begin{proof}
	We prove the statement by induction over the recursive calls of \FnInsertNode in \FnInsert. Let $N=(V,E)$ be the leaf in $T$ where $p$ is inserted. Condition~\ref{enum:ct-leaf-node} in \cref{def:clustree} is guaranteed for $N$ by \FnInsert. The algorithm \FnInsertNode ensures that $p$ is marked as a center only if it is not within distance $\costbound$ of any other center. Let $C$ be the center set of $N$ before inserting $p$. By condition~\ref{enum:ct-cost}, $C$ has $\cost[V]{C} \leq \costbound$. Let $C'$ be the center set of any optimal solution on $V$. If $\cost[V]{C'} \leq \costbound / 2$, each center of $C$ covers at least one cluster of $C'$. By pigeonhole principle, $p$ is not blocked by a center in $C$ if and only if $\lvert C \rvert < k$. Otherwise, if $\cost[V]{C'} > \costbound/2$, $p$ is chosen if no center in $C$ covers $p$ and $\lvert C' \rvert < k$, or $N$ is marked as witness. It follows that condition~\ref{enum:ct-cost} is still satisfied for $N$ after \FnInsertNode terminates.
	
	Now, let $N=(V,E)$ be any inner node in a call to $\FnInsertNode(N,p)$. We note that such call is only made if $p$ was marked as a center in its child $N'$ on which $\FnInsertNode$ was called before by $\FnInsert$. Thus, $p$ is inserted into $V$ if and only if $p$ is a center in $N'$. Therefore, condition~\ref{enum:ct-cost} is satisfied for $N$. By the above reasoning, condition~\ref{enum:ct-innernode} is also satisfied.
\end{proof}

\begin{lemma}
	Let $T$ be a clustering tree on a point set $P$. After calling $\FnDelete(T, p)$ (see \cref{alg:det-guess}) for some point $p \in P$, $T$ is a clustering tree on $P \setminus \{ p \}$.
\end{lemma}
\begin{proof}
	We prove the statement by induction over the loop's iterations in \FnDelete. Let $N=(V,E)$ be the leaf in $T$ where $p$ was inserted. \FnDeleteNode deletes $p$ from $N$ and iterates over all unblocked points to mark them as centers one by one. After deleting $p$, condition~\ref{enum:ct-leaf-node} holds for $N$. Let $U$ be the set of unblocked points after removing $p$, let $C$ be the center set after removing $p$ from $V$, and let $C'$ be the center set of any optimal solution on $V \setminus \{ p \}$. If $\cost[V]{C'} \leq \costbound / 2$, each unblocked point from $U$ covers at least one cluster of $C'$ with cost $\costbound$. Therefore, any selection of $k - \lvert C \rvert$ points from $U$ that do not block each other together with $C$ is a center set for $V$ with cost $\costbound$. Otherwise, if $\cost[V]{C'} > \costbound / 2$, a set of at most $k - \lvert C \rvert$ points from $U$ is chosen, or $N$ is marked as witness. It follows that condition~\ref{enum:ct-cost} is still satisfied for $N$ after \FnDeleteNode terminates. Let $C''$ be the center set that is returned by \FnDeleteNode. \FnDelete inserts all points from $C''$ into the parent node. This reinstates condition~\ref{enum:ct-innernode} on the parent node. Then, \FnDeleteNode recurses on the parent.
\end{proof}

\subsection{Correctness of the Algorithm}

We use the following notion of super clusters and its properties to prove the $O(k)$ upper bound on the approximation ratio of the algorithm.

\begin{definition}[super cluster]
	\label{def:supercluster}
	Let $P$ be a set of points and let $C$ be a center set with $\cost[P]{C} \leq 2\costbound$. Consider the graph $G = (C, E)$, where $E = \{ (u,v) \mid d(u,v) \leq 2\costbound \}$. For every connected component in $G$, we call the union of clusters corresponding to this component a \emph{super cluster}.
\end{definition}

\begin{lemma}
	\label{lem:clustree-supercluster}
	Let $T$ be a clustering tree constructed by \cref{alg:det-guess} with node-cost $\costbound$ and no node marked as witness. For any node $N$ in $T$, $N$ contains one point from each super cluster in $\subtreepoints{N}$ that is marked as center.
\end{lemma}
\begin{proof}
	Let $S$ be a supercluster of $P$. Let $N$ be a node in $T$ that contains a point $p \in S$. For the sake of contradiction, assume that there exists no point $q \in S \cap N$ that is marked as center. By \cref{def:supercluster}, the $k$-center clustering cost of the centers in $N$ is greater than $\costbound$. By \cref{def:clustree}.\ref{enum:ct-cost}, this implies that a point must be marked as witness. This is a contradiction to the assumption that no node is marked as witness.
\end{proof}

The following simple observation leads to the bound of $O(\log (n/k) / \log \log (n+\Delta))$ on the approximation ratio.

\begin{observation}
	\label{lem:kcenter-of-cluster}
	Let $c > 0$, $\ell \in \mathbb{N}$, let $P_1, \ldots, P_\ell$ be sets of points and let $C_1, \ldots, C_\ell$ so that $\cost[P_i]{C_i} \leq c$ for every $i \in [\ell]$. For every $k$-center set $C'$ with $\cost[\cup_i C_i]{C'} \leq c'$ we have $\cost[\cup_i P_i]{C'} \leq c + c'$ by the triangle inequality.
\end{observation}

We combine the previous results and obtain the following approximation ratio for our algorithm.

\begin{lemma}
	\label{lem:det-kcenter-cost}
	Let $T$ be a clustering tree on a point set $P$ that is constructed by \cref{alg:det-guess} with node-cost $\costbound$. Let $C$ be the points in the root of $T$ that are marked as centers. If no node in $C$ is marked as witness, $ \cost[P]{C} \leq \min\{k, \log(n/k) / \log B \} \cdot 4 \opt{P}$. Otherwise, $\opt{P} \geq \costbound / 2$.
\end{lemma}
\begin{proof}
	Assume that $\opt{P} \leq \costbound / 2$, as otherwise, there exists a node that stores at least $k+1$ points that have pairwise distance $\costbound$, which implies the claim. First, we prove $\cost[P]{C} \leq \log(n/(Bk)) / \log B \cdot \costbound$. It follows from \cref{lem:kcenter-of-cluster} that $C$ has $k$-center cost $2\costbound$ on the points stored in the root's children. Since $T$ has depth at most $\log(n/(Bk)) / \log B$, it follows by recursively applying \cref{lem:kcenter-of-cluster} on the children that $C$ also has cost $\log(n/(Bk)) / \log B \cdot 2 \costbound$ on $P$.
	
	Now, we prove $\cost[P]{C} \leq k$. Let $p \in P$, and let $S$ be the super cluster of $p$ with corresponding optimal center set $C'$. We have $d(p,C') \leq \costbound / 2$. By \cref{lem:clustree-supercluster}, there exists a center $q \in C$ so that $d(q,C') \leq \costbound / 2$. By the definition of super clusters, for every $x,y \in C'$, $d(x,y) \leq (k-1) \cdot 2\costbound$. It follows from the triangle inequality that $d(p,q) \leq 2k \costbound$.
\end{proof}

\subsection{Update Time Analysis}

\begin{lemma}
	\label{lem:det-time}
	The amortized update time of \cref{alg:det-guess} is $O(Bk \log (n/k) / \log B)$.
\end{lemma}
\begin{proof}
	To maintain blocking graphs efficiently, the graphs are stored in adjacency list representation and the degrees of the vertices as well as the number of centers are stored in counters. When a point $p$ is inserted, $3Bk \log (n/(Bk)) / \log B$ tokens are paid into the account of $p$. Each token pays for a (universally) constant amount of work.
	
	Each point is inserted at most once in each of the $\log (n/(Bk)) / \log B$ nodes from the leaf where it is inserted up to the root. The key observation is that it is marked as a center in each of these nodes at most once (when it is inserted, or later when a center is deleted): Marking a point as center is irrevocable until it is deleted. For each node $N$ and point $p \in N$, it can be checked in constant time whether $p$ can be marked as a center in $N$ by checking its degree in the blocking graph. Marking $p$ as a center takes time $O(Bk)$ because it is sufficient to check the distance to all other at most $Bk$ points in $N$ and insert the corresponding blocking edges. For each point $p$, we charge the time it takes to \emph{mark} $p$ as a center to the account of $p$. Therefore, marking $p$ as center withdraws a most $Bk \log (n/(Bk))$ tokens from its account in total.
	
	Consider the insertion of a point $p$ into a node. As mentioned, $p$ is inserted in at most $\log (n/(Bk)) / \log B$ nodes, and in each of these nodes, it is inserted at most once (when it is inserted into the tree, or when it is marked as a center in a child node). Inserting a point into a node $N$ requires the algorithm to check the distance to all at most $k$ centers in $N$ to insert blocking edges, which results in at most $k \log (n/(Bk)) / \log B$ work in total.
	
	It remains to analyze the time that is required to update the tree when a point $p$ is deleted. For any node $N$, if $p$ is not a center in $N$, deleting $p$ takes constant time. Otherwise, if $p$ is a center, the algorithm needs to check its at most $Bk$ neighbors in the blocking graph one by one whether they can be marked as centers. Checking a point $q$ takes only constant time, and marking $q$ as a center has already been charged to $q$ by the previous analysis. All centers that have been marked have to be inserted into the parent of $N$, but this has also been charged to the corresponding points. Therefore, deleting $p$ consumes at most $Bk \log (n/(Bk)) / \log B$ tokens from the account of $p$.
\end{proof}

\begin{lemma}
	\label{lem:det-time-worstcase}
	The worst-case insertion time of \cref{alg:det-guess} is $O(Bk \log(n/k) / \log B)$, and the worst-cast deletion time is $O(Bk^2 \log(n/k) / \log B)$.
\end{lemma}
\begin{proof}
	To maintain blocking graphs efficiently, the graphs are stored in adjacency list representation and the degrees of the vertices as well as the number of centers are stored in counters.
	
	Each point is inserted at most once in each of the $\log (n/(Bk)) / \log B$ nodes from the leaf where it is inserted up to the root. For each node $N$ and point $p \in N$, it can be checked in constant time whether $p$ can be marked as a center in $N$ by checking its degree in the blocking graph. Marking $p$ as a center takes time $O(Bk)$ because it is sufficient to check the distance to all other at most $Bk$ points in $N$ and insert the corresponding blocking edges.
	
	Consider the insertion of a point $p$ into the tree. As mentioned, $p$ is inserted in at most $\log (n/(Bk)) / \log B$ nodes. Therefore, during the initial insertion of $p$ into the tree, the insertion procedure may insert $p$ into at most $\log (n/(Bk)) / \log B$ nodes. Inserting a point into a node $N$ requires the algorithm to check the distance to all at most $k$ centers in $N$ to insert blocking edges and possibly marking it as a center, which results in at most $O(Bk \log (n/(Bk)) / \log B)$ time in total.
	
	Now, consider the deletion of $p$ from the tree. For any node $N$, if $p$ is not a center in $N$, deleting $p$ takes constant time. Otherwise, if $p$ is a center, the algorithm needs to check its at most $Bk$ neighbors in the blocking graph one by one whether they can be marked as centers. Checking a point $q$ takes only constant time, and marking $q$ as a center takes time $O(Bk)$. All at most $k$ points that have been marked have to be inserted into the parent of $N$, which takes time $O(Bk \cdot k)$. As the length of any path from a leaf to the root is at most $\log(n/(Bk)) / \log B$, deleting $p$ requires at most $O(Bk^2 \log (n/(Bk)) / \log B)$ time.
\end{proof}

\subsection{Main Theorem}

It only remains to combine all previous results to obtain \cref{thm:mpt-det-upper}.

\begin{theorem}[\cref{thm:mpt-det-upper}]
	\label{lem:det-main-kcenter}
	Let $\epsilon, k > 0$ and $B \geq 2$. There exists a deterministic algorithm for the dynamic $k$-center problem that has amortized update time $O(Bk \log n \log(\dmax) / \log(1+\epsilon))$ and approximation factor $(1+\epsilon) \cdot \min\{4k, 4\log(n/k) / \log B\}$. The worst-cast insertion time is $O(Bk \log n \log\Delta / \log(1+\epsilon))$, and the worst-case deletion time is $O(Bk^2 \log n \log\Delta / \log(1+\epsilon))$.
\end{theorem}
\begin{proof}
	Since $P = (\points{1}, \ldots, \points{n})$ is a dynamic point set, its size $n_i := \lvert \points{i} \rvert$ can increase over time. Therefore, we need to remove the assumption that the clustering tree has depth $\log(\max_{t \in [n]} n_t)$. First, we note that we can insert and delete points so that the clustering tree $T$ is a complete $B$-ary tree (the inner nodes induce a full $B$-tree, and all leaves on the last level are aligned left): We always insert points into the left-most leaf on the last level of $T$ that is not full; when a point is deleted from a leaf $N$ that is not the right-most leaf $N'$ on the last level of $T$, we delete an arbitrary point $q$ from $N'$ and insert $q$ into $N$. Deleting and reinserting points this way can be seen as two update operations, and therefore, it can only increase time required to update the tree by a factor of $3$. Furthermore, any leaf can be turned into an inner node by adding a copy of itself as its left child and adding an empty node as its right child. Vice versa, an empty leaf and its (left) sibling can be contracted into its parent. This way, the algorithm can guarantee that the depth of the tree is between $\lfloor \log(n_t/(Bk)) / \log B \rfloor$ and $\lceil \log(n_t/(Bk)) / \log B \rceil$ at all times $t \in [n]$.
	
	Recall that $d(x,y) \geq 1$ for every $x,y \in \cX$, and $\dmax = \max_{x,y \in \cX} d(x,y)$. For every $\Gamma \in \{ (1+\epsilon)^i \mid i \in [\lceil \log_{1+\epsilon} (\dmax) \rceil]$, the algorithm maintains an instance $T_\Gamma$ of a clustering tree with node-cost $\costbound = (1+\epsilon)^\Gamma$ by invoking \cref{alg:det-guess}. After every update, the algorithm determines the smallest $\Gamma$ so that no node in $T_\Gamma$ is marked as witness, and it reports the center set of the root of $T_\Gamma$. The bound on the cost follows immediately from \cref{lem:det-kcenter-cost}. Since there are at most $\log(\dmax) / \log(1+\epsilon)$ instances, the bounds on the running time follow from \cref{lem:det-time,lem:det-time-worstcase}.	
\end{proof}

\section{Algorithms for $k$-Sum-of-Radii and $k$-Sum-of-Diameters}

\label{sec:primal-dual}

In this section we present our (randomized) dynamic $(13.008+\epsilon)$-approximation
algorithm for $k$-sum-of-radii with an amortized update time of $k^{O(1/\epsilon)}\log\Delta$,
against an oblivious adversary. Our strategy is to maintain a 
bi-criteria approximation with $O(k/\epsilon)$ clusters whose sum of radii is at most $(6+\epsilon)\OPT$
(and which covers all input points). We show how to use an arbitrary
offline $\alpha$-approximation algorithm to turn this solution into
a $(6+2\alpha+\eps)$-approximate solution with only $k$ clusters.
Using the algorithm in \cite{CharikarP04} (for which $\alpha=3.504+\eps$),
this yields a dynamic $(13.008+\epsilon)$-approximation for $k$-sum-of-radii,
and hence a $(26.015+\eps)$-approximation for $k$-sum-of-diameters. %
We provide the algorithm and a high-level overview of its analysis in \cref{sec:pd-bicriteria,sec:pd-dynamic}. The details of the formal proofs can be found in \cref{sec:pd-proofs}.

\subsection{Bicriteria Approximation}
\label{sec:pd-bicriteria}

Assume that we are given an $\epsilon>0$ such that w.l.o.g.~it holds
that $1/\epsilon\in\N$. We maintain one data structure for each value
$\OPT'$ that is a power of $1+\epsilon$ in $[1,\Delta]$, i.e.,
$O(\log_{1+\epsilon}\Delta)$ many. The data structure for each such
value $\OPT'$ outputs a solution of cost at most $(13.008+\epsilon)\OPT'$
or asserts that $\OPT>\OPT'$. We output the solution with smallest
cost which is hence a $(13.008+O(\epsilon))$-approximation.
We describe first how to maintain the mentioned bi-criteria approximation.
We define $z:=\epsilon\OPT'/k$. Our strategy is to maintain a solution
for an auxiliary problem based on a Lagrangian relaxation-type approach.
More specifically, we are allowed to select an arbitrarily
large number of clusters, however, for each cluster we need to pay
a fixed cost of $z$ plus the radius of the cluster.
For the radii we allow only integral multiples of $z$ that are bounded
by $\OPT'$, i.e., only radii in the set $R=\{z,2z,...,\OPT'-z,\OPT'\}$.

This problem can be modeled by an integer program. We formulate the
problem as an LP $(P)$ which has a variable
$x_{p}^{(r)}$ for each combination of a point $p$ and a radius~$r$
from a set of (suitably discretized) radii $R$, and a constraint
for each point $p$. Let $(D)$ denote its dual LP, see below, where
$z:=\epsilon\OPT'/k$.

\begin{alignat*}{9}
	\min & \,\,\, & \sum_{p\in P}\sum_{r\in R}x_{p}^{(r)}(r+z) &  &  &  &  &  &  &  & \max & \,\,\, & \sum_{p\in P}y_{p}\,\,\,\,\,\,\,\,\,\,\,\\
	\mathrm{s.t.} &  & \sum_{p'\in P}\sum_{r:d(p,p')\le r}x_{p'}^{(r)} & \ge1 & \,\,\, & \forall p\in P & \,\,\,\,\,\,\,\: & (P) & \,\,\,\,\,\,\,\,\,\,\,\,\,\,\, &  & \mathrm{s.t.} &  & \sum_{p'\in P:d(p,p')\le r}y_{p'} & \le r+z & \,\,\, & \forall p\in P,r\in R & \,\,\,\,\,\,\,\:(D)\\
	&  & x_{p}^{(r)} & \ge0 &  & \forall p\in P\,\,\forall r\in R &  &  &  &  &  &  & y_{p} & \ge0 &  & \forall p\in P
\end{alignat*}

We describe first an offline
primal-dual algorithm that computes an integral solution to~$(P)$ and
a fractional solution to the dual $(D)$, so that their costs differ
by at most a factor $6$. By weak duality, this implies that our solution
for $(P)$ is a $6$-approximation. 
We initialize $x\equiv0$ and $y\equiv0$, define $U_{1}:=P$, and
we say that a pair $(p',r')$ \emph{covers a point }$p$ if $d(p',p)\le r'$.
Our algorithm works in iterations. %

At the beginning of the $i$-th iteration, assume that we are given
a set of points $U_{i}$, a vector $\big(x_{p}^{(r)}\big)_{p\in P,r\in R}$,
and a vector $\left(y_{p}\right)_{p\in P}$, such that $U_{i}$ contains
all points in $P$ that are not covered by any pair $(p,r)$ for which
$x_{p}^{(r)}=1$ (which is clearly the case for $i=1$ since $U_{1}=P$
and $x\equiv0$). We select a point $p\in U_{i}$ uniformly at random
among all points in $U_{i}$. We raise its dual variable $y_{p}$
until there is a value $r\in R$ such that the dual constraint for
$(p,r)$ becomes \emph{half-tight, }meaning that %
$\sum_{p':d(p,p')\le r}y_{p'}=r/2+z.$ %

Note that it might be that the constraint is already at least half-tight
at the beginning of iteration $i$ in which case we do not raise $y_{p}$
but still perform the following operations. Assume that $r$ is the
largest value for which the constraint for $(p,r)$ is at least half-tight.
We define $x_{p}^{(2r)}:=1$, $(p_{i},r_{i}):=(p,2r)$ and set $U_{i+1}$
to be all points in $U_{i}$ that are not covered by $(p,2r)$. 
\begin{lemma}
	\label{lem:pd-iteration-time} Each iteration $i$ needs a running
		time of $O(\lvert U_{i}\rvert+ik/\epsilon)$. \end{lemma} 
We stop
when in some iteration $i^{*}$ it holds that $U_{i^{*}}=\emptyset$
or if we completed the $(2k/\epsilon)^{2}$-th iteration and $U_{(2k/\epsilon)^{2}+1}\ne\emptyset$.
Suppose that $x$ and $y$ are the primal and dual vectors after the
last iteration. In case that $U_{(2k/\epsilon)^{2}}\ne\emptyset$
we can guarantee that our dual solution has a value of more than $\OPT'$
from which we can conclude that $\OPT>\OPT'$; therefore, we stop
the computation for the estimated value $\OPT'$ in this case. 
\begin{lemma}[{restate=[name=]lemPdIterationsBound}]
	\label{lem:pd-iterations-bound}
	If $U_{(2k/\epsilon)^2 + (2k/\epsilon)}\ne\emptyset$ then $\OPT>\OPT'$. 
\end{lemma}
If $U_{(2k/\epsilon)^{2}}=\emptyset$
we perform a pruning step in order to transform $x$ into a
solution whose cost is at most by a factor $3$ larger than the cost
of $y$. We initialize $\bar{S}:=\emptyset$.
Let $S=\{(p_{1},r_{1}),(p_{2},r_{2}),...\}$ denote the set
of pairs $(p,r)$ with $x_{p}^{(r)}=1$. We sort the pairs in $S$
non-increasingly by their respective radius $r$. Consider a pair
$(p_{j},r_{j})$. We insert the cluster $(p_{j},3r_{j})$ in our solution $\bar{S}$
and delete from $S$ all pairs $(p_{j'},r_{j'})$ such that $j'>j$
and $d(p_{j},p_{j'})<r_{j}+r_{j'}$. Note that $(p_{j},3r_{j})$ covers
all points that are covered by any deleted pair $(p_{j'},r_{j'})$
due to our ordering of the pairs. Let $\bar{S}$ denote
the solution obained in this way and let $\bar{x}$ denote the corresponding
solution to $(P)$, i.e., $\bar{x}_{p}^{(r)}=1$ if and only if $(p,r)\in\bar{S}$.
We will show below that $\bar{S}$ is a feasible solution to~$(P)$
with at most $O(k/\epsilon)$ clusters. Let $\bar{C}\subseteq P$
denote their centers, i.e., $\bar{C}=\{p\mid\exists r\in R:(p,r)\in\bar{S}\}$.
\begin{lemma} \label{lem:mpt-pd-pruning-time} Given $x$ we can
	compute $\bar{x}$ in time $O((k/\epsilon)^{4})$ and $\bar{x}$ selects
	at most $O(k/\epsilon)$ centers. %
\end{lemma} 
We transform now the bi-criteria approximate solution $\bar{x}$ into a feasible solution
$\tilde{x}$ with only $k$ clusters. To this end, we invoke the offline
$(3.504+\epsilon)$-approximation algorithm from \cite{CharikarP04}
on the input points $\bar{C}$. Let $\hat{S}$ denote the set of pairs
$(\hat{p},\hat{r})$ that it outputs (and note that not necessarily
$\hat{r}\in R$ since we use the algorithm as a black-box). Note that
$\hat{S}$ covers only the points in $\bar{C}$, and not necessarily
all points in $P$. 
On the other hand, the solution $\hat{S}$ has a cost of at most $2\OPT$ since
	we can always find a solution with this cost covering $\bar{C}$, even if we are only allowed
	to select centers from $\bar{C}$.
Thus, based on $\bar{S}$ and $\hat{S}$ we compute
a solution $\tilde{S}$ with at most $k$ clusters that covers $P$.
We initialize $\tilde{S}:=\emptyset$. For each pair $(\hat{p},\hat{r})\in\hat{S}$
we consider the points $\bar{C}'$ in $\bar{C}$ that are covered
by $(\hat{p},\hat{r})$. Among all these points, let $\bar{p}$ be
the point with maximum radius $\bar{r}$ such that $(\bar{p},\bar{r})\in\bar{S}$.
We add to $\tilde{S}$ the pair $(\hat{p},\hat{r}+\bar{r})$ and remove
$\bar{C}'$ from $\bar{C}$. We do this operation with each pair $(\hat{p},\hat{r})\in\hat{S}$.
Let $\tilde{S}$ denote the resulting set of pairs. 
\begin{lemma}[{restate=[name=]lemPdOfflineTime}]
	\label{lem:pd-offline-time}
	Given $\bar{S}$ and $\hat{S}$ we can compute $\tilde{S}$ in time
	$O(k^3/\epsilon^2)$. 
\end{lemma}

We show that $\tilde{S}$ is a feasible solution with small cost.
We start by bounding the cost of $\bar{x}$ via $y$. 
\begin{lemma}[{restate=[name=]lemPdCostA}]
	\label{lem:pd-cost-a}
	We have that $\bar{x}$ and $y$ are feasible solutions to (P) and
	(D), respectively, for which we have that 
	$
	\sum_{p\in P}\sum_{r\in R}\bar{x}_{p}^{(r)}(r+z)\le 6\cdot\sum_{p\in P}y_{p}\le6\OPT'.
	$
\end{lemma}
Next, we argue that $\tilde{S}$ is feasible and bound
its cost by the cost of $S'$ and the cost of $\bar{x}$. 
\begin{lemma}[{restate=[name=]lemPdCostB}]
	\label{lem:pd-cost-b}
	We have that $\tilde{S}$ is a feasible solution with cost at most
	$\sum_{(\hat{p},\hat{r})\in \hat{S}}\hat{r}+\sum_{p\in P}\sum_{r\in R}\bar{x}_{p}^{(r)}(r+z)\le(13.008+\epsilon)\OPT'$. 
\end{lemma}

\subsection{Dynamic Algorithm}
\label{sec:pd-dynamic}

We describe now how we maintain the solutions $x,\bar{x},y,S,\bar{S}$,
and $\tilde{S}$ dynamically when points are inserted or deleted.
Our strategy is similar to \cite{ChaFul18}.

Suppose that a point $p$ is inserted. For each $i\in\{1,...,2k/\epsilon+1\}$
we insert $p$ into the set $U_{i}$ if $U_{i}\ne\emptyset$ and $p$
is not covered by a pair $(p_{j},r_{j})$ with $j\in\{1,...,i-1\}$.
If there is an index $i\in\{1,...,2k/\epsilon+1\}$ such that $U_{i-1}\ne\emptyset$
(assume again that $U_{0}\ne\emptyset$), $U_{i}=\emptyset$, and
$p$ is not covered by any pair $(p_{j},r_{j})$ with $j\in\{1,...,i-1\}$,
we start the above algorithm in the iteration $i$, being initialized
with $U_{i}=\{p\}$ and the solutions $x,y$ as being computed previously.

Suppose now that a point $p$ is deleted. We remove $p$ from each
set $U_{i}$ that contains $p$. If there is no $r\in R$ such that
$(p,r)\in S$ then we do not do anything else. Assume now that $(p,r)\in S$
for some $r\in R$. %
The intuition is that this does not happen very often since in each
iteration $i$ we choose a point uniformly at random from $U_{i}$.
More precisely, in expectation the adversary needs to delete a constant
fraction of the points in $U_{i}$ before deleting $p$. Consider
the index $i$ such that $(p,r)=(p_{i},r_{i})$. We restart the algorithm
from iteration $i$. More precisely, we initialize $y$ to the values
that they have after raising the dual variables $y_{p_{1}},...,y_{p_{i-1}}$
in this order as described above until a constraint for the respective
point $p_{j}$ becomes half-tight. We initialize $x$ to the corresponding
primal variables, i.e., $x_{p_{j}}^{(2r_{j})}=1$ for each $j\in\{1,...,i-1\}$
and $x_{p'}^{(r')}=0$ for all other values of $p',r'$. Also, we
initialize the set $U_{i}$ to be the obtained set after removing
$p$. With this initialization, we start the algorithm above in iteration
$i$, and thus compute like above the (final) vectors $y,x$ and based
on them $S,\bar{S}$, and $\tilde{S}$.

When we restart the algorithm in some iteration $i$ then it takes
time $O(|U_{i}|k^{2})$ to compute the new set $S$. We can charge
this to the points that were already in $U_{i}$ when $U_{i}$ was
recomputed the last time %
and to the points that were inserted into $U_{i}$ later. After a
point $p$ was inserted, it is charged at most $O(k/\epsilon)$ times
in the latter manner since it appears in at most $O(k/\epsilon)$
sets $U_{i}$. Finally, given $S$, we can compute the sets $\bar{S}$,
$\hat{S}$, and $\tilde{S}$ in time $O(k^{3}/\epsilon^{2})$. The
algorithm from \cite{CharikarP04} takes time $n^{O(1/\epsilon)}$ if the input has size $n$.
One can show that this yields an update time of $k^{O(1/\epsilon)}+(k/\epsilon)^{4}$
for each value $\OPT'$. Finally, the same set $\tilde{S}$ yields
a solution for $k$-sum-of-diameters, increasing the approximation
ratio by a factor of~2. 

\thmPdMain*

\subsection{Deferred Proofs}
\label{sec:pd-proofs}

\SetKwFunction{FnPrimalDual}{PrimalDual}
\SetKwFunction{FnPrune}{Prune}

In this section, we formally prove the correctness of \cref{alg:pd} for the $k$-sum-of-radii and the $k$-sum-of-diameter problem. To compute a solution in the static setting, the algorithm is invoked as $\FnPrimalDual(P, \emptyset, R, z, k, \epsilon, 0, 0, 0)$, where $P$ is the set of input points, $R$ is the set of radii, $z = \epsilon \OPT' / k$ is the facility cost and $k$ is the number of clusters. \Cref{alg:pd} is a pseudo-code version of the algorithm described earlier.

\begin{algorithm}
	\SetKwData{null}{null}
	\KwData{point set $P$, unassigned point sets $U = \{ U_0, \ldots \}$, radii set $R$, facility cost $z$, number of clusters $k$, precision $\epsilon$, primal vector $x = \{ x^{(r)}_p \mid r \in R \wedge p \in P \}$, dual vector $Y = \{ y_p \mid p \in P \}$}
	\Fn{\FnPrimalDual($P, U, R, z, k, \epsilon, x, y, i$)}{
		\While{$U_i \neq \emptyset$}{
			$p_i \gets$ uniformly random point from $U_i$ \;
			$\delta_i \gets \max ( \{ 0 \} \cup \{\delta' \mid \delta' \in \R \wedge \forall r' \in R : \sum_{p' \in P : d(p_i, p') \le r'} y_{p'} + \delta' \leq r'/2 + z \} )$ \;
			\If{$\delta_i > 0$}{
				$r_i \gets \max \{r' \mid r' \in R \wedge \sum_{p' \in P : d(p_i, p')\le r'} y_{p'} + \delta_i = r'/2 + z \}$ \;
			}
			\Else{
				$r_i \gets \max \{r' \mid r' \in R \wedge \sum_{p' \in P : d(p_i, p')\le r'} y_{p'} \geq r'/2 + z \}$ \;
			}
			$y_{p_i} \gets \delta_i$; $x^{(2r_i)}_{p_i} \gets 1$ \;
			$U_{i+1} \gets U_i \setminus \{ p' \mid p' \in U_i \wedge d(p_i, p') \leq 2r_i \}$ \;
			$i \gets i + 1$ \;
			\If{$i > (2k/\epsilon)^2$}{
				\Return{``$\OPT' < \OPT$''} \;
			}
		}
		$\bar{S} \gets \emptyset$; $S \gets$ sort $(p_j, r_j)_{j \in [i-1]}$ non-increasingly according to $r_j$ \;
		\ForAll{$(p,r) \in S$}{
			\If{$\nexists (p_j,r_j) \in \bar{S} : d(p, p_j) \leq r + r_j$}{
				$\bar{S} \gets \bar{S} \cup \{ (p, 3r) \}$ \;
			}
		}
		\Return{$\bar{S}, U, x, y, r, i$} \;
	}
	\caption{\label{alg:pd} Pseudo code of the primal-dual algorithm for $k$-sum-of-radii as described in \cref{sec:primal-dual}.}
\end{algorithm}

\paragraph{Proof of \cref{lem:pd-iteration-time}}

We bound the running time of a single primal-dual step. This implies \cref{lem:pd-iteration-time}.

\begin{lemma}[\cref{lem:pd-iteration-time}]
	The running time of the $i$\xth/ iteration of the while-loop in \FnPrimalDual is $O(ik / \epsilon + \lvert U_i \rvert)$.
\end{lemma}
\begin{proof}
	In each iteration, at most one entry of $y$, i.e., $y_{p_i}$, increases. Therefore, at most $i-1$ entries of $y$ are non-zero at the beginning of iteration $i$. By keeping a list of non-zero entries, each sum corresponding to a constraint of the dual program can be computed in time $O(i)$. Since $\lvert R \rvert \leq k / \epsilon$, it follows that $\delta_i$ and $r_i$ can be computed in time $O(ik/\epsilon)$. To construct $U_{i+1}$, it is sufficient to iterate over $U_i$ once.
\end{proof}

\paragraph{Proof of \cref{lem:pd-iterations-bound}}

We show that our choice of $z = \epsilon \OPT' / k$ will result in a solution $\bar{S}$ if $\OPT \leq \OPT'$.

\lemPdIterationsBound*
\begin{proof}
	Since, by weak duality, $\sum_{p \in P} y_p$ is a lower bound to the optimum value of $(P)$, we bound the number of iterations that are sufficient to guarantee $\sum_{p \in P} y_p > \OPT'$. Call an iteration of the while-loop in \FnPrimalDual \emph{successful} if $\delta_i > 0$, and \emph{unsuccessful} otherwise. First, observe that for each successful iteration $i$, the algorithms increases $y_{p_i}$ by at least $z/2$. This is due to the fact that all values in $R$ are multiples of $z$, and therefore $r/2 + z$ is a multiple of $z/2$ for any $r \in R$. Since the algorithm increases $y_{p_i}$ as much as possible, all $y_{p_i}$ are multiples of $z/2$. Therefore, after $\OPT' / (z/2) +1= 2k / \epsilon +1$ successful iterations, it holds that $\OPT \geq \sum_{p \in P} y_p > \OPT'$.
	
	Now, we prove that for each successful iteration, there are at most $\lvert R \rvert$ unsuccessful iterations. Then, it follows that after $((2k/\epsilon)+1) \cdot \lvert R \rvert \leq (2k/\epsilon)^2 + (2k/\epsilon)$ iterations, $\OPT > \OPT'$. Let $i$ be an unsuccessful iteration. The crucial observation for the following argument is that for the maximum $r_i \in R$ so that $(p_i, r_i)$ is at least half-tight and for any $j < i$ so that $d(p_i, p_j) \leq r_i$ and $y_{p_j} > 0$, we have $r_j < r_i$: otherwise, $p_i$ would have been removed from $U_j$. On the other hand, such $p_j$ must exist because the dual constraint $(p_i, r_i)$ is at least half-tight but $y_{p_i} = 0$. Now, we \emph{charge} the radius $r_i$ to the point $p_j$ and observe that we will never charge $r_i$ to $p_j$ again. This is due to the fact that all points $p \in U_{i-1}$ with distance $d(p_j, p) \leq r_i$ are removed from $U_i$ because $d(p_i, p) \leq d(p_i, p_j) + d(p_j, p) \leq 2r_i$. Therefore, for any successful iteration $j$ and any $r \in R$, there is at most one $i > j$ so that $r = r_i$ is charged to $p_j$, each accounting for an unsuccessful iteration. 
\end{proof}

\paragraph{Proof of \cref{lem:mpt-pd-pruning-time}}

We argue that the pruning step has running time $\textrm{poly}(k/\epsilon)$ to prove \cref{lem:mpt-pd-pruning-time}.

\begin{lemma}[\cref{lem:mpt-pd-pruning-time}]
	\label{lem:pd-pruning-time}
	The for-loop in \FnPrimalDual has running time $O((k/\epsilon)^4)$.
\end{lemma}
\begin{proof}
	Since $i \leq 2(2k/\epsilon)^2$, it holds that $\lvert S \rvert \leq \lvert \{ p \in P \mid \exists  r \in R : x^{(r)}_p > 0 \} \rvert \leq 2(2k/\epsilon)^2$. Therefore, sorting $S$ requires at most $O((2k/\epsilon)^2 \log (2k/\epsilon))$ time. In each iteration of the loop, it suffices to iterate over $\bar{S}$. Since $\lvert \bar{S} \rvert \leq \lvert S \rvert$, the claim follows.
\end{proof}

\paragraph{Proof of \cref{lem:pd-offline-time}}

\lemPdOfflineTime*
\begin{proof}
	Recall that we sort points in $S$ non-increasingly. Therefore, for every point $p$, there exists only one $(p,r) \in S$ that is inserted into $\bar{S}$. It follows that $\lvert \bar{S} \rvert \leq \lvert S \rvert \leq (2k/\epsilon)^2$ and since $\lvert \hat{S} \rvert \leq k$, it suffices to iterate over $\bar{S}$ for every $(\hat{p}, \hat{r}) \in \hat{S}$.
\end{proof}

\paragraph{Proof of \cref{lem:pd-cost-a}}

We turn to the feasibility and approximation guarantee of the solution that is computed by \cref{alg:pd}. First, we observe that the dual solution is always feasible.

\begin{lemma}
	\label{lem:always-dual-feasible}
	During the entire execution of \cref{alg:pd}, no dual constraint $(p, r)$ is violated.
\end{lemma}
\begin{proof}
	For the sake of contradiction, let $i$ be the first iteration of the while-loop in \FnPrimalDual after which there exists $(p,r)$ such that $\sum_{p' \in P : d(p, p') \leq r} > r + z$. From this definition, it follows that $d(p, p_i) \leq r$ as only $y_{p_i}$ has increased in iteration $i$. By the triangle inequality, for all $p' \in P$ so that $d(p,p') \leq r$, we have that $d(p_i,p') \leq 2r$. Therefore, the dual constraint $(p_i, 2r)$ is more than half-tight:
	\begin{align*}
		\sum_{p' \in P : d(p_i, p') \leq 2r} y_{p'}
		\geq \sum_{p' \in P : d(p, p') \leq r} y_{p'}
		> r + z
		= \frac{2r}{2} + z .
	\end{align*}
	This is a contradiction to the choice of $(p_i, r_i)$.
\end{proof}

The primal solution is also feasible to the LP, and its cost is bounded by $6\OPT'$.

\lemPdCostA*
\begin{proof}
	Since $U_i = \emptyset$, $x$ is feasible at the end of \cref{alg:pd}. By the construction of $\bar{S}$, it follows that $\bar{x}$ is feasible. By \cref{lem:always-dual-feasible}, $y$ is always feasibile for (D). It remains to bound the cost of $\bar{S}$.
	
	By the construction of $\bar{S}$, for any $(p, r) \in \bar{S}$, $x^{(r/3)}_p = 1$ and $(p, r/3)$ is at least half-tight in $y$. In other words, $r + z \le 6 \cdot (r/(2 \cdot 3) + z) = 6 \sum_{d(p,p') \leq r/3} y_p$. Let $\bar{S}(p,r) = \{ p' \in P \mid d(p, p') \leq r \}$. For all $(p_1,r_1), (p_2,r_2) \in \bar{S}$, $p_1 \neq p_2$, we have that $\bar{S}(p_1,r_1/3)$ and $\bar{S}(p_2,r_2/3)$ are disjoint due to the construction of $\bar{S}$. Therefore, it holds that
	\begin{equation*}
		\sum_{p\in P}\sum_{r\in R}\bar{x}_{p}^{(r)}(r+z)
		= \sum_{(p,r) \in \bar{S}} (r+z)
		\leq 6\cdot\sum_{p\in P}y_{p} .
	\end{equation*}
	By weak duality, $6\cdot\sum_{p\in P}y_{p} \leq 6\OPT'$.
\end{proof}

\paragraph{Proof of \cref{lem:pd-cost-b}}

Finally, we prove that the pruned solution is a feasible $k$-sum-of-radii solution with cost bounded by $(13.008 + \epsilon)\OPT'$.

\lemPdCostB*
\begin{proof}
	First observe that the cost of $\hat{S}$ is bounded by $2\cdot 3.504 \cdot \OPT$ since we can construct a solution with cost at most $2\OPT$ that covers $\bar{C}$ using only centers from $\bar{C}$: for each point $p \in \bar{C}$, take the point $p' \in \OPT$ covering $p$ with some radius $r'$, and select $p$ with  radius $2r'$.
	
	Let $p \in P$, let $(\bar{p}, \bar{r}) \in \bar{S}$ so that $d(p, \bar{p}) \leq \bar{r}$ and let $(\hat{p}, \hat{r}) \in \hat{S}$ that was chosen to cover $\bar{p}$. By the triangle inequality, $d(p, \bar{p}) \leq \hat{r} + \bar{r}$. 
	Let $(\hat{p}, \tilde{r})$ be the corresponding tuple in $\tilde{S}$. By the construction of $\tilde{S}$, we have that $\hat{r} + \bar{r} \leq \tilde{r}$. Therefore, $\tilde{S}$ is feasible. It follows that the cost of $\tilde{S}$ is at most
	
	\begin{equation*}
		\sum_{(\hat{p},\hat{r})\in \hat{S}}\hat{r}+\sum_{p\in P}\sum_{r\in R}\bar{x}_{p}^{(r)}(r+z)
		\leq (2\cdot3.504 + 6 + \epsilon) \OPT'. \qedhere
	\end{equation*}
\end{proof}

\paragraph{Proof of \cref{thm:pd-main}}

\thmPdMain*
\begin{proof}
	Consider a fixed choice of $\OPT'$. This assumption will be removed at the end of the proof. We invoke \cref{alg:pd} on the initial point set as for the static setting. Consider an operation $t$, and let $\breve{S}, \breve{U}, \breve{x}, \breve{y}, \breve{r}, \breve{i}$ be the state before this operation.
	
	If a point $p$ is inserted, the algorithm checks, for every $j \in \{ 1, \ldots, i-1 \}$, if $d(p_j, p) \leq 2r_i$. If this is not the case for any $j$, $p$ is added to $\breve{U_j}$ and the algorithm proceeds. Otherwise, the algorithm stops. If the last check fails and $i > (2k/\epsilon)^2$, the algorithm stops, too. Otherwise, it runs $\FnPrimalDual(P, \breve{U}, R, z, k, \epsilon, x, y, i)$ with the updated $\breve{U}$. In any case, the point $p$ deposits a budget of $2k/\epsilon$ tokens for each possible $U_i$, $i \in [2k/\epsilon]$, i.e., $(2k / \epsilon)^2$ tokens in total. By \cref{lem:pd-iteration-time,lem:pd-pruning-time}, the total running time is $O((2k/\epsilon)k/\epsilon + 0 + (k/\epsilon)^4 + (2k / \epsilon)^2) = O((k/\epsilon)^4)$.
	
	If a point $p$ is deleted, the algorithm deletes $p$ from all $\breve{U_i}$ it is contained in. If for all radii $r \in R$ we have that $\breve{x}^{(r)}_p = 0$ then the algorithms stops. Otherwise, let $j$ be so that $p \in \breve{U}_j \setminus{\breve{U}}_{j+1}$, i.e., $p$ is the center of the $j$\xth/ cluster. The algorithm sets, for all $r \in R$ and all $j' \geq j$, $\breve{x}^{(r)}_{p_{j'}} = 0$, $\breve{y}_{p_{j'}} = 0$ and, for all $j' > j$, $\breve{U}_{j'} = \emptyset$. Then, it calls $\FnPrimalDual(P, \breve{U}, R, z, k, \epsilon, \breve{x}, \breve{y}, j)$. By \cref{lem:pd-iteration-time,lem:pd-pruning-time}, the total running time is $O((2k/\epsilon)^2 k/\epsilon + (2k/\epsilon)\lvert U_j \rvert + (k/\epsilon)^4)$.
	
	The correctness of the algorithm follows from the fact that if $\OPT' \geq \OPT$, the algorithm produces a feasible solution irrespective of the choice of the $p_i$, $i \in [2k / \epsilon]$, by \cref{lem:pd-cost-b} and observing that the procedure described above simulates a valid run of the algorithm for $P = P_{t-1} \cup \{ p \}$ and $P = P_t \setminus \{ p \}$, respectively. Finally, we prove that the expected time to process all deletions up to operation $t$ is bounded by $O(t \cdot 2k/\epsilon)$. The argument runs closely along the running time analysis in~\cite{ChaFul18}.
	
	Let $t' \leq t$, $j \in [2k/\epsilon]$ and let $\bar{U}^{(t')}_j$ be the set $U^{(t')}_i$ that was returned by \FnPrimalDual after the last call that took place before operation $t'$ so that the argument $i$ is such that $i \leq j$. Note that this is the last call to \FnPrimalDual before operation $t'$ when $U_j$ is reclustered. We decompose $U^{(t')}_j$ into $A^{(t')}_j = U^{(t')}_j \setminus \bar{U}^{(t')}_j$ and $B^{(t')}_j = U_j \cap \bar{U}^{(t')}_j$ and define the random variable $T^{t'}_i$, where $T^{t'}_i = \lvert B^{(t')} \rvert$ if operation $t'$ deletes center $p_i$ and $T^{t}_i = 0$ otherwise. Next, we bound $E[\sum_{t' < t} \sum_{i \in [2k / \epsilon]} T^{(t')}_i]$. For $t' <t$ and $i \in [2k / \epsilon]$, consider $E[T^{(t')}_i]$. Since $p_i$ was picked uniformly at random from $B^{(t')}$, the probability that operation $t'$ deletes $p_i$ is $1 / \lvert B^{(t')} \rvert$. Therefore, $E[T^{t'}_i] = 1$. By linearity of expectation, $E[\sum_{t' < t} \sum_{i \in [2k / \epsilon]} T^{(t')}_i] \leq 2k/\epsilon$. If operation $t'$ deletes $p_j$, $U_j$ is reclustered at operation $t'$ and any point in $A$ is not in $A^{(T^{t'}_i)}$ for any $t'' \geq t'$. Therefore, each point $p \in A$ can pay $(2k/\epsilon)$ tokens from its insertion budget if $p_j$ is deleted. The expected amortized cost for all operations up to operation $t$ is therefore at most $O((2k/\epsilon)^2 k/\epsilon + (2k/\epsilon)^2 + (k/\epsilon)^4) = O((k/\epsilon)^4)$.

	Now, we remove the assumption that $\OPT'$ is known. Recall that $d(x,y) \geq 1$ for every $x,y \in X$, and $\dmax = \max_{x,y \in \cX} d(x,y)$. For every $\Gamma \in \{ (1+\epsilon)^i \mid i \in [\lceil \log_{1+\epsilon} (k \dmax) \rceil]$, the algorithm maintains an instance of the LP with $\costbound = (1+\epsilon)^\Gamma$. After every update, the algorithm determines the smallest $\Gamma$ for which a solution is returned. Recall that the algorithm from \cite{CharikarP04} takes time $O(n^{O(1/\epsilon)})$. The total expected amortized cost is $O(k^{O(1/\epsilon)} \log \dmax)$.
\end{proof}

\section{Lower Bound for Arbitrary Metrics}\label{sec:LB}
We now demonstrate that any algorithm which approximates the optimal $k$-center cost, in an arbitrary metric space of $n$ points, must run in $\Omega(nk)$ time. Specifically, the input to an algorithm for $k$-center in arbitrary metric spaces is both the point set $P$  \textit{and} the metric $d$ over the points. In particular, the input can be represented via the distance matrix  distance matrix $\bD \in \R^{n \times n}$ over the point set $P$, and the behavior of such an algorithm can be described by a sequences of adaptive queries to $\bD$. 

The above setting casts the problem of approximating the cost of the optimal $k$-center clustering as a \textit{property testing} problem  \cite{goldreich1998property,goldreich2017introduction}, where the goal is to solve the approximation problem while making a small number of queries to $\cD$. Naturally, the query complexity of such a clustering algorithm lower bounds its runtime, so to prove optimality of our dynamic $k$-center algorithms it suffices to focus only on the query complexity. In particular, in what follows we will demonstrate that any algorithm that approximates the optimal $k$-center cost to any non-trivial factor with probability $2/3$ must query at least $\Omega(nk)$ entries of the matrix.  In particular, this rules out any fully dynamic algorithm giving a non-trivial approximation in $o(k)$ amortized update time for general metric spaces. 

Moreover, we demonstrate that this lower bound holds for the $(k,z)$-clustering objective, which includes the well studied $k$-medians and $k$-means. Recall that this problem is defined as outputting a set $\cC \subset \cX$ of size at most $k$ which minimizes the objective function
\[  \cost[P]{\cC}[k][z] =   \sum_{p \in P} d^z(p,\ell(p))    \]
where $\ell(p)$ is the cluster center associated with the point $p$. 
Note that  $( \cost[P]{\cC}[k][z])^{1/z}$ is always within a factor of $n$ of the optimal $k$-center cost. Thus, it follows that if $k$-center cannot be approximated to any non-trivial factor (including factors which are polynomial in $n$) in $o(nk)$ queries to $\bD$, the same holds true for $(k,z)$-clustering for any constant $z$. Thus, in the proofs of the following results we focus solely on proving a lower bound for approximation $k$-center to any factor $R$, which will therefore imply the corresponding lower bounds for $(k,z)$-clustering. 

We do so by first proving Theorem \ref{thm:LBBig}, which gives a $\Omega(nk)$ lower bound when $n = \Omega(k \log k)$. Next, in Proposition \ref{prop:lb}, we prove a general $\Omega(k^2)$ lower bound for any $n > k$, which will complete the proof of Theorem \ref{thm:LB}. We note that the proof of Proposition \ref{prop:lb} is fairly straightforward, and the main challenge will be to prove Theorem \ref{thm:LBBig}.

\begin{theorem} \label{thm:LBBig}
Fix and $k \geq 1$ and $n > C k \log k$ for a sufficiently large constant $C$. Then any algorithm which, given oracle access the distance matrix $\cD \in \R^n$ of a set $X$ of $n$ points in a metric space, determines correctly with probability $2/3$ whether the optimal $k$-center cost on $X$ is at most $1$ or at least $R$, for any value $R >1$, must make at least $\Omega(k n)$ queries in expectation to $\cD$. The same bound holds true replacing the $k$-center objective with $(k,z)$-clustering, for any constant $z>0$. 
\end{theorem}
\begin{proof}
We suppose there exists such an algorithm that makes at most an expected $k n/8000$ queries to $\cD$. By forcing the the algorithm to output an arbitrary guess for $c$ whenever it queries a factor of $20$ more entries than its expectation, by Markov's inequality it follows that there is an algorithm which correctly solves the problem with probability $2/3 - 1/20 > 6/10$, and always makes at most $kn/400$ queries to $\cD$.

\paragraph{The Hard Distribution.}
We define a distribution $\cD$ over $n \times n$ distance matrices $\cD$ as follows. First, we select a random hash function $h:[n] \to [k]$, a uniformly random coordinate $i \sim [n]$. We then set $\bD(h)$ to be the matrix defined by 
$\bD_{p,q}(h) = 1$ for $p \neq q$ if $h(p) = h(q)$, and $\bD_{p,q}(h) = R$ otherwise, where $R$ is an arbitrarily large value which we will later fix. We then flip a coin $c \in \{0,1\}$. If $c=0$, we return the matrix $\bD(h)$, but if $c = 1$, we define the matrix $\bD(h,i)$ to be the matrix resulting from changing every $1$ in the $i$-th row and column of $\bD(h)$ to the value $2R$. It is straightforward to check that the resulting distribution satisfies the triangle inequality, and therefore always results in a valid metric space. We write $\cD_0 ,\cD_1$ to denote the distribution $\cD$ conditioned on $c=0,1$ respectively. In the testing problem, a matrix $\bD \sim \cD$ is drawn, and the algorithm is allowed to make an adaptive sequence of queries to the entries of $\bD$, and thereafter correctly determine with probability $2/3$ the value of $c$ corresponding to the draw of $\bD$. 

Note that a draw from $\cD$ can then be described by the values $(h,i,c)$, where $h \in \cH = \{h': [n] \to [k] \}$, $i\in[n]$, and $c \in \{0,1\}$. Note that, under this view, a single matrix $\bD \sim \cD_0$ can correspond to multiple draws of $(h,i,0)$. Supposing there is a randomized algorithm which is correct with probability $6/10$ over the distribution $\cD$ and its own randomness, it follows that there is a deterministic algorithm $\cA$ which is correct with probability $6/10$ over just $\cD$, and we fix this algorithm now. 

Let $(d_1,p_1),(d_2,p_2),\dots,(d_t,p_t)$ be an adaptive sequence of queries and observations made by an algorithm $\cA$, where $d_i \in \{1,R,2R\}$ is a distance and $p_i \in \binom{n}{2}$ is a position in $\cD$, such that the algorithm queries position $p_i$ and observed the value $d_i$ in that position.

\begin{claim}\label{claim:lb1}
There is an algorithm with optimal query vs. success probability trade-off which reports $c=1$ whenever it sees an entry with value $d_i = 2R$, otherwise it reports $c=0$ d if it never sees a distance of value $2R$. 
\end{claim}
\begin{proof}
 To see this, first note that if $c=0$, one never sees a value of $2R$, so any algorithm which returns $c=0$ after observing a distance of size $2R$ is always incorrect on that instance. 

For the second claim, suppose an algorithm $\cA$ returned that $c=1$ after never having seen a value of $2R$. Fix any such sequence  $S= \{(d_1,p_1),(d_2,p_2),\dots,(d_t,p_t)\}$ of adaptive queries and observations such that $d_i \neq 2R$ for all $i=1,2,\dots,t$. We claim that $\pr{c=0 | S} \geq \pr{c=1 |S}$. To see this, let $(h,i,1)$ be any realization of a draw from $\cD_1$, and note that $\pr{(h,i,1)} = \pr{(h,i,0)} = \frac{1}{2 n } k^{-n}$. Let $F_0(S)$ be the set of tuples $(h,i)$ such that the draw $(h,i,0)$ could have resulted in $S$, and $F_1(S)$ the set of tuples $(h, i)$ such that $(h,i,1)$  could have resulted in $s$. Let $(h,i,1)$ be a draw that resulted in $S$. Then $(h,i,0)$ also results in $S$, because the difference between the resulting matrices $\cD$ is supported only on positions which were initially $2R$ in the matrix generated by $(h,i,1)$. Thus $F_1(S) \subseteq F_2(S)$, which demonstrates that $\pr{c=0 | S} \geq \pr{c=1 |S}$. Thus the algorithm can only improve its chances at success by reporting $c=0$, which completes the proof of the claim.
\end{proof}

\paragraph{Decision Tree of the Algorithm.}
The adaptive algorithm $\cA$ can be defined by a $3$-ary decision tree $T$ of depth at most $ k n/400$, where each non-leaf node $v \in T$ is labelled with a position $p(v) = (x_v,y_v) \in [n] \times [n]$, and has three children corresponding to the three possible observations $\bD_{p(x)} \in \{1,R,2R\}$. Each leaf node contains only a decision of whether to output $c=0$ or $c=1$. For any $v \in T$, let $v_1,v_{R}, v_{2R}$ denote the three children of $v$ corresponding to the edges labelled $1,R$ and $2R$,
Every child coming from a ``$2R$'' edge is a leaf, since by the above claim the algorithm can be assumed to terminate and report that $c=1$ whenever it sees the value of $2R$. For any vertex $v \in T$ at depth $\ell$, let $S(v) = \{ (d_1,p_1),\dots,(\cdot, p(v))\}$ be the unique sequence of queries and observations which correspond to the path from the root to $v$. Note that the last entry $(\cdot,  p(v)) \in S(v)$ has a blank observation field, meaning that at $v$ the observation $p(v)$ has not yet been made.

For any $v \in T$ and $i \in [n]$, we say that a point $i$ is \textit{light} at $v$ if the number of queries $(d_j,p_j) \in S $ with $i \in p_j$ is less than $k/2$. If $i$ is not light at $v$ we say that it is \textit{heavy} at $v$. For any $i \in [n]$, if in the sequence of observations leading to $v$ the algorithm observed a $1$ in the $i$-th row or column, we say that $i$ is \textit{dead} at $v$, otherwise we say that $i$ is \textit{alive}. We write $\pr{v}$ to denote the probability, over the draw of $\bD \sim \cD$, that the algorithm traverses the decision tree to $v$, and $\pr{v | \; c= b}$ for $b \in \{0,1\}$ to denote this probability conditioned on $\bD \sim \cD_b$.  Next, define $F_b(v) = F_b(S(v))$ for $b \in \{0,1\}$, where $F_b(S)$ is as above. 
Note that if $(h,i) \in F_0(v)$ for some $i \in [n]$, then $(h,j) \in F_0(v)$ for all $j \in [n]$, since the matrices generated by $(h,i,0)$ are the same for all $i \in [n]$. Thus, we can write $h \in F_0(v)$ to denote that $(h,i) \in F_0(v)$ for all $i \in [n]$.

\begin{claim}\label{claim:lb2}
Let $v \in T$ be a non-leaf node where at least one index $i$ in $p(v) = (i,j)$ is alive and light at $v$. Then we have \[\prb{\bD \sim \cD}{\bD_{p(v)} = 1 | S(v), c=0} \leq \frac{2}{k}\]
\end{claim}
\begin{proof}
Fix any function $h \in \cH$ such that $h \in F_0(v)$: namely, $h$ is consistent with the observations seen thus far. 
Let $h_1,\dots,h_k \in \cH$ be defined via $h_t(j) = h(j)$ for $j \neq i$, and $h_t(i) = t$, for each $t \in [k]$. We claim that $h_t \in F_0(v)$ for at least $k/2$ values of $t$. 
To show this, first note that the values of $\{h(j)\}_{j \neq i}$ define a graph with at most $k$ connected components, each of which is a clique on the set of values $j \in [n] \setminus \{i\}$ which map to the same hash bucket under $h$. 
The only way for $h_t \notin F_0(v)$ to occur is if an observation $(i,\ell)$ was made in $S(v)$ such that $h(\ell) = t$. Note that such an observation must have resulted in the value $R$, since $i$ is still alive (so it could not have been $1$). In this case, one knows that $i$ was not in the connected component containing $\ell$. However, since $i$ is light, there have been at most $k/2$ observations involving $i$ in $S(v)$. Each of these observations eliminate at most one of the $h_t$'s, from which the claim follows. 

Given the above, it follows that if at the vertex $v$ we observe $\bD_{p(v)} = \bD_{i,j} = 1$, then we eliminate every $h_t$ with $t \neq h(j)$ and $h_t \in F_0(v)$. Since for every set of values $\{h(j)\}_{j \neq i}$ which are consistent with $S(v)$ there were $k/2$ such functions $h_t \in F_0(v)$, it follows that only a $2/k$ fraction of all $h \in \cH$ result in the observation $\bD_{p(v)} = 1$. Thus, $|\cF_0(v_1)| \leq \frac{2}{k}|\cF_0(v)|$, which completes the proof of the proposition.
\end{proof}

Let $\cE_1$ be the set of leafs $v$ which are children of a $2R$ labelled edge, and let $\cE_0$ be the set of all other leaves. Note that we have $\pr{v \; | \; c=0} = 0$ for all $v \in \cE_1$, and moreover $\sum_{v \in \cE_0} \pr{v | c=0} = 1$. For $v \in T$, let $\theta(v)$ denote the number of times, on the path from the root to $v$, an edge $(u,u_1)$ was crossed where at least one index $i \in p(u)$ was alive and light at $u$. Note that such an edge kills $i$, thus we have $\theta(v) \leq n$ for all nodes $v$. Further, define $\hat{\cE}_0 \subset \cE_0$ to be the subset of vertices $v \in \cE_0$ with $\theta(v) < n/20$. We now prove two claims regarding the probabilities of arriving at a leaf $v \in \hat{\cE}_0$. 

\begin{claim}\label{claim:lb3}
Define $\hat{\cE}_0$ as above. Then we have 
\[  \sum_{v \in \hat{\cE}_0} \pr{v | c=0} > 9/10 \]
\end{claim}
\begin{proof}
We define indicator random variables $\bX_1,\bX_2,\dots,\bX_t \in \{0,1\}$, where $t \leq  kn/400$ is the depth of $T$, such that $\bX_i = 1$ if the $i$-th observation made by the algorithm causes a coordinate $i \in [n]$, which was prior to observation $i$ both alive and light, to die, where the randomness is taken over a uniform draw of $\bD \sim \cD_0$. Note that the algorithm may terminate on the $t'$-th step for some $t'$ before the $t$-th observation, in which case we all trailing variables $\bX_{t'},\dots,\bX_t$ to $0$. By Claim \ref{claim:lb2}, we have $\ex{\bX_i} \leq 2k$ for all $i \in [t]$, so $\ex{\sum_{i \in [t]} \bX_i} <  n/ 200$. By Markov's inequality, we have $\sum_{i \in [t]} \bX_i < n/20$ with probability at least $9/10$. Thus with probability at least $9/10$ over the draw of $\bD \sim \cD_0$ we land in a leaf vertex $v$ with $\theta(v) < n/20$, implying that $v \in \hat{\cE}_0$ as needed. 
\end{proof}

\begin{claim}\label{claim:lb4}
For any $v \in \hat{\cE_0}$, we have $  \pr{v \; | \; c=1} > (9/10 )\pr{v \; | \; c=0}$.
\end{claim}
\begin{proof}
Fix any $v \in \hat{\cE}_0$. By definition, when the algorithm concludes at $v$, at most $n/5$ indices were killed while having originally been light. Furthermore, since each heavy index requires by definition at least $k/2$ queries to become heavy, and since each query contributes to the heaviness of at most $2$ indices, it follows that at most $kn/400 (4/k) = n/100$ indices could ever have become heavy during any execution. Thus there are at least $n - n/20 - n/100 > (9/10)n$ indices $i$ which are both alive and light at $v$. 

Now fix any $h \in \cF_0(v)$. We show that $(h,i) \in \cF_1(v)$ for at least $(9/10)n$ indices $i \in [n]$, which will demonstrate that $|\cF_1(v)| > (9/10)  |\cF_0(v)|$, and thereby complete the proof. In particular, it will suffice to show that is true for any $i \in [n]$ which is alive at $v$. To see why this is the case, note that by definition if $i$ is alive at $v \in \cE_0$, then $S(v)$ includes only observations in the $i$-th row and column of $\cD$ which are equal to $R$. It follows that none of these observations would change if the input was instead specified by $(h,i,1)$ instead of $(h,j,0)$ for any $j \in [n]$, as the difference between the two resulting matrices are supported on values where $\cD$ is equal to $1$ in the $i$-th row and column of $\cD$. Thus if $i$ is alive at $v$, we have that $(h,i) \in \cF_1(v)$, which completes the proof of the claim.
\end{proof}
\noindent Putting together the bounds from Claims \ref{claim:lb3} and \ref{claim:lb4}, it follows that 
\[  \sum_{v \in \hat{\cE}_0} \pr{v\; |\;  c=1} > (9/10)^2 = .81 \]
Moreover, because by Claim \ref{claim:lb1} the algorithm always outputs $c=0$ when it ends in any $v \in \cE$, it follows that the algorithm incorrectly determined the value of $c$ with probability at least $.81$ conditioned on $c=1$, and therefore is incorrect with probability at least $.405 > 4/10$ which is a contradiction since $\cA$ was assumed to have success probability at least $6/10$. 

\paragraph{From the Hard Distribution to $k$-Centers.} 

To complete the proof, it suffices to demonstrate that the optimal $k$-center cost is at most $1$ when $\bD \sim \bD_0$, and at least $R$ when $\bD \sim \bD_1$. The first case is clear, since we can choose at least one index in the pre-image of $h^{-1}(t)\subseteq [n]$ for each $t \in [k]$ to be a center. For the second case, note that conditioned on $|h^{-1}(t)| \geq 2$ for all $t \in [k]$, the resulting metric contains $k+1$ points which are pairwise-distance at least $R$ from each other. In particular, for the resulting metric, at least one point must map to a center which it is distance at least $R$ away from, and therefore the cost is at least $R$. Now since $n = \Omega( k \log k)$ with a sufficently large constant, it follows by the coupon collector's argument that with probability at least $1/1000$, we have that $|h^{-1}(t)| \geq 2$ for all $t \in [k]$. Moreover, that the $1/1000$ probability under which does not occur can be subtracted into the failure probability of $.405$ in the earlier argument, which still results in a $.404 > 4/10$ failure probability, and therefore leads and leading to the same contradiction, which completes the proof. 
\end{proof}

\begin{proposition}\label{prop:lb}
Fix any $1 \leq k < n$. Then any algorithm which, given oracle access the distance matrix $\cD \in \R^n$ of a set $X$ of $n$ points in a metric space, determines correctly with probability $2/3$ whether the optimal $k$-center cost on $X$ is at most $1$ or at least $R$, for any value $R >1$, must make at least $\Omega(k^2)$ queries in expectation to $\cD$. The same bound holds true replacing the $k$-center objective with $(k,z)$-clustering, for any constant $z>0$. 
\end{proposition}
\begin{proof}
By the same arguements given in Theorem \ref{thm:LBBig}, one can assume that the existence of such an algorithm that would violate the statement of the proposition implies the existence of a deterministic algorithm which always makes at most $c k^2$ queries to $\cD$ and is correct with probability $3/5$, for some arbitrarily small constant $c$. 
In what follows, we assume $n=k+1$, and for larger $n$ we will simply add $n-(k+1)$ duplicate points on top of the first point in the following distribution; note that any algorithm for the dataset with the duplicate point can be simulated, with no increase in query complexity, via access to the distance matrix on the first $k+1$ points.

The hard instance is as then as follows.
With probability $1/2$, we give as input the distance matrix $\bD \in \R^{k+1 \times k+1} $ with $\bD_{i,j} =R $ for all $i \neq j$. Note that any $k$-center solution must have one of the points in a cluster centered at another, and therefore the optimal $k$-center cost is at least $R$ for this instance. In the second case, the input is $\bD$ but with a single entry $\bD_{i,j} = \bD_{j,i} = 1$, where $(i,j)$ is chosen uniformly at random. Note that the result is still a valid metric space in all cases. Moreover, note that the optimal $k$-center cost is $1$, and is obtained by choosing all points except $i$ (or alternatively except $j$). 

By the same (and in fact simplified) argument as in Claim \ref{claim:lb1}, we can assume the algorithm returns that the $k$-center cost is at most $1$ if and only if it sees an entry with value equal to $1$. Since the algorithm is deterministic, and since the only distance other than $1$ is $R$, we can define a deterministic set $S$ of $c k^2$ indices in $\binom{k+1}{2}$ such that the adaptive algorithm would choose exactly the set $S$ if, for every query it made, it observed the distance $R$ (and therefore would return that the $k$-center cost was at most $1$ at the end). Then in the second case, the probability that $(i,j)$ is contained in $S$ is at most $\frac{4}{c}$. Setting $c > 40$, it follows that the algorithm is incorrect with probability at least $1/2-\frac{4}{c} > 2/5$, contradicting the claimed success probability of $3/5$, and completing the proof. 
\end{proof}

\section{Lower Bound Against Adaptive Adversaries}
\label{sec:LBadap}

In this section we present our lower bounds on the approximation
guarantees of dynamic algorithms for $k$-center, $k$-median, $k$-means,
$k$-sum-of-radii, $k$-sum-of-diameters, and $(k,z)$-clustering against
a \emph{metric-adaptive} adversary. Recall that a point-adaptive adversary needs to fix the metric space in advance, while a metric-adaptive adversary only needs to answer a distance query consistently with all answers it gave previously (so the metric itself is chosen adaptively). First,
we define a generic strategy for an adversary that in each operation
creates or deletes a point and answers the distance queries of the
algorithm. Then we show how to derive lower bounds for all of our
problems from this single strategy. Some formal details of the
analysis are deferred to \cref{sec:appendix-lower-adaptive}.

\subsection{Strategy of the Adaptive Adversary}

Let $f(k,n)$ be a positive function that is, for every fixed $k$, non-decreasing in $n$. Suppose that there is an algorithm (for any of the problems under
consideration) that in an amortized sense queries the distances
between at most $f(k,n)$ pairs of points per operation of the adversary,
where $n$ denotes the number of points at the beginning of the respective
operation. Note that any algorithm with an amortized update time
of at most $f(k,n)$ fulfills this condition. To determine the distance
between two points the algorithm asks a distance query to the adversary. 
We present now an adversary ${\cal A}$ whose goal
is to maximize the approximation ratio of the algorithm. To
record past answers and to give consistent answers, ${\cal A}$ maintains
a graph $G=(V,E)$ which contains a vertex $v_{p}\in V$ for each
point $p$ that has been inserted previously (including points that
have been deleted already). Intuitively, with the edges in $E$
the adversary keeps track of previous answers to distance queries.
Each vertex $v_{p}$ is labeled as \emph{open, closed, }or \emph{off}.
If a point $p$ has not been deleted yet, then its vertex $v_{p}$
is labeled as open or closed. Once point $p$ is deleted, then
$v_{p}$ is labeled as off. Intuitively, if $p$ has not been deleted
yet, then $v_{p}$ is open if it has small degree and closed if
it has large degree. In the latter case, ${\cal A}$ will delete $p$
soon. All edges in $G$ have length 1, and for two vertices $v,v'\in V$
we denote by $d_{G}(v,v')$ their distance in $G$.

When choosing the next update operation, ${\cal A}$ checks whether there
is a closed vertex $v_{p}$. If yes, ${\cal A}$ picks an arbitrary
closed vertex $v_{p}$, deletes the corresponding point $p$, and
labels the vertex $v_{p}$ as off. Otherwise, it adds a new point
$p$, adds a corresponding vertex $v_{p}$ to $G$, and labels $v_{p}$
as open.

Suppose now that the algorithm queries the distance $d(p,p')$ for
two points $p,p'$ while processing an operation. Note that $p$ and/or
$p'$ might have been deleted already. If both $v_{p}$ and $v_{p'}$
are open then ${\cal A}$ reports to the algorithm that $d(p,p')=1$
and adds an edge $\{v_{p},v_{p'}\}$ to $E$. Intuitively, due to
the edge $\{v_{p},v_{p'}\}$ the adversary remembers that it reported
the distance $d(p,p')=1$ before and ensures that in the future it
will report distances consistently. Otherwise, ${\cal A}$ considers
an augmented graph $G'$ which consists of $G$ and has
in addition an edge $\{v_{\bar{p}},v_{\bar{p}'}\}$ of length 1
between any pair of open vertices $\bar{p},\bar{p}'$. The adversary
computes the shortest path $P$ between $p$ and $p'$ in $G'$ and
reports that $d(p,p')$ equals the length of $P$. If $P$ uses an
edge between two open vertices $\bar{p},\bar{p}'$, then ${\cal A}$
adds the edge $\{v_{\bar{p}},v_{\bar{p}'}\}$ to $G$. Note that $P$
can contain at most one edge between two open vertices since it
is a shortest path in a graph in which all pairs of open vertices
have distance~1. %
Observe that if both $v_{p}$ and $v_{p'}$ are open then this procedure
reports that $d(p,p')=1$ and adds an edge $\{v_{p},v_{p'}\}$ to
$E$ which is consistent with our definition above for this
case. If a vertex $v_{p}$ has degree at least $100f(k,t)$ for the current operation $t$, then $v_{p}$ is labeled as closed. A closed vertex never becomes open again.

In the next lemma we prove some properties about this strategy of
${\cal A}$. For each operation $t$, denote by $G_{t}=(V_{t},E_{t})$
the graph $G$ at the beginning of the operation $t$. 
Recall that the value of $n$ right before the operation $t$ (which is  the number of
current points) equals the number of open and closed
vertices in $G_{t}$.
We show that
the the number of open vertices is $\Theta(t)$, each vertex has
bounded degree, and there exist arbitrarily large values $t$
such that in $G_{t}$ there are no closed vertices (i.e., only open
and off vertices). \begin{lemma} \label{lem:mpt-det-low-adv-prop}
	For every operation $t>0$ the strategy of the adversary ensures the
	following properties for $G_{t}$
	\begin{enumerate}
		\item the number of open vertices in $G_{t}$ is at least $92t/100$,
		\label{it::mpt-det-low-adv-prop-numac} 
		\item each vertex in $G_{t}$ has a degree of at most $101f(k,n)$, \label{it::mpt-det-low-adv-prop-deg} 
		\item there exists an operation $t'$ with $t < t' \leq 2t$ such that $G_{t'}$
		contains only open and off vertices, but no closed vertices. \label{it::mpt-det-low-adv-prop-clean} 
	\end{enumerate}
\end{lemma} We say that an operation $t\in\N$ is a \emph{clean operation}
if in $G_{t}$ there are no closed vertices. For any clean operation
$t$, denote by $\bar{G}_{t}=(\bar{V}_{t},\bar{E}_{t})$ the subgraph
of $G_{t}$ induced by the open vertices in $V_{t}$.

\vspace{-0.2cm}

\paragraph*{Consistent metrics. }

The algorithm does not necessarily know the complete metric of the
given points, it knows only the distances reported by the adversary.
In particular, there might be many possible metrics that are consistent
with the reported distances. For each $t\in\N$ denote by $Q_{t}$
the points that were inserted before operation $t$, including
all points that were deleted before operation $t$, and let
$P_{t}\subseteq Q_{t}$ denote the points in $Q_{t}$ that are not
deleted. Given a metric $M$ on any point set $P'$, for all pairs
of points $p,p'\in P'$ we denote by $d_{M}(p,p')\ge0$ the distance
between $p$ and $p'$ according to $M$. For any $t\in\N$ we say
that a metric $M$ for the point set $Q_{t}$ is \emph{consistent}
if for any pair of points $p,p'\in Q_{t}$ for which the adversary
reported the distance $d(p,p')$ before operation~$t$, it
holds that $d(p,p')=d_{M}(p,p')$. In particular, any consistent metric
might be the true underlying metric for the point set $Q_{t}$.

The key insight is that for each clean operation $t$, we can build
a consistent metric with the following procedure. Take the graph $G_{t}$
and insert an arbitrary set of edges of length $1$ between pairs
of open vertices (but no edges that are incident to off vertices),
and let $G'_{t}$ denote the resulting graph. Let $M$ be the shortest
path metric according to $G'_{t}$. If a metric $M$ for $Q_{t}$
is constructed in this way, we say that $M$ is an \emph{augmented
	graph metric for }$t$.

\begin{lemma} \label{lem:metric-consistent}Let $t\in\N$ be a clean
	operation and let $M$ be an augmented graph metric for $t$. Then
	$M$ is consistent. 
\end{lemma} 
In particular, there are no shortcuts
via off vertices in $G_{t}$ that could make the metric $M$
inconsistent.

We fix a clean operation $t\in\N$. We define some metrics that are consistent
with $Q_{t}$ that we will use later for the lower bounds for our
specific problems. The first one is the ``uniform'' metric $M_{\mathrm{uni}}$
that we obtain by adding to $G_{t}$ an edge between \emph{each} pair
of open vertices in $G_{t}$. As a result, $d_{M_{\mathrm{uni}}}(p,p')=1$
for any $p,p'\in P_{t}$. 
\begin{lemma} \label{lem:muni-consistent}
	For each clean operation $t$ the corresponding metric $M_{\mathrm{uni}}$
	is consistent. 
\end{lemma} 
In contrast to $M_{\mathrm{uni}}$, our
next metric ensures that there are distances of up to $\Omega(\log n)$
between some pairs of points. Let $p^{*}\in P_{t}$ be a point such
that $v_{p^{*}}$ is open. For each $i\in\N$ let $V^{(i)}\subseteq V_{t}$
denote the open vertices $v\in V_{t}$ with $d_{G_{t}}(v_{p^{*}},v)=i$,
and let $V^{(n)}\subseteq V_{t}$ denote the vertices in $G_{t}$
that are in a different connected component than $p^{*}$. Since
the vertices in $G_{t}$ have degree at most $100 f(k,n)$, more than half of all
vertices are in sets $V^{(i)}$ with $i\ge\Omega(\log n/ \log f(k,n))$.

We define now a metric $M(p^{*})$ as the shortest path metric in
the graph defined as follows. We start with~$G_{t}$; for each $i,i'\in\N$
we add to $G_{t}$ an edge $\{v_{p},v_{p'}\}$ between any pair of
vertices $v_{p}\in V^{(i)}$, $v_{p'}\in V^{(i')}$ such that $|i-i'|\le1$.
As a result, for any $i,i'\in\N$ and any $v_{p}\in V^{(i)}$, $v_{p'}\in V^{(i')}$
we have that $d_{M(p^{*})}(p,p')=\max\{|i-i'|,1\}$, i.e., $d_{M(p^{*})}(p,p')=1$
if $i=i'$ and $d_{M(p^{*})}(p,p')=|i-i'|$ otherwise. 
\begin{lemma}
	\label{lem:mstar-consistent} For each clean operation $t$ and each
	point $p^{*}\in V_{t}$ the metric $M(p^{*})$ is consistent. 
\end{lemma}
For any two thresholds $\ell_{1},\ell_{2}\in\N_{0}$ with $\ell_1 < \ell_2$  we define a metric
$M_{\ell_{1},\ell_{2}}(p^{*})$ (which is a variation of $M(p^{*})$)
as the shortest path metric in the following graph. Intuitively,
we group the vertices in $\bigcup_{i=0}^{\ell_{1}}V^{(i)}$ to one
large group and similarly the vertices in $\bigcup_{i=\ell_{2}}^{\infty}V^{(i)}$.
Formally, in addition to the edges defined for $M(p^{*})$, for each
pair of vertices $v_{p}\in V^{(i)}$, $v_{p'}\in V^{(i')}$ we add
an edge $\{v_{p},v_{p'}\}$ if $i\le i'\le\ell_{1}$ or $\ell_{2}\le i\le i'$.%

\begin{lemma} \label{lem:mstarrange-consistent} For each clean operation
	$t$, each point $p^{*}\in V_{t}$, and each $\ell_{1},\ell_{2}\in\N_{0}$
	the metric $M_{\ell_{1},\ell_{2}}(p^{*})$ is consistent. \end{lemma}

\paragraph*{Lower bounds. }

Consider a clean operation $t$. The algorithm cannot distinguish
between $M_{\mathrm{uni}}$ and $M(p^{*})$ for any $p^{*}\in P_{t}$.
In particular, for the case that $k=1$ (for any of our clustering
problems) the algorithm selects a point $p^{*}$ as the center, and
then for each $i$ it cannot determine whether the distance
of the points in $V^{(i)}$ to $p^{*}$ equals 1 or $i$.
However, there are at least $n/2$ points in sets $V^{(i)}$
with $i\ge\Omega(\log n/\log f(k,n)))$ and hence they contribute a large amount
to the objective function value. This yields the following lower bounds,
already for the case that $k=1$. With more effort, we can show them even for bi-criteria approximations,
i.e., for algorithms that may output $O(k)$ centers, but where the approximation
ratio is still calculated with respect to the optimal cost on $k$
centers.

For $k$-center the situation changes
if the algorithm does
not need to be able to report an upper bound of the value of its computed
solution (but only the solution itself), since if $k=1$, then any
point is a 2-approximation. However, for arbitrary $k$ we can argue
that there must be $3k$ consecutive sets $V^{(i)},V^{(i+1)},...,V^{(i+3k-1)}$
such that the algorithm does not place any center on any point corresponding
to the vertices in these sets and hence incurs a cost of at least
$3k/2$. On the other hand, for the metric $M_{i,i+3k-1}(p^{*})$ the
optimal solution selects one center from each set $V^{(i+1)},V^{(i+4)},V^{(i+7)},...$
which yields a cost of only 1. %

\thmDetLB*

\subsection{Deferred Proofs}
\label{sec:appendix-lower-adaptive}

We introduce some formal notation and definitions we use to revisit the adversarial strategy that generates an input stream and answers distance queries on the set of currently known points. Then, we derive lower bound constructions for the aforementioned problems that are based on the metric space defined by the stream and the answers to the algorithm.

\SetKwFunction{FnGenerate}{GenerateStream}
\SetKwFunction{FnAnswer}{AnswerQuery}

We describe a strategy for an adversary $\mathcal{A}$ that generates a stream of update operations~$\s$ and answers distance queries~$\q$ on pairs of points by any dynamic algorithm with a guarantee on its amortized  complexity. In the following presentation, the adversary constructs the underlying metric space ad hoc. More precisely, the adversary constructs two metric spaces simultaneously that cannot be distinguished by the algorithm and its queries. All subsequent lower bounds stem from the fact that the problem at hand has different optimal costs on the input for the two metric spaces. When the algorithm outputs a solution, the adversary can fix a metric space that induces high cost for the centers chosen by the algorithm.

During the execution of the algorithm, the adversary maintains a graph $\kg$.
Each point that was inserted by the adversary is represented by a node in $\kg$. All query answers given to the algorithm by the adversary can be derived from 
$G$
using the shortest path metric $\dsp[\kg]{\cdot}{\cdot}$ on $\kg$. We denote the algorithm's $i$\xth/ query after update operation $t$ by $\q[t][i]$, and the adversary's answer by $\ans{\q[t][i]}$. For $t > 0$, we denote the number of queries asked by the algorithm between the $t$\xth/ and the $(t+1)$\xth/ update operation by $c(t)$. If $t$ is clear from context, we simplify notation and write $c := c(t)$. Note that, in this section, we use a slightly extended notation when indexing graphs when compared to other sections. Details follow.

We number the update operations consecutively starting with 1 using index $t$ and after each update operation, we index the distance queries that the algorithm issues while processing the update operation and the immediately following value- or solution-queries using index $i$. Let $c > 0$ and let $\kg[0][c(0)]$ be the empty graph.  For every $t > 0, i \in [c(t)]$, consider $i$-th query issued by the algorithm processing the $t$-th update operation.  The graph $\kg[t][i]$ has the following structure. For every point $x$ that is inserted in the first $t$ operations, $\V{\kg[t][0]}$ contains a node $x$. All edges have length $1$, and it holds that $\V{\kg[t][i]} \supseteq \V{\kg[t][i-1]}$ and $\E{\kg[t][i]} \supseteq \E{\kg[t][i-1]}$.

Edges are inserted by the adversary as detailed below. Let $\preceq$ denote the predicate that corresponds to the lexicographic order. In particular, the adversary maintains the following invariant, which is parameterized by the update operation $t$ and the corresponding query $i$: for all $(t',j) \preceq (t,i)$ and $(u,v) \defeq \q[t'][j]$, $\ans{\q[t'][j]} = \dsp[\kg[t][i]]{u}{v}$. In other words, any query given by the adversary remains consistent with the shortest path metric on all versions of 
$G$
after the query was answered. The adversary distinguishes the following types of nodes in $\kg[t][i]$ to answer a query. Recall that $f \defeq f(k,n)$ is an upper bound on the amortized complexity per update operation of the algorithm, which is non-decreasing in $n$ for fixed $k$.

\begin{definition}[type of nodes]
	\label{def:node_types}
	Let $t \geq 0, i \in [c]$ and let $u \in \V{\kg[t][i]}$. If $u$ has degree less than $100f(k,i)$ for all $i \in [t]$, it is \emph{open} after update $t$, otherwise it is \emph{closed}. In addition, the adversary can mark closed nodes as \emph{off}. We denote the set of open, closed and off nodes in $\kg[t][i]$ by $\actno[t][i]$, $\pasno[t][i]$ and $\disno[t][i]$, respectively.%
\end{definition}

For operation $t$, the adversary answers the $i$\xth/ query $\q[t][i]$ according to the shortest path metric on $\kg[t][i-1]$ with the additional edge set $\actno[t][i-1] \times \actno[t][i-1]$. In other words, the adversary (virtually) adds edges between all open nodes in $\kg[t][i-1]$ and reports the length of a shortest path between the query points in the resulting graph. After the adversary answered query $\q[t][i]$, the resulting %
graph $\kg[t][i]$ is $\kg[t][i-1]$ plus the (unique) edge $e$ of the shortest path between two open nodes that is not in $\kg[t][i-1]$ if such edge exists. If the connected component of $e$ in the resulting graph does not contain any open node, we also add an edge between the connected component and a node with degree at most $50f(k,i)$. Thus, the algorithm maintains the invariant that each connected component has at least one open vertex (see \cref{lem:kg-actnum} for details). A key element of our analysis is that \empty{all} answers up to operation $t$ and query $i$ are equal to the length of the shortest paths between the corresponding query points in $\kg[t][i]$. The generation of the input stream and the answers to all queries are formally given by \cref{alg:stream} and~\ref{alg:answer}, respectively.

\begin{algorithm}
	\Fn{\FnGenerate{$t$}}{
		\If{there exists a closed node $x \in \V{\kg[t-1][c]}$}{
			mark $x$ as off in $\kg[t][0]$\;
			\Return{$\langle$ delete $x$ $\rangle$}
		}
		\Else{
			let $x$ be a new point, i.e., that was not returned by the adversary before\;
			\Return{$\langle$ insert $x$ $\rangle$}
		}
	}
	\caption{\label{alg:stream} Construction of element $\s[t]$ of $\s$}
\end{algorithm}
\begin{algorithm}
	\Fn{\FnAnswer{$\q[t][i] = (x,y)$}}{
		let $\akg[t][i] = (\V{\kg[t][i-1]}, \E{\kg[t][i-1]} \cup (\actno[t][i-1] \times \actno[t][i-1]))$\;
		let $p = (e_1, \ldots, e_k)$ be a shortest path between $x$ and $y$ in $\akg[t][i]$\;
		set $\kg[t][i] \defeq \kg[t][i-1]$\;
		\If{$\exists e_i \notin \E{\kg[t][i-1]}$}{
			insert $e_i$ into $\kg[t][i]$ \label{lin:alg_answer_instert_a} \;
			let $C$ be the connected component of $e_i$ in $\kg[t][i]$ \;
			\If(\tcp*[f]{make sure $C$ contains an open node}){$C \cap \actno[t][i] = \emptyset$}{
				let $u = \arg\min_{u' \in C} \deg(u')$ \tcp*{pick node with degree $100f(k,t)$}
				let $v = \arg\min_{v' \in \actno[t][i]} \deg(v')$ \tcp*{pick node with degree at most $50f(k,t)$}
				insert $(u,v)$ into $\kg[t][i]$ \label{lin:alg_answer_instert_b} \;
			}
		}
		\Return{length of $p$}
	}
	\caption{\label{alg:answer} Answer of the adversary to query $\q[t][i]$}
\end{algorithm}

\subsubsection{Adversarial Strategy}

Let $n_t$ be the number of open and closed nodes, i.e., the number of current points for the algorithm, after operation $t$. In the whole section we use the notations  of $\kg[t'][i], \actno[t][i]$ etc.  from Definition~\ref{def:node_types}.

\paragraph{Proof of \cref{lem:mpt-det-low-adv-prop}}

The next three lemmas prove the three claims in \cref{lem:mpt-det-low-adv-prop}.

\begin{lemma}[\cref{lem:mpt-det-low-adv-prop} (\ref{it::mpt-det-low-adv-prop-numac})]
	\label{lem:kg-actnum}
	For every $t > 0$, the number of open nodes in $\kg[t][0]$ is at least $92t / 100$, and for every $t > 0, j \geq 0$, there exists at least one node with degree at most $50f(k,t)$ in $\kg[t][j]$.
\end{lemma}
\begin{proof}
	Recall that $f(k,n)$ is a positive function that is non-decreasing in $n$. We prove the first part of the claim by induction. By the properties of $f$, the case $t = 1$ follows trivially. Let $t \geq 2$. For any $i \in [t]$, the algorithm's query budget increases by $f(k, n_i)$ queries after the $i$\xth/ update. Since $f(k,i)$ is non-decreasing, nodes inserted after update operation $i-1$ can only become closed if their degrees increase to at least $100f(k,i)$. Answering a query $i$ inserts at most two edges into $\kg[t][i-1]$, and the sum of degrees increases by at most four. Therefore, it holds that $\pasnum[t][0] \leq \sum_{i \in [t]} 4f(k, n_i) / (100f(k, i)) \leq \sum_{i \in [t]} 4f(k, i) / (100f(k, i)) \leq 4t / 100$. It follows that the adversary will delete at most $4t / 100$ points in the first $t$ operations and insert points in the other at least $(1-4/100)t$ operations. The number of open nodes after update $t$ is $\actnum[t][0] \geq t - \pasnum[t][0] - 4t / 100 \geq 92t / 100$.
	
	To prove the second part of the claim, let $s_{t,j}$ be the number of nodes in $\kg[t][j]$ with degree at most $50f(k,t)$. Similarly as before, we have $n_t - s_{t,j} \leq \sum_{i \in [t]} 4f(k, n_i) / (50f(k, i)) \leq \sum_{i \in [t]} 4f(k, i) / (50f(k, i)) \leq 4t / 50$. Therefore, $s_{t,j} \geq n_t - 4t / 50 \geq (t - 4t / 100) - 4t / 50 \geq 1$.
\end{proof}

\begin{lemma}[\cref{lem:mpt-det-low-adv-prop} (\ref{it::mpt-det-low-adv-prop-deg})]
	\label{lem:max-degree}
	For every $t > 0, i \in [c]$, all nodes have degree at most $100f(k,t) + 1$ in $\kg[t][i]$.
\end{lemma}
\begin{proof}
	By definition, the claim is true for open nodes. Edges are only inserted into $G$
	if the algorithm queries for the distance between two nodes $x,y$ and the adversary determines a shortest path between $x$ and $y$ that contains edges that are not present in $\kg[t][i-1]$ (see \cref{alg:answer}). Since all such edges are edges between open nodes, only degrees of open nodes in $\kg[t][i]$ increase. The adversary finds a shortest path on a supergraph of $\actno[t][i] \times \actno[t][i]$. Therefore, any shortest path it finds contains at most one edge with two open endpoints. It follows that a query increases the degree of any open node in $\kg[t][i]$ by at most one, which may turn it into a closed node with degree $100f(k,t)$. If the connected component of this edge contains no open node, it also inserts an edge from an open node (with degree at most $50 f(k,t)$) to the vertex with smallest degree in the component. Since the vertex of the component that became closed most recently always has degree $100f(k,t)$, this increases the degree of every closed node at most once by $1$.
\end{proof}

\begin{lemma}[\cref{lem:mpt-det-low-adv-prop} (\ref{it::mpt-det-low-adv-prop-clean})]
	\label{lem:no-closed-nodes}
	For every $t > 0$, there exists a clean update operation $t'$, $t < t' \le 2t$, i.e., $\kg[t'][0]$ only contains open and off nodes, but no closed nodes.
\end{lemma}
\begin{proof}
	We prove the claim by induction over the operations $t$ with the properties that $\kg[t][0]$ contains no closed node, but $\kg[t+1][0]$ contains at least one closed node. The claim is true for the initial (empty) graph $\kg[0][0]$. Let $t > 0$. We prove that in at least one operation $t' \in \{ t+1, \ldots, 2t \}$, the number of closed nodes is $0$. For the sake of contradiction, assume that for all $t'$, $t < t' \leq 2t$, the number of closed nodes is non-zero, i.e., $\pasnum[t'][0] > 0$. We call an open node \emph{semi-open} if it has degree greater than $50f(k,i)$ after some update $i \in [2t]$. Otherwise, we call it \emph{fully-open}. Recall that, similarly, a vertex becomes closed if it has degree at least $100f(k,i)$ after some update $i \in [2t]$ (and never becomes open again).
	
	For any $i \in [2t]$, the algorithm's query budget increases by $f(k, n_i)$ queries after the $i$\xth/ update. Since $f(k,i)$ is non-decreasing, nodes inserted after update operation $i-1$ can only become semi-open if their degrees increase to at least $50f(k,i)$ (resp. $100f(k,i)$) by the definition of semi-open (resp. closed). Also due to the monotonicity of $f$, the number of semi-open or closed nodes is maximized if the algorithm invests its query budget as soon as possible. It follows that the number of semi-open or closed nodes up to operation $2t$ is at most $\sum_{j \in [2t]} 4f(k,j) / (50f(k,j)) \leq 4t/50$. Without loss of generality, we may assume that all semi-open nodes are closed (so the algorithm does not need to invest budget to make them closed).
	
	After update operation $2t$, the algorithm's total query budget from all update operations is at most $\sum_{i \in [2t]} f(k,n_i) \leq \sum_{i \in [2t]} f(k,i) \leq 2t f(k,2t)$. Answering a query $i$ inserts at most two edges into $\kg[t][i-1]$, and the sum of degrees increases by at most four. The algorithm may use its budget to increase the degree of at most $2t \cdot 4f(k,2t) / (50f(k,2t)) \leq 8t/50$ fully-open nodes to at least $100f(k,2t)$, i.e., to make them closed nodes. 
	Recall our assumption that $\pasnum[i][0] > 0$ for all $i \in \{t+1, \ldots, 2t\}$. The adversary deletes one point corresponding to a closed node in each update operation from $\{t+1, \ldots, 2t \}$. Therefore, the number of closed nodes after update operation $2t$ is
	\begin{equation*}
		\pasnum[2t][0] \leq \frac{4t}{50} + \frac{8t}{50} - t \leq \frac{12t}{50} - t < 0.
	\end{equation*}
	This is a contradiction to the assumption.
\end{proof}

\paragraph{Proofs of \cref{lem:metric-consistent,lem:muni-consistent,lem:muni-consistent,lem:mstar-consistent,lem:mstarrange-consistent}}

The following observation follows immediately from the properties of shortest path metrics.

\begin{observation}
	\label{lem:convex-compliance}
	Let $G=(V,E)$ and $G'=(V,E')$ be two graphs so that $E \subseteq E'$. If a sequence of queries is consistent with the shortest path metric on $G$ as well as on $G'$, then, for any $E''$, $E \subseteq E'' \subseteq E'$, it is also consistent with the shortest path metric on $(V, E'')$.
\end{observation}

The following lemma together with \cref{lem:convex-compliance} implies \cref{lem:metric-consistent,lem:muni-consistent,lem:muni-consistent,lem:mstar-consistent,lem:mstarrange-consistent,lem:mmulti-consistent} by setting $G = \kg[t][i-1]$ and $G' = (\V{\kg[t][i-1]}, \E{\kg[t][i-1]} \cup (\actno[t][i-1] \times \actno[t][i-1]))$ in \cref{lem:convex-compliance}.

\begin{lemma}
	\label{lem:sp-consistent}
	For any $t, t' > 0$, $i, i' \in [c]$ so that $(t',i') \prec (t,i)$, the answer given to query $\q[t'][i']$ is consistent with the shortest path metric on $G'$.
\end{lemma}
\begin{proof}
	Let $t' > 0$, $i' \in [c]$ so that $(t',i') \prec (t,i)$ and denote $(x,y) \defeq q \defeq \q[t'][i']$. We prove that $\ans{q} = \dsp[\kg[t][i]]{x}{y}$. As neither vertices nor edges are deleted after they have been inserted, for every $(t_1, i_1) \prec (t_2, i_2)$, $\kg[t_2][i_2]$ is a supergraph of $\kg[t_1][i_1]$. Thus, if there exists a path $P$ between $x$ and $y$ in $\kg[t'][i']$, a shortest path in $\kg[t][i]$ between $x$ and $y$ cannot be longer than $P$.
	
	It remains to prove $\ans{q} \leq \dsp[\kg[t][i]]{x}{y}$. For the sake of contradiction, assume that there exist $t'',i''$ 
	so that $\ans{q} = \dsp[\kg[t''][i''-1]]{x}{y}$, but $\ans{q} > \dsp[\kg[t''][i'']]{x}{y}$. By \cref{def:node_types}, closed nodes never become open. For any closed node $v \in \pasno[t''][i''-1]$, it follows that $\dsp[\kg[t''][i'']]{v}{\actno[t][i]} \geq \dsp[\kg[t''][i''-1]]{v}{\actno[t][i-1]}$ as \cref{alg:answer} only inserts edges between vertices in $\actno[t][i]$ into $\kg[t][i]$. Therefore, any shortest path between $x$ and $y$ in the graph $(\V{\kg[t''][i''-1]}, \E{\kg[t''][i''-1] \cup (\actno[t''][i''-1] \times \actno[t''][i''-1])}$ has length at least $\dsp[\kg[t''][i''-1]]{x}{y} = \ans{q}$.
\end{proof}

\subsubsection{Lower Bounds for Clustering}

Our lower bounds apply for the case that the algorithm is allowed to choose as centers any points that have ever been inserted as well as to the case where  centers must belong to the set of current points. In the whole section we use the notations  of $\kg[t'][i], \actno[t][i]$ etc. from the introduction of \cref{sec:appendix-lower-adaptive}.

\paragraph{Proof of \cref{thm:det-1lb}}

For any $\ell\in\N_{0}$, we define a metric $M_{\ell}(P^{*})$
on a subset of points $P^{*}$ as the shortest path metric on the
following graph. For each pair of open vertices $u,v$ we add an edge if $d(u,P^{*})\geq\ell$
and $d(v,P^{*})\geq\ell$.

\begin{lemma} \label{lem:mmulti-consistent} For each clean update operation
	$t$, each subset of points $P^{*}\in V_{t}$, and each $\ell\in\N_{0}$
	the metric $M_{\ell}(P^{*})$ is consistent.
\end{lemma}
\begin{proof}
	The metric $M_{\ell}(P^{*})$ is an augmented graph metric for $t$ and, thus, for a clean update operation $t$, it
	is consistent by Lemma~\ref{lem:metric-consistent}.
\end{proof}

\begin{lemma}[\cref{thm:det-1lb}, part 1]
	Let $k \ge 2$.
	Consider any dynamic algorithm for maintaining an approximate $k$-center solution of a dynamic point set  that (1) queries amortized $f(k,n)$ distances per operation, where $n$ is the number of current points,
	and (2) outputs at most $g(k) \in O(k)$ centers.
	For any $t\geq 2$ such that $t$ is a clean operation, the approximation factor of the algorithm's solution
	(with respect to an optimal $k$-center solution) against an
	adaptive adversary right after operation $t$ is at least
	$\Omega\left(\min\left\{ k,\frac{\log n}{k\log{f(k,2n)}}\right\} \right)$.
\end{lemma}
\begin{proof}
	Denote $G\defeq(V,E)\defeq\kg[t][0]$ and $A\defeq\actno[t][0]$.
	By \cref{lem:kg-actnum}, the number of open nodes in $G$ is
	at least $92t/100\geq t$, which implies that $|A| \ge 92t/100$. 
	Thus,  after operation $t$,  the number $n$ of current points is at least  $92t/100$.

	Without loss of generality, we assume that $G[A]$ has exactly one connected component: If this is not the case, let $C_1, \ldots, C_s$ be the connected components of $G[A]$ and observe that we may insert a path connecting the connected components by inserting
	into $G[A]$ edges  $(v_1,v_2), \ldots, (v_{s-1}, v_s)$ of length $1$, where $v_i \in C_i$ are arbitrary vertices. This increases the maximum degree of nodes in $G$ by at most $2$ and
	the shortest-path metric $M$ on the resulting graph that is constructed in this way is an augmented graph metric for $t$. As $t$ is clean, Lemma~\ref{lem:metric-consistent} shows that $M$ is consistent.
	
	Let $x \in V$ be any node. By \cref{lem:max-degree}, the number $n^{(i)}$ of nodes that have distance at most $i$ to $x$ is at most $\sum_{j \in [i]} (100f(1,t)+3)^j \leq (100f(k,t)+3)^{i+1}$. 
	Consider the largest $\ell$ such that $(100f(k,t)+3)^{\ell+1} < n$. It follows that there exists a node $z$ at distance $\ell+1$ to $x$. Furthermore $(100f(k,t)+3)^{\ell+3} \ge n^{(\ell+1)} \ge n$, which implies that
	$\ell+2 \ge  \log_{100f(k,t)+3} n.$
	
	Let $S = \{ s_1, \ldots, s_{g(k)} \}$ be the solution of the algorithm. For $i \in [\ell+1]$, let us define $V^{(i)}$ to be the set of vertices $v$ in $G[A]$ with $d_{G[A]}(x,v) = i$. By pigeon hole principle, there must exist a consecutive sequence $(i_1, \ldots, i_m)$ so that $m \geq \ell / (g(k)+1)$ and, for all $i \in \{ i_1, \ldots, i_m \}$, $S \cap V^{(i)} = \emptyset$. Let $k' =  \min\{3k-1, m/2\}$. Consider the metric $M_{i_1,i_{k'}}(x)$ and let $y_{i_1}, \ldots, y_{k'}$ be elements from the respective sets $V^{(i_1)}, \ldots,V^{(i_{k'})}$.
	
	The algorithm's solution $S$ has cost at least $k'$ because $\dsp[G]{S}{V_{i_{k'}}} \geq k'$. The solution $\{ y_{3j-2} \mid j \in \N \wedge 3j-2 \in [k'] \}$ is optimal and has cost $1$. It follows that the approximation factor of $S$ is greater than or equal to $k' = \min\{3k-1, m/2\}$.
	
	Let $n$ be the number of points at iteration $t$. We calculate that
	\begin{alignat*}{1}
		\min\left\{ 3k/2,m/2\right\}  & \ge\Omega\left(\min\left\{ k,m\right\} \right)\\
		& \ge\Omega\left(\min\left\{ k,\frac{\ell}{g(k)+1}\right\} \right)\\
		& \ge\Omega\left(\min\left\{ k,\frac{1}{g(k)}\left(\frac{\log n}{\log{(103f(k,t))}}-1\right)\right\} \right)\\
		& \ge\Omega\left(\min\left\{ k,\frac{\log n}{k\log{(103f(k,2n))}}\right\} \right)\\
		& \ge\Omega\left(\min\left\{ k,\frac{\log n}{k\log103+k\log{f(k,2n)}}\right\} \right)\\
		& \ge\Omega\left(\min\left\{ k,\frac{\log n}{k\log{f(k,2n)}}\right\} \right)
	\end{alignat*}
	using that $g(k)=O(k)$.
\end{proof}

\begin{lemma}[\cref{thm:det-1lb}, part 2]
	\label{lem:diam}
	Consider any dynamic algorithm for computing the diameter of a dynamic point set  that queries amortized $f(1,n)$ distances per operation, where $n$ is the number of current points,
	and outputs at most $g\geq1$ centers.
	For any $t\geq 2$ such that $t$ is a clean operation, the approximation factor of the algorithm's solution
	(with respect to the correct diameter) against an
	adaptive adversary right after operation $t$ is at least
	$\log (92t/100) / (\log{103f(1,t)}) -1$ and $92t/100 \le n \le t$.
\end{lemma}
\begin{proof}
	Denote $G\defeq(V,E)\defeq\kg[t][0]$ and $A\defeq\actno[t][0]$.
	By \cref{lem:kg-actnum}, the number of open nodes in $G$ is
	at least $92t/100\geq t$, which implies that $|A| \ge 92t/100$. 
	Thus,  after operation $t$,  $n \ge 92t/100$.

	Without loss of generality, we assume that $G[A]$ has exactly one connected component: If this is not the case, let $C_1, \ldots, C_s$ be the connected components of $G[A]$ and observe that we may insert a path connecting the connected components by inserting
	into $G[A]$ edges  $(v_1,v_2), \ldots, (v_{s-1}, v_s)$ of length $1$, where $v_i \in C_i$ are arbitrary vertices. This increases the maximum degree of nodes in $G$ by at most $2$ and
	the shortest-path metric $M$ on the resulting graph that is constructed in this way is an augmented graph metric for $t$. As $t$ is clean, Lemma~\ref{lem:metric-consistent} shows that $M$ is consistent.
	
	Let $x \in V$ be any node. By \cref{lem:max-degree}, the number $n^{(i)}$ of nodes that have distance $i$ to $x$ is at most $\sum_{j \in [i]} (100f(1,t)+3)^j \leq (100f(1,t)+3)^{i+1}$. 
	Consider the largest $i$ such that $(100f(1,t)+3)^{i+1} < n$. It follows that there exists a node at distance $i+1$ to $x$. Furthermore $(100f(1,t)+3)^{i+2} \ge n^{(i+1)} \ge n$, which implies that
	$i+2 \ge  \log_{100f(1,t)+3} n.$
	As $f(1,t) \ge 1$ for all values of $t$, there exists a shortest path $P$ starting at $x$ of length at least $i+1 \geq \log_{100f(1,t)+3}n - 1 \geq (\log n/ \log{(103f(1,t))}) -1 =: \ell$. It follows that the diameter is at least $\ell$.
	On the other hand, $M$ can be extended by adding an edge between any pair of open nodes, resulting in the
	consistent metric  $M_{\mathrm{uni}}$. For this metric  the diameter of $G$ is 1.
	As the algorithm cannot tell whether $\ell$ or 1 is the correct answer, and it always has to output a value that is as least as large as the correct answer, it will output at least $\ell$. Thus, the approximation ratio is at least $\ell \ge \log (92t/100)/ \log{103f(1,t)} -1$.
	Note that this implies a lower bound for the approximation ratio for $1$-sum-of-radii, and $1$-sum-of-diameter.
\end{proof}

\begin{lemma} \label{lem:det-low-1pclus} 
	Consider any dynamic algorithm for $(1,p)$-clustering that queries amortized $f(1,n)$ distances per operation, where $n$ is the number of current points,
	and outputs at most $1 \le g\leq n$ centers.
	For any $t\geq1$ such that $t$ is a clean operation, the approximation factor of the algorithm's solution
	(with respect to the optimal $(1,p)$-clustering cost) against an
	adaptive adversary right after operation $t$ is at least $\left[\frac{\log(t/4g)}{p + \log(101f(1,t))}\right]^{p}/4$ and $92t/100 \le n \le t$. \end{lemma} 
\begin{proof}
	Denote $G\defeq(V,E)\defeq\kg[t][0]$ and $A\defeq\actno[t][0]$.
	By \cref{lem:kg-actnum}, the number of open nodes in $G$ is
	at least $92t/100\geq t$, which implies that $|A| \ge 92t/100$. 
	Thus,  after operation $t$,  $n \ge 92t/100$.
	
	Let $C$ be the centers that are picked by the algorithm after operation $t$. By the assumption of the lemma, $|C| \le g$.
	Consider $M_{\ell}(C)$, where $\ell\defeq\log(t/4g)/(p+\log(101 f(1,t)))$.
	By \cref{lem:max-degree}, for any $s\in C$, the size of $\lvert\{x\mid x\in A\wedge d(x,s)<\ell\}\rvert$
	is at most $\sum_{i\in[\ell-1]}(101 f(1,t))^{i}<(101 f(1,t))^{\ell}$. 
	Let $V_{+}\defeq\{x\mid x\in A\wedge d(x,C)\geq\ell\}$.
	Since $\lvert C\rvert\leq g$, it holds that $\lvert V_{+}\rvert>|A| -g(101 f(1,t))^{\ell}  \ge 92t / 100 -g(t/4g)^{\log ((101 f(1,t)) / (p + \log (101 f(t,1)))}\ge 92t/100 - t/4 \ge t/2$.
	Since $d(s,V_{+})\geq\ell$ for any $s\in C$, the $(1,p)$-clustering
	cost of $S$, and thus the cost of the algorithm, is at least $\lvert V_{+}\rvert\cdot\ell^{p}\geq t/2\cdot\ell^{p}$.
	
	We will show that if instead a single point corresponding to a vertex of $V_{+}$ is picked as center, then the cost is at most $2t$, which provides an upper bound on the cost of the optimum solution.
	It follows that the approximation factor achieved by the algorithm is at least $\ell^p / 4$.

	To complete the proof consider a point $x$ whose corresponding point
	$v_{x}$ belongs to $V_{+}$. One can easily show that for
	each $\alpha\in\N$ it holds that $\alpha^{p}\le(101f(1,t)\cdot2^{p})^{2\alpha/3}$
	since for each $\alpha\in\N$ it holds that $\alpha^{p}=2^{p\log\alpha}\le2^{p\cdot\frac{2\alpha}{3}}$.
	Thus, the $(1,p)$-clustering cost if a single point $x$ is chosen
	as center is at most
	
	\begin{alignat*}{1}
		\sum_{i=0}^{\ell-1}g(101f(1,t))^{i}\cdot(\ell-i)^{p}+t\cdot1^{p} & \le g\left((101f(1,t))^{\ell}\cdot\sum_{i=1}^{\ell}\frac{1}{(101f(1,t))^{i}}\cdot i^{p}\right)+t\\
		& \le g\left((101f(1,t))^{\ell}\cdot\sum_{i=1}^{\ell}\frac{(101f(1,t)\cdot2^{p})^{2i/3}}{(101f(1,t))^{i}}\right)+t\\
		& \le g\left((101f(1,t))^{\ell}\cdot\sum_{i=1}^{\ell}\frac{2^{p\cdot2i/3}}{(101f(1,t))^{i/3}}\right)+t\\
		& \le g\left((101f(1,t))^{\ell}\cdot\sum_{i=1}^{\ell}\left(\frac{2^{\frac{2p}{3}}}{(101f(1,t))^{1/3}}\right)^{i}\right)+t\\
		& \le g(101f(1,t))^{\ell}2^{p\ell}+t\\
		& \le g(101f(1,t)\cdot2^{p})^{\ell}+t\\
		& \le5t/4\le2t
	\end{alignat*}
	using that $\ell=\frac{\log(t/4g)}{p+\log(101f(1,t))}=\frac{\log(t/4g)}{\log(2^{p}101f(1,t))}=\log_{2^{p}101f(1,t)}t/4g$.
	This yields an approximation ratio of at least $\frac{t/2\cdot\ell^{p}}{2t}=\Omega(\ell^{p})$. 
\end{proof}

\begin{lemma}[\cref{thm:det-1lb}, part 3]
	For any $k  \geq 1$, any dynamic algorithm which returns a set of $k$-centers against an adaptive adversary
	with an amortized update time of $f(k,n)$, for
	an arbitrary function $f$,  must have an approximation ratio of $\Omega\left(\left(\frac{\log(n)}{z+\log f(1,2n)}\right)^{z}\right)$ for the $(1,z)$-clustering.
\end{lemma}
\begin{proof}
	Let $t\in\N$. By Lemma~\ref{lem:mpt-det-low-adv-prop}, there is
	a value $t'$ with $t<t'\le2t'$ such that $t'$ is a clean operation.
	Let $n$ be the number of open points at iteration $t'$. By \cref{lem:kg-actnum}
	we know that $t'\ge n\ge92t'/100$. Note that hence $t'\le2n$.
	
	Recall that we assumed that the function $f(k,n)$ is non-decreasing
	in $n$ (for any fixed $k$). Suppose that after operation $t'$ we
	query the solution value of an algorithm for $1$-center, $1$-sum-of-radii,
	or $1$-sum-of-diameter. By Lemma~\ref{lem:diam} its approximation
	ratio is at least $\frac{\log(92t'/100)}{\log(103f(k,t'))} -1\ge\frac{\log(92n/100)}{\log(103f(k,2n))} -1=\Omega\left(\frac{\log(n)}{\log(f(k,2n))}\right)$.
	Suppose that instead we query the solution from the algorithm for
	$(1,p)$-clustering. By Lemma~\ref{lem:det-low-1pclus}, its approximation
	ratio is at least
	
	\begin{alignat*}{1}
		\left[\frac{\log(t'/4g)}{p+\log(101f(1,t'))}\right]^{p}/4 & \ge\left[\frac{\log(n/4g)}{p+\log(101f(1,2n))}\right]^{p}/4\\
		& \ge\left[\frac{\log(n)-\log(4g)}{p+\log(101f(1,2n))}\right]^{p}/4\\
		& \ge\left[\frac{\log(n)}{1.1p+1.1\log(101f(1,2n))}\right]^{p}/4\\ 
		& \ge\left[\frac{\log(n)}{1.1p+1.1(7+\log(f(1,2n))}\right]^{p}/4\\ 
		& \ge\left[\frac{\log(n)}{1.1p+8+1.1\log(f(1,2n))}\right]^{p}/4\\ 
		& \ge\left[\frac{\log(n)}{2p+8+2\log(f(1,2n))}\right]^{p}/4\\ 
		& =\Omega\left(\left(\frac{\log n}{2p+8+2\log f(1,2n)}\right)^{p}\right)
	\end{alignat*}
	using that $g=O(1)$. Hence, for $k$-median and $k$-means if we
	take $p=1$ and $p=2$, respectively, this yields bounds of $\Omega\left(\frac{\log n}{10+2\log f(1,2n)}\right)$
	and $\Omega\left(\left(\frac{\log n}{12+2\log f(1,2n)}\right)^{2}\right)$,
	respectively. 	
\end{proof}

\bibliography{cluster}
\end{document}